\documentclass{LMCS}

\def\dOi{11(3:10)2015}
\lmcsheading%
{\dOi}
{1--56}
{}
{}
{Oct.~20, 2014}
{Sep.~17, 2015}
{}

\ACMCCS{[{\bf Theory of computation}]: Formal languages and automata theory}

\usepackage[utf8]{inputenc}
\usepackage[colors=false,newdocumentcommand]{generic}

\begin{document}

\let\oldparagraph\paragraph
\renewcommand{\paragraph}[1]{\oldparagraph{\bf #1.~}}

\makeatletter
\newcommand*{\ifitalic}{\ifx\f@shape\my@test@it%
                          \expandafter\@firstoftwo
                        \else
                          \expandafter\@secondoftwo
                        \fi}
\newcommand*{\my@test@it}{it}
\makeatother

\newcommand{\MSO}{\texorpdfstring{\ifitalic{MSO$\,^\sim$}{MSO$^\sim\mspace{-1mu}$}}{MSO}\xspace}
\newcommand{\FO}{\texorpdfstring{\ifitalic{FO$\,^\sim$}{FO$^\sim\mspace{-1mu}$}}{FO}\xspace}

\newcommand{\dom}{{\mathsf{dom}}}
\renewcommand{\min}{{\mathsf{min}}}
\renewcommand{\max}{{\mathsf{max}}}
\newcommand{\mem}{{\mathsf{mem}}}
\newcommand{\meml}{{\mathsf{mem}_{\cL}}}
\newcommand{\memr}{{\mathsf{mem}_{\cR}}}

\newcommand{\len}[1]{{|#1|}}
\newcommand{\dlen}[1]{{||#1||}}
\newcommand{\ang}[1]{{\langle #1 \rangle}}
\newcommand{\intr}[1]{{\llbracket #1 \rrbracket}}

\newcommand{\blank}{\phantom{x}}
\newcommand{\orbit}[1]{{{#1}^{\mathsf{o}}}}
\newcommand{\orbiteq}{\stackrel{\mathsf{o}}{=}} 

\newcommand{\true}{{\mathsf{true}}}
\newcommand{\false}{{\mathsf{false}}}
\newcommand{\rigid}{{\mathsf{rigid}}}
\newcommand{\semirigid}{{\mathsf{semirigid}}}
\providecommand{\fact}{}
\renewcommand{\fact}{{\mathsf{fact}}}
\newcommand{\prej}{{\mathsf{pre}\text{-}J}}
\newcommand{\arity}{{\mathsf{arity}}}

\newcommand{\gR}{{\mathrel{\mathcal{R}}}}
\newcommand{\gL}{{\mathrel{\mathcal{L}}}} 
\newcommand{\gJ}{{\mathrel{\mathcal{J}}}} 
\newcommand{\gH}{{\mathrel{\mathcal{H}}}}
\newcommand{\gRO}{{\mathrel{\orbit{\mathcal{R}}}}}
\newcommand{\gLO}{{\mathrel{\orbit{\mathcal{L}}}}} 
\newcommand{\gJO}{{\mathrel{\orbit{\mathcal{J}}}}}
\newcommand{\gHO}{{\mathrel{\orbit{\mathcal{H}}}}}

\tikzstyle{every picture} += [remember picture]
\tikzstyle{textbox} = [shape=rectangle, inner sep=0.5pt]
\newcommand{\td}[2]{\raisebox{-1.8pt}{%
                    \tikz%
                    \node[textbox](#1){\phantom{$f$}%
                                       \hspace{-5pt}%
                                       \raisebox{-0.6pt}[0pt][0pt]{$#2$}%
                                       \hspace{-4pt}%
                                       \phantom{$f$}};%
                   }}
\newcommand{\tl}[3]{\begin{tikzpicture}[overlay]%
                      \path [gray] (#1) edge [#2] (#3);%
                    \end{tikzpicture}}
\newcommand{\te}[3]{\begin{tikzpicture}[overlay]%
                      \path [gray] (#1) edge [bend left=#2] (#3);%
                    \end{tikzpicture}}
\newcommand{\ta}[3]{\begin{tikzpicture}[overlay]%
                      \path [gray, <-, >=stealth] (#1) edge [bend left=#2] (#3);%
                    \end{tikzpicture}}




\setlength{\marginparwidth}{2.9cm}
\setlength{\marginparsep}{0.01cm}

\renewenvironment{itemize}[1][]{\smallskip%
                                \ifthenelse{\equal{#1}{}}{\begin{compactitem}}{\begin{compactitem}[#1]}%
                                \setlength{\itemsep}{\smallskipamount}}%
                               {\end{compactitem}%
                                \smallskip}
\renewenvironment{enumerate}[1][]{\smallskip%
                                  \ifthenelse{\equal{#1}{}}{\begin{compactenum}}{\begin{compactenum}[#1]}%
                                  \setlength{\itemsep}{\smallskipamount}}%
                                 {\end{compactenum}%
                                  \smallskip}
\renewenvironment{condlist}{\smallskip%
                            \begin{compactenum}[{(C}1{)}]%
                            \setlength{\itemsep}{\smallskipamount}}%
                           {\end{compactenum}%
                            \smallskip}
\renewenvironment{proplist}{\smallskip%
                            \begin{compactenum}[{(P}1{)}]%
                            \setlength{\itemsep}{\smallskipamount}}%
                           {\end{compactenum}%
                            \smallskip}
\renewenvironment{axiomlist}{\smallskip%
                            \begin{compactenum}[{(A}1{)}]%
                            \setlength{\itemsep}{\smallskipamount}}%
                           {\end{compactenum}%
                            \smallskip}

\setdefaultitem{\scalebox{0.75}{\textbullet}}%
               {\normalfont\bfseries \textendash}%
               {\textasteriskcentered}{\textperiodcentered}

\title[Logics with rigidly guarded data tests]{Logics with rigidly guarded data tests}

\author[T.~Colcombet]{Thomas Colcombet\rsuper a}
\address{{\lsuper a}CNRS / LIAFA}
\email{thomas.colcombet@liafa.univ-paris-diderot.fr}

\author[C.~Ley]{Clemens Ley\rsuper b}
\address{{\lsuper b}Independent researcher}
\email{ley.clemens@gmail.com}

\author[G.~Puppis]{Gabriele Puppis\rsuper c}
\address{{\lsuper c}CNRS / LaBRI}
\email{gabriele.puppis@labri.fr}

\titlecomment{{\lsuper*}An extended abstract with preliminary results similar
              to those presented here has appeared in~\cite{rigid-mso}.}

\sloppy
\maketitle

\begin{abstract}
The notion of orbit finite data monoid was recently introduced by Boja{\'n}czyk 
as an algebraic object for defining recognizable languages of data words. 
Following B\"uchi's approach, we introduce a variant of monadic second-order
logic with data equality tests that captures precisely the data languages 
recognizable by orbit finite data monoids.
We also establish, following this time the approach of Sch\"utzenberger, 
McNaughton and Papert, that the first-order fragment of this logic defines 
exactly the data languages recognizable by aperiodic orbit finite data monoids.
Finally, we consider another variant of the logic that can be interpreted
over generic structures with data. The data languages defined in this
variant are also recognized by unambiguous finite memory automata.	

\end{abstract}

\section{Introduction}\label{sec:introduction}

Data words have been introduced as a generalization of words over finite alphabets, 
where the term ``data'' denotes the presence of symbols from an infinite alphabet. 
Usually, languages of data words, data languages for short, are assumed to be closed 
under permutation of the data values. This invariance under permutation makes any 
property concerning the data values, other than equality, irrelevant. Some examples 
of data languages are:
\begin{itemize}
  \item[$L_1:$] ~\emph{the sets of words containing at least three distinct data values},
  \item[$L_2:$] ~\emph{the sets of words where the first and last positions carry the same data value},
  \item[$L_3:$] ~\emph{the sets of words with no consecutive occurrences of the same data value},
  \item[$L_4:$] ~\emph{the sets of words where each data value occurs at most once}. 
\end{itemize}
The intention behind data values in data words (or data trees, \dots) is 
to model, e.g., the keys in a database, or the process or user identifiers in the
log of a system. Those numbers are used as identifiers, and we are interested 
only in comparing them with equality. The invariance under permutation of data 
languages captures this intention. Data words can also be defined to have both 
a data value and a letter from a finite alphabet at each position. This is 
more natural in practice, and does not make any difference in the results 
to follow. 

The paper aims at understanding better how the classical theory of regular 
languages can be extended to data languages. The classical theory associates 
regular languages to finite state automata or, equivalently, to finite monoids. 
For instance, important properties of regular languages can be detected by 
exploiting equivalences with properties of the monoid -- see, for instance,
Straubing's book \cite{finite_automata_and_circuit_complexity} or Pin's survey 
\cite{mathematical_foundations_of_automata} for an overview of the approach.

In \cite{data_monoids} Boja{\'n}czyk formalized a notion of recognizability for
data languages by introducing generalizations of monoids, called \emph{data monoids}.
In the journal version of this paper \cite{nominal_monoids}, the algebraic framework 
of data monoids has been further generalized and connected to the theory of nominal 
sets, which was originally developed by Fraenkel in 1922. Here we are mainly interested 
in data languages recognized by \emph{orbit-finite data monoids}, which can be seen
as the analogue of finite monoids for languages over infinite alphabets. As a matter 
of fact, all regular languages over a finite alphabet can be seen as data languages 
recognized by orbit-finite data monoids. Other examples of data languages recognized 
by orbit-finite data monoids are the languages $L_1,L_2,L_3$ that we described above. 

Concerning the possibility of defining data languages by logical formulas,
a natural approach consists of extending classical logics by introducing a 
new predicate $x\sim y$, which holds at positions $x$ and $y$ whenever the 
data values under $x$ and $y$ are equal. In particular, one may think that 
the monadic second-order logic with this new predicate is a good candidate 
to equivalently specify recognizable languages, namely, it would play the 
role of monadic logic in the standard theory of regular languages. However, 
this is not the case, as monadic logic happens to be much too expressive. 
One inclusion indeed holds: every language of data words recognized by an 
orbit-finite monoid is definable in monadic logic extended with the data
equality predicate. However, the converse does not hold, as witnessed by 
the formula
\begin{align*}
  \forall x,y \quad x\neq y ~\rightarrow~ x\nsim y,
  \tag{$\dagger$}
\end{align*}
defining the language $L_4$ above, which is known not to be recognizable 
by orbit-finite data monoids. 
More generally, it has been shown that monadic logic (in fact, even 
first-order logic) extended with the data equality predicate has an 
undecidable satisfiability problem and it can express properties not 
implementable by reasonable automaton models \cite{machines_for_infinite_alphabets}. 

The general goal of this paper is to understand better the expressive 
power of the orbit-finite data monoid model by comparing it with 
automaton-based models and logical formalisms for data words.
In particular, we aim at answering the following question:

\begin{quote}
{\em Is there a variant of monadic second-order logic that defines precisely 
the data languages recognizable by orbit-finite data monoids?}
\end{quote}

\noindent
We answer this question positively by introducing a variant of monadic second-order 
logic with rigidly guarded data equality tests, \emph{rigidly guarded \MSO} for short. 
This logic allows testing equality of two data values only when the two positions are 
related in a bijective way (we say rigid). That is, data equality tests are allowed 
only in formulas of the form 
$$
  \varphi(x,y) ~\et~ x\sim y
$$
where $\varphi$ is rigid, namely, it defines a partial bijection. For example, one 
can express the existence of two consecutive positions sharing the same data value:
$\exists x,y.~ (x=y+1) ~\et~ x\sim y$. The guard $(x=y+1)$ is rigid since $x$ uniquely 
determines $y$, and $y$ uniquely determines $x$. 
However, it is impossible to describe the language 
$L_4$ in this logic. In particular, the above formula $(\dagger)$ is logically 
equivalent to $\neg\exists x,y.~ x\neq y ~\et~ x\sim y$, but this time the guard 
$x\neq y$ is not rigid: for a given $x$, there can be several $y$ such that $x\neq y$.
It may seem a priori that the fact that rigidity is a semantic property is a severe
drawback. This is not the case since (i) rigidity can be enforced syntactically 
(see Section \ref{sec:logics}), and (ii) rigidity is decidable for 
formulas in our logic (cf. Corollary \ref{cor:decidability}).

To validate the robustness of our approach, we also answer positively to the following 
question inspired by the seminal works of Schützenberger, McNaughton, and Papert:

\begin{quote}
{\em Does the rigidly guarded \FO logic (i.e., the first-order fragment of 
rigidly guarded \MSO) correspond to aperiodic orbit-finite data monoids?}
\end{quote}

\noindent
The idea underlying the use of guards with data tests can be generalized in different ways.
In the present paper, we also consider a less constrained version of rigidly guarded \MSO, 
which allows one to compare the data values at two positions $y$ and $z$, whenever both 
$y$ and $z$ are determined from a common position $x$ by means of suitable formulas.
The resulting logic, called semi-rigidly guarded \MSO, can be interpreted over more
general structures, such as graphs with data on nodes, and still retains the decidability
properties of rigidly guarded \MSO. Towards the end of the paper, we study the 
expressiveness of semi-rigidly guarded \MSO on data words and we prove that this 
logic is strictly subsumed by unambiguous finite memory automata \cite{finite_memory_automata}.


\paragraph{Related work} 
This work is related to the well known theory of regular languages.
By this we specifically refer to two key results, namely, the 
equivalence between recognizability by finite state automata and
definability in monadic logic \cite{weak_s1s}, and the characterization
of first-order definability for regular languages 
\cite{aperiodic_monoids,counter_free_automata}.

The other branch of related work is concerned with languages of data words.
The first related contribution in this direction is due to Kaminski and 
Francez \cite{finite_memory_automata,nondeterministic_reassignment}, 
who introduced finite memory automata (FMA for short). 
These automata possess a fixed finite set of registers that 
can be used to store data values. At each step an FMA can compare the current 
data value with the values stored in the registers and, on the basis
of these tests and the current control state, it can determine the target 
control state of its transition, and whether or not the current value 
(or a new guessed value) is stored into some register (replacing the previous content).
This model of automaton, in its non-deterministic form, has a decidable emptiness 
problem and an undecidable universality problem; decidability of universality
is however recovered in the deterministic variant of FMA.
Deterministic FMA also have minimal canonical forms, provided that a 
suitable policy in the use of registers is enforced \cite{minimal_finite_memory_automata,automata_with_group_actions}
(such a policy does not affect the expressive power of the model). 
Many other automaton models for data languages have been proposed in the 
literature, such as automata with pebbles \cite{machines_for_infinite_alphabets}, 
automata with hash tables \cite{regularity_for_data_languages}, walking 
automata \cite{walking_data_automata}, data and class automata 
\cite{data_automata_journal,class_automata_journal}. We refer the 
interested reader to~\cite{automata_for_xml} for a survey on these models.

As concerns the logical approach, several logics for reasoning effectively on data 
languages have been proposed, most notably: fragments of first-order logics with data 
equalities/disequalities \cite{data_automata_journal,two_variable_fo_with_two_orders}, 
variants of XPath called Core-Data-XPath \cite{two_variable_fo_on_trees_journal}, 
modal logics with registers \cite{freeze_ltl,freeze_mucalculus}. 
The differences between all such formalisms are reflected in the fact that it is 
difficult to obtain algebraic characterizations for robust classes of data languages. 
In \cite{automata_vs_logics,data_automata_journal,class_automata,fresh_register_automata} 
some preliminary results on relating automata to logics are given. However, the 
algebraic theory for these automaton models is not fully developed yet.
As a matter of fact, the question of characterizing the first-order logic definable 
language among the languages recognized by deterministic FMA remains open.

The idea of guarding tests with rigid formulas was originally presented in \cite{rigid-mso}.
A similar idea was also exploited in \cite{event-clock-automata} in order to design
a class of timed automata that could be determinized. 
More recently, a similar idea has been investigated in \cite{rigid-regular-path-queries}
with the aim of developing a robust formalism for querying graph databases.

\paragraph{Contributions and structure of the paper}
Our main contributions can be summarized as follows:
\begin{enumerate}
  \item We show how orbit-finite data monoids can be finitely represented 
        by systems of equations involving terms with variables for data values.
        We further develop the theory of Green's relations for data monoids,
        proving, for instance, that all $\cH$-classes in an orbit-finite 
        data monoid are finite (or, equally, that all orbit-finite data groups 
        are finite).
  \item We introduce a logic, called rigidly guarded \MSO, which can be seen 
        as a natural weakening of MSO logic with data equality tests. 
        We then show that rigidly guarded \MSO is exactly as expressive as 
        orbit-finite data monoids, and that its first-order fragment 
        corresponds to aperiodic orbit-finite data monoids.
  \item We show that an extension of rigidly guarded \MSO is decidable, even on
        general classes of structures with data (e.g., data trees).
        We show that the same extension of rigidly guarded \MSO defines 
        a proper subclass of data languages recognized by non-deterministic 
        (in fact, unambiguous) finite memory automata. 
\end{enumerate}
Section \ref{sec:background} gives some background knowledge on the theory of nominal sets,
data languages and data monoids. In particular, it explains how orbit-finite data monoids 
can be finitely represented and further develops the theory of Green's relations for these 
monoids.
Section \ref{sec:logics} introduces variants of rigidly guarded logics and shows how to 
decide satisfiability of their formulas over generic classes of data words, data trees, 
and data graphs.
Section \ref{sec:logic2monoids} describes the translation from rigidly guarded \MSO 
(resp., \FO) formulas to orbit-finite data monoids (resp., aperiodic orbit-finite data 
monoids) recognizing the same languages of data words. Section \ref{sec:monoids2logic} 
describes the converse translation, namely, from (aperiodic) orbit-finite data monoids 
to rigidly guarded \MSO (resp., \FO) formulas. 
Section \ref{sec:logics-for-automata} relates data languages defined by variants 
of rigidly guarded \MSO to data languages recognized by finite memory automata. 
Section \ref{sec:conclusion} provides an assessment of the results and related open problems.

\section{Nominal sets and data monoids}\label{sec:background}

In this paper, $D$ will usually denote an infinite set of \emph{data values} (e.g., $d,e,f,\ldots$)
and $A$ will denote a finite set of \emph{symbols} (e.g., $a,b,c,\ldots$). A \emph{data word} over 
the alphabet $D\times A$ is a finite sequence $w=(d_1,a_1) \ldots (d_n,a_n)$ in $(D \times A)^*$. 
The domain of $w$, denoted $\dom(w)$, is $\{1, \ldots, n\}$.

\medskip
We begin by giving a short account of the theory of nominal sets, which can then be used 
to derive natural notions of recognizability of data languages (we freely use some 
terminology and concepts from \cite{data_monoids,nominal_monoids,automata_with_group_actions}).

A \emph{(data) renaming} on $D$ is a permutation on the set $D$ of data values that is 
the identity on {\sl all but finitely many values}. We let $G_D$ the set of all renamings 
on $D$. One obtains a group $\cG_D=(G_D,\circ)$ by equipping $G_D$ with the operation of 
functional composition; we call this group the \emph{group of renamings on $D$}.
The above definitions are naturally generalized to any (possibly finite) subset $C$ of $D$; 
for example, we can talk about the group of renamings on $C$.

Renamings act on sets as follows. Given a set $S$, an \emph{action} of the group $\cG_D$ on $S$ 
is a group morphism $\hat\blank$ from $\cG_D$ to the group of bijections on $S$, namely, a 
function $\hat\blank$ that maps the identity $\iota$ of $\cG_D$ to the identity $\hat\iota$ 
on $S$ and such that $\widehat{\tau\circ\pi}=\hat\tau\circ\hat\pi$ for all renamings 
$\tau,\pi\in \cG_D$. We call \emph{$\cG_D$-set} any set $S$ equipped with an action
$\hat\blank$ of $\cG_D$ on $S$. 

Given an element $s$ of a $\cG_D$-set $(S,\hat\blank)$, we define the \emph{orbit of $s$} 
as the set of all elements of the form $\hat\tau(s)$, for all renamings $\tau\in\cG_D$. 
Note that orbits are either disjoint or equal, so they can be seen as equivalence classes 
induced by the possible renamings. We say that a $\cG_D$-set is \emph{orbit-finite} if it 
has only finitely many orbits. 

A subset $S'$ of a $\cG_D$-set $(S,\hat\blank)$ is said to be \emph{equivariant} 
if it is preserved by the action of renamings, namely, if $\hat\tau(S')=S'$ for all
renamings $\tau\in\cG_D$ (equivalently, one could say that $S'$ is a union of orbits of $S$).
The concept of equivariant subset can be applied specifically to a function $f:S\then T$
between two $\cG_D$-sets $(S,\hat\blank)$ and $(T,\check\blank)$; in this case one easily 
verifies that $f$ commutes with the renamings, namely, $f(\hat\tau(s)) = \check\tau(f(s))$ 
for all $f\in\cG_D$ and all $s\in S$. 
Similarly, by considering the standard action of renamings on sets of data words
(i.e., $\hat\tau((d_1,a_1) \ldots (d_n,a_n)) \eqdef (\tau(d_1),a_1) \ldots (\tau(d_n),a_n)$),
we define a \emph{data language} over $D\times A$ as an equivariant subset of $(D\times A)^*$
(this basically means that membership in the language is invariant under renamings of data values).

\subsection{Data monoids}\label{subsec:monoids}

Recall that a monoid is an algebraic structure $\cM=(M,\cdot)$ where $\cdot$ is an associative 
product on $M$ admitting an identity $1_\cM$ such that $1_\cM \cdot s = s\cdot 1_\cM = s$
for all $s\in M$. 
A monoid $\cM=(M,\cdot)$ is said to be \emph{aperiodic} if for all elements~$s\in M$, 
there is~$n\in\bbN$ such that~$s^n=s^{n+1}$.
A (\emph{monoid}) \emph{morphism} is a function $h$ between two monoids $\cM=(M,\cdot)$ 
and $\cN=(N,\odot)$ such that $h(1_\cM)=1_{\cN}$ and $h(s\cdot t)=h(s)\odot h(t)$ for 
all $s,t\in M$.
The concept of data monoid is nothing but that of a monoid with an equivariant product:

\begin{definition}\label{def:data-monoid}
A \emph{data monoid} (over a set $D$ of data values) is a triple $\cM=(M,\cdot,\hat\blank)$, 
where $(M,\cdot)$ is a monoid, $\hat\blank$ is an action of $\cG_D$ on $M$, and $\cdot$ is 
an equivariant function with respect to $\hat\blank$. In particular, for all renamings 
$\tau,\pi\in\cG_D$ and all elements $s,t\in M$, we have:
\begin{itemize}
  \item $\widehat{\tau\circ\pi}=\hat\tau\circ\hat\pi$, 
  \item $\hat\iota(s)=s$, where $\iota$ is the identity renaming, 
  \item $\hat\tau(1_\cM)=1_\cM$, where $1_\cM$ is the identity of $(M,\cdot)$,
  \item $\hat\tau(s)\cdot\hat\tau(t) = \hat\tau(s\cdot t)$.
\end{itemize}
\end{definition}

Unless otherwise stated, data monoids will be defined over the set $D$ 
of all data values. Moreover, to simplify the notation, we will often 
use an implicit notation for the group action $\hat\blank$; for example,
when $\hat\blank$ is understood from the context, we can write $\tau(s)$ 
in place of $\hat\tau(s)$.

The \emph{free data monoid} over $D\times A$ is an example of a data monoid, 
where the elements are the data words over $D\times A$, the product is the 
juxtaposition of data words, and the action is the standard one, mapping any 
renaming $\tau$ to the automorphism $\hat\tau$ defined by 
$\hat\tau\big((d_1,a_1) \ldots (d_n,a_n)\big) 
 = (\tau(d_1),a_1) \ldots (\tau(d_n),a_n)$.

\medskip
We now show how to extract the ``memory'' of a monoid element $s$, which intuitively is the 
minimum set of data values that are important for distinguishing $s$ from all other elements 
of the data monoid. 
Given a data monoid $\cM$ and an element $s$ in it, 
we say that a renaming $\tau$ is a \emph{stabilizer} of $s$ if $\tau(s)=s$. 
A set $C\subseteq D$ of data values \emph{supports} an element $s$ if all 
renamings that are the identity on $C$ are stabilizers of $s$. 
It is known that the intersection of two sets that support $s$ is 
again a set that supports $s$ 
\cite{data_monoids,nominal_monoids,abstract_syntax}. 
We can thus define \emph{the memory} of $s$, denoted $\mem(s)$, 
as the intersection of all sets that support $s$. 

We remark that there exist finite monoids whose elements have infinite memory 
(see~\cite{data_monoids} for an example). On the other hand, monoids that 
are homomorphic images of the free monoid contains only elements with finite 
memory. 
As we are mainly interested in homomorphic images of the free monoid, 
hereafter we will consider only monoids whose elements have finite memory 
-- this property is called the \emph{finite support axiom}.

\begin{definition}\label{def:memory}
Let $\cM$ be a data monoid. 
We define \emph{the memory} of an element $s$ in $\cM$ as
$$
  \mem(s) ~=~ 
  \bigcap\big\{\, C\subseteq D ~:~ \forall \tau\in G_D.~ (\forall d\in C.\: \tau(d)=d) 
                                                    ~\then~ \tau(s)=s \,\big\} \ .
$$
and we assume that this set is always finite.
A data value is said to be \emph{memorable} in $s$ if it belongs to $\mem(s)$.
\end{definition}

A \emph{morphism} between two data monoids $\cM=(M,\cdot,\hat\blank)$ and 
$\cN=(N,\odot,\check\blank)$ is a monoid morphism that is equivariant, namely, 
a function $h:M\then N$ such that 
\begin{itemize}
  \item $h(1_\cM)=1_{\cN}$, 
  \item $h(s\cdot t)=h(s)\odot h(t)$ for all $s,t\in M$, 
  \item $h(\hat\tau(s))=\check\tau(h(s))$ for all $s\in M$ and all renamings $\tau\in\cG_D$. 
\end{itemize}
A data language $L\subseteq (D\times A)^*$ is \emph{recognized} by a 
morphism $h:(D\times A)^*\then\cM$ if the membership of a word 
$w\in (D\times A)^*$ in $L$ is determined by the element $h(w)$ of $\cM$, 
namely, if $L=h^{-1}(h(L))$. 

We conclude the preliminary discussion on data monoids by recalling the 
definition of \emph{orbit-finite} $\cG_D$-set, that is, a $\cG_D$-set that 
admits only finitely many orbits $\{\tau(s) ~:~ \tau\in\cG_D\}$. 
This property can be naturally applied to the domain of a data monoid 
$\cM$, resulting in the concept of \emph{orbit-finite data monoid}. 
Below, we give an example of a data language that is recognized by an 
orbit-finite data monoid and an example of a data language that is recognized 
only by orbit-infinite data monoids.

\begin{example}\label{ex:orbit-finite-language}
Consider the language $L_2 = \{d_1\ldots d_n\in D^* ~:~ n\ge 1,~d_1=d_n\}$ 
introduced at the beginning of Section \ref{sec:introduction}. 
One can construct the syntactic data monoid recognizing $L_2$ 
by considering the classes of the two-sided Myhill-Nerode equivalence on data words. 
More precisely, the class of a non-empty word $w = d_1\ldots d_n$ can be 
identified with the pair $(d_1,d_n)$ of data values, while the class of the empty word
is a distinguished element behaving as the identity. Accordingly, the product of two 
elements $(d,e)$ and $(f,g)$, distinct from the identity, is the pair $(d,g)$. 
This syntactic data monoid admits only three orbits: the singleton orbit 
containing the identity element, the orbit $\{(d,d) ~:~ d\in D\}$, 
and the orbit $\{(d,e) ~:~ d\neq e \in D\}$.
\end{example}

\begin{example}\label{ex:orbit-infinite-language}
Consider the language $L_4 = \{d_1\ldots d_n\in D^* ~:~ \forall i\neq j\le n.~ d_i\neq d_j\}$.
The element of the syntactic monoid of $L_4$ that corresponds to a word $w \nin L_4$ behaves 
as a null element $0$: the product of $0$ with any other element of the syntactic monoid 
gives again $0$. On the other hand, the element that correspond to a word
$w = d_1\ldots d_n \in L_4$ can be identified with the set $\{d_1,\ldots,d_n\}$ of data 
values. Accordingly, the product of the syntactic monoid maps any two disjoint sets of 
data values to their union, and any two intersecting sets of data values to the null element $0$. 
It is easy to see that this syntactic monoid has infinitely many orbits.
\end{example}

\subsection{Finite presentations of data monoids}\label{subsec:presentations}

Orbit-finite data monoids are infinite objects that need to be represented in a finite
way in order to be used in algorithms. Here we propose to represent these objects by 
means of systems of equations involving terms. The starting point consists of looking 
at restrictions of data monoids to finite sets of data values:

\begin{definition}\label{def:restriction}
Given a data monoid $\cM=(M,\cdot,\hat\blank)$ and a (finite or infinite) set 
$C\subseteq D$, we define the \emph{restriction of $\cM$ to $C$} as the data 
monoid $\cM|_C=(M|_C,\,\cdot\,|_C,\hat\blank|_C)$, where $M|_C$ consists of all 
elements $s\in M$ such that $\mem(s)\subseteq C$, $\cdot\,|_C$ is the restriction 
of $\cdot$ to $M|_C$, and $\hat\blank|_C$ is the restriction of $\hat\blank$ to 
$\cG_C$ and $M|_C$.
\end{definition}

Despite the fact that the restriction of a data monoid to a finite set $C$ is 
still a data monoid, one has to keep in mind that data monoids over finite sets 
do not satisfy the same properties as those over infinite sets. For instance, 
the Memory Theorem from \cite{data_monoids} does not hold for data monoids over 
finite sets. However, most of the properties that we outline hereafter hold 
independently of whether data monoids are defined over finite or infinite sets 
of data values.

We observe that if $s$ and $t$ are elements in the same orbit of a data monoid, 
then their memories have the same cardinality. This allows us to denote 
by $\dlen{\cM}$ the maximum cardinality of the memories of the elements of 
an orbit-finite data monoid $\cM$. 
The following proposition shows that the restriction of an orbit-finite data monoid 
$\cM$ over a {\sl sufficiently large} finite set $C$ uniquely determines $\cM$. 
A more careful analysis shows that a number of natural operations on orbit-finite 
data monoids can be performed at the level of the finite restriction.  
Some noticeable examples of such operations are the disjoint union and the product
of two orbit-finite data monoids and the quotient of an orbit-finite data monoid 
with respect to a congruence. Thus, restrictions of orbit-finite data monoids
provide a convenient way to effectively manipulate orbit-finite data monoids.

\begin{proposition}\label{prop:restriction}
Let $\cM$, $\cN$ be orbit-finite data monoids such that $\dlen{\cM}=\dlen{\cN}$ 
and let $C\subseteq D$ be a set of cardinality at least $2\dlen{\cM}$. 
If $\cM|_C$ and $\cN|_C$ are isomorphic, then so are $\cM$ and $\cN$.
\end{proposition}

\begin{proof}
Let $\cM=(M,\cdot,\hat\blank)$ and $\cN=(N,\odot,\check\blank)$ and let $f_C$ be a data monoid
isomorphism from $\cM|_C$ to $\cN|_C$. We show how extend $f_C$ to an isomorphism from $\cM$ to 
$\cN$. Given $s\in M$, we let $\tau$ be any renaming such that $\tau(\mem(s)) \subseteq C$
(such a renaming exists since $\len{\mem(s)}\le\len{C}$); we then observe that the element
$\hat\tau(s)$ belongs to the data monoid $M|_C$ and we accordingly define
$$
  f(s) ~\eqdef~ \hat\tau^{-1}(f_C(\hat\tau(s))) \ .
$$
We prove that the function $f$ is well defined, namely, that $f(s)$ does not depend on 
the choice of the renaming $\tau$. To do so, we consider two renamings $\tau$ and $\pi$ 
such that $\tau(\mem(s))\subseteq C$ and $\pi(\mem(s))\subseteq C$, we define 
$t=\hat\tau^{-1}(f_C(\hat\tau(s)))$ and $t'=\hat\pi^{-1}(f_C(\hat\pi(s)))$, and we prove 
that $t=t'$. Let $\theta = \pi \circ \tau^{-1}$. Since $\pi = \theta \circ \tau$, we have
$$
  t'~=~\check\pi^{-1}(f_C(\hat\pi(s)))~=~\check\pi^{-1}(f_C(\hat\theta(\hat\tau(s)))) \ .
$$
Since $\theta$ is a renaming over $C$ and $f_C$ is a morphism between data monoids over $C$, 
we have $f_C \circ \hat\theta = \check\theta \circ f_C$ and hence 
$$
  \check\pi^{-1}(f_C(\hat\theta(\hat\tau(s)))) ~=~ \check\pi^{-1}(\check\theta(f_C(\hat\tau(s))) \ .
$$
Moreover, since $\tau^{-1}= \pi^{-1} \circ \theta$, we get
$$
  \check\pi^{-1}(\check\theta(f_C(\hat\tau(s))) ~=~ \check\tau^{-1}(f_C(\hat\tau(s)))
                                                ~=~ t \ .
$$
This proves that the function $f$ is well defined.


Next, we claim that $f$ is a bijection from $M$ to $N$. Surjectivity is straightforward, since 
for every element $t\in N$, there exists a renaming $\tau$ such that $\tau(\mem(t))\subseteq C$,
and hence, if we let $s = \hat\tau^{-1}(f_C^{-1}(\hat\tau(t)))$, we have $f(s)=t$. 
The proof that $f$ is injective is analogous to the proof that $f$ is well defined, 
and thus omitted. It remains to prove that $f$ is a data monoid isomorphism.

\smallskip\noindent
{\em Commutativity with renamings.}
We claim that $f$ commutes with the action of renamings. Given an element $s\in M$ and
a renaming $\pi\in\cG_D$, we choose a renaming $\tau$ such that $\tau(\mem(s))\subseteq C$
and $\tau(\mem(\hat\pi(s))) \subseteq C$ hold (note that such a renaming exists since 
$\len{\mem(s)\cup\mem(\hat\pi(s))} \le \len{C}$). In particular, both elements $\hat\tau(s)$ 
and $\hat\tau(\hat\pi(s))$ belong to the data monoid $\cM|_C$.
We also define the renaming $\theta=\tau\circ\pi\circ\tau^{-1}$. Note that, by construction,
we have $\hat\theta(\hat\tau(s))=\hat\tau(\hat\pi(s))$. Moreover, by exploiting the definition 
of $f$ and the fact that $f_C$ is a data monoid morphism from $\cM|_C$ to $\cN|_C$, we obtain
$$
\begin{array}{rclclcl}
  f(\hat\pi(s)) &=& \check\tau^{-1}(f_C(\hat\tau(\hat\pi(s))))      
                &=& \check\tau^{-1}(f_C(\hat\theta(\hat\tau(s))))   \\[1ex]
                &=& \check\tau^{-1}(\check\theta(f_C(\hat\tau(s)))) 
                &=& \check\pi(\check\tau^{-1}(f_C(\hat\tau(s))))    
                &=& \check\pi(f(s)) \ .
\end{array}
$$

\smallskip\noindent
{\em Commutativity with products.}
We conclude the proof by showing that $f$ preserves identities and commutes with products. 
Recall that $M|_C$ (resp., $N|_C$) contains the identity $1_\cM$ of $\cM$ (resp., the identity 
$1_\cN$ of $\cN$). Since $f_C$ is a monoid morphism from $\cM|_C$ to $\cN|_C$, it follows 
that $f(1_\cM)=f_C(1_\cM)=1_\cN$. Let us now consider two elements $s,t\in M$. Let $\tau$ 
be a renaming such that $\tau(\mem(s))\subseteq C$ and $\tau(\mem(t))\subseteq C$
(again, such a renaming exists since $\len{\mem(s)\cup\mem(t)}\le\len{C}$). In particular,
both elements $\hat\tau(s)$ and $\hat\tau(t)$ belong to $\cM|_C$. Since $f_C$ is a monoid 
morphism, we obtain
$$
\begin{array}{rcl}
  f(s\cdot t) &=& \check\tau^{-1}\big(f_C(\hat\tau(s\cdot t))\big) 
              ~=~ \check\tau^{-1}\big(f_C(\hat\tau(s)\cdot\hat\tau(t))\big)        \\[1ex]
              &=& \check\tau^{-1}\big(f_C(\hat\tau(s))\odot f_C(\hat\tau(t))\big)                 
              ~=~ \check\tau^{-1}\big(f_C(\hat\tau(s))\big) \:\odot\: 
                  \check\tau^{-1}\big(f_C(\hat\tau(t))\big)                        \\[1ex]
              &=& f(s)\odot f(t) \ .
\end{array}
$$
We have just shown that $\cM$ and $\cN$ are isomorphic data monoids.
\end{proof}

\smallskip
Proposition \ref{prop:restriction} shows that, assuming orbit-finiteness, one 
can represent an infinite data monoid by a finite restriction of it. It is also 
possible to give more explicit representations of orbit-finite data monoids using 
what we call \emph{term-based presentation systems}. According to such systems, 
elements are represented by terms of the form $o(d_1, \ldots, d_k)$, where $o$ is 
an \emph{orbit name}, with an associated arity~$k$, and $d_1,\dots,d_k$ are 
distinct data values. Terms are furthermore considered modulo an equivalence 
relation $\dsim$ and equipped with a binary product operation $\odot$. 
Before entering the details of term-based presentation systems, 
we explain the general idea by means of an example.

\begin{example}\label{ex:presentation}
Let $L_1 = \{d_1\ldots d_n\in D^* \::\: 
             \exists i,j,k\le n.~ d_i\neq d_j,\: d_j\neq d_k,\, d_i\neq d_k\}$
be the language of data words with at least three distinct values. 
The elements of the syntactic data monoid of $L_1$ can be conveniently 
represented by terms, as follows: the empty word is represented by the 
term $o(\emptystr)$ of arity $0$; the equivalence class of a constant 
data word $d\ldots d$ is represented by the term $p(d)$ or arity $1$; 
the equivalence class of a data word containing exactly two distinct 
data values $d,e$ is represented by the term $q(d,e)$ or, equivalently, 
by the term $q(e,d)$; the equivalence class for all remaining words is 
represented by another term $r(\emptystr)$ of arity $0$. Accordingly, 
the syntactic data monoid of $L_1$ is represented by the following 
system of equations, where $d,e,f,g$ denote pairwise distinct data values 
and $t$ denotes a generic term built up from the orbit names $o,p,q,r$:
$$
\begin{array}{l}
\begin{array}{r@{~}c@{~}ccc@{~}c@{~}ccl}
  o(\emptystr) &\odot& t &\:\dsim&\: t &\odot& o(\emptystr) &\dsim& t            \\[1ex]
  r(\emptystr) &\odot& t &\:\dsim&\: t &\odot& r(\emptystr) &\dsim& r(\emptystr)
\end{array} 
\\[3.5ex]
\begin{array}{r@{~}c@{~}ccl}
  p(d) &\odot& p(d)      &\dsim& p(d)                    \\[1ex]
  p(d) &\odot& p(e)      &\dsim& q(d,e)
\end{array}
\\[3.5ex]
\begin{array}{r@{~}c@{~}l}
  q(d,e) &~\dsim&~ q(e,d)
\end{array}
\end{array}
\quad
\begin{array}{l}
\begin{array}{r@{~}c@{~}ccl}
  q(d,e) &\odot& p(d)    &\dsim& p(d) \,\odot\, q(d,e) ~\:\dsim\:~ q(d,e) \\[1ex]
  q(d,e) &\odot& q(d,e)  &\dsim& q(d,e)
\end{array}
\\[4ex]
\begin{array}{r@{~}c@{~}ccl}
  q(d,e) &\odot& p(f)    &\dsim& p(f) \,\odot\, q(d,e) ~\:\dsim\:~ r(\emptystr) \\[1ex]
  q(d,e) &\odot& q(d,f)  &\dsim& r(\emptystr) \\[1ex]
  q(d,e) &\odot& q(f,g)  &\dsim& r(\emptystr) \ .
\end{array}
\end{array}
$$
\end{example}

Hereafter, we will focus on those term-based presentation systems that correctly
represent data monoids, namely, whose binary operation $\odot$ is associative 
over the equivalence classes. Clearly, if the orbit names of the term-based 
presentation system range over a finite set, then the represented data monoid 
is orbit-finite. We will see below that a converse result also holds, showing 
that every orbit-finite data monoid can be represented by a term-based 
presentation system that uses only finitely many orbit names. This allows 
us to represent orbit-finite data monoids by finite systems of equations 
involving terms and products between them. 

\medskip
We now give a formal definition of our term-based presentation system. We denote by 
$T_{O,C}$ the set of all \emph{terms} of the form $o(d_1,\ldots,d_k)$, where $o$ is 
an orbit name from a finite set $O$, $k$ is the arity of $o$, and $d_1,\ldots,d_k$ 
are pairwise distinct data values from a (finite or infinite) subset $C$ of $D$.

\begin{definition}\label{def:term-rep-system}
Let $O$ be a finite set of orbit names and let $C$ be a (finite or infinite) 
set of data values. A \emph{term-based presentation system} $\cS$ over $(O,C)$
consists of a set of terms $T=T_{O,C}$, a binary operation $\odot$ on $T$, an 
action $\check\blank$ defined by 
$\check\tau(o(d_1,\ldots,d_n)) = o(\tau(d_1),\ldots,\tau(d_n))$, 
and an equivalence $\dsim$ on $T$ satisfying the following properties 
for all terms $s,t,u,v\in T$ and all renamings $\tau\in\cG_C$:
\begin{itemize}
  \item {\em (identity)} 
        there is a term $1_T$ of arity $0$ such that $1_T \odot s = s \odot 1_T = s$,
  \item {\em (equivariance)} $\check\tau(s) \odot \check\tau(t) \:=\: \check\tau(s \odot t)$,
  \item {\em (associativity up to $\dsim$)} $(s \odot t) \odot u \:\dsim\: s \odot (t \odot u)$,
  \item {\em (congruence for products)} 
        if $s \dsim t$ and $u \dsim v$, then $s\odot u \:\dsim\: t\odot v$,
  \item {\em (congruence for renamings)} if $s \dsim t$ then $\check\tau(s) \dsim \check\tau(t)$.
\end{itemize}
\end{definition}

Let $\cS=(T,\odot,\check\blank,\dsim)$ be a term-based presentation system. We remark 
that $(T,\odot,\check\blank)$ is not necessarily a data monoid because associativity 
only holds up to congruence $\dsim$. We say that $\cS$ \emph{represents} the 
structure $\cM=(M,\cdot,\hat\blank)$ if 
\begin{itemize}
  \item $M$ is the set of $\dsim$-equivalence classes of terms in $T$, 
  \item $\cdot$ is the binary operation on $M$ defined by 
        $[s]_\dsim \cdot [t]_\dsim = [s \odot t]_\dsim$,
  \item $\hat\blank$ maps any renaming $\tau\in\cG_C$ to the function $\hat\tau$
        defined by $\hat\tau([s]_\dsim) = [\check\tau(s)]_\dsim$
\end{itemize}
(it is easy to check that both $\cdot$ and $\hat\blank$ are well defined). 

\begin{proposition}\label{prop:presentation}
Every term-based presentation system represents an orbit-finite data monoid. 
Conversely, every orbit-finite data monoid is represented by a term-based 
presentation system.
\end{proposition}

\begin{proof}
We prove the first claim. Let $\cS=(T,\odot,\check\blank,\dsim)$ be a term-based presentation 
system over a set $C$ of data values and let $\cM=(M,\cdot,\hat\blank)$ be the structure represented 
by $\cS$. It is easy to see that $\cdot$ is an associative operation and that $[1_T]_\dsim$ behaves
as an identity for $\cdot$. This means that $(M,\cdot)$ is a monoid. 
Below, we verify the other properties of orbit-finite data monoids:
\begin{enumerate}       
  \item Since the identity term $1_T$ has arity $0$, we have that all renamings $\tau\in\cG_C$ 
        are stabilizers of $[1_T]_\dsim$, that is,
        $\hat\tau([1_T]_\dsim) = [\check\tau(1_T)]_\dsim = [1_T]_\dsim$.
  \item We now check that $\widehat{\tau \circ \pi} = \hat\tau \circ \hat\pi$ 
        for all data renamings $\tau, \pi \in\cG_C$. Let $[o(\bar{d})]$ be an element of $M$ 
        (we will drop the subscript $\dsim$ in the rest of the proof). Then we have
        $$
          \qquad\quad
          \widehat{\tau \circ \pi}[o(\bar{d})] ~=~ [o((\tau \circ \pi)(\bar{d}))] 
                                               ~=~ [o(\tau (\pi(\bar{d})))] 
                                               ~=~ \hat\tau([o(\pi(\bar{d}))]) 
                                               ~=~ \hat\tau \circ \hat\pi([o(\bar{d})]) \ .
        $$
  \item If $\iota$ is the identity renaming on $C$, then 
        $\hat\iota([o(\bar{d})]) = [o(\iota(\bar{d}))] = [o(\bar{d})]$. 
  \item Let $[s], [t] \in M$ and let $\tau\in\cG_C$ be a data renaming. 
        We prove that $\hat\tau([s] \cdot [t]) = \hat\tau([s]) \cdot \hat\tau([t])$. 
        Assume that $s=o(\bar{d})$, $t=p(\bar{e})$, and $s \odot t = q(\bar{f})$. 
        Then
        $$
          \qquad\quad
          \hat\tau([s] \cdot [t]) ~=~ \hat\tau([s \odot t]) 
                                  ~=~ \hat\tau([q(\bar{f})]) 
                                  ~=~ [q(\tau(\bar{f}))] 
                                  ~=~ [\check\tau(q(\bar{f}))] 
                                  ~=~ [\check\tau(s \odot t)] \ .
        $$
        Moreover, from the equivariance property of Definition \ref{def:term-rep-system},
        we know that
        $$
          \check\tau(s \odot t) ~\dsim~ \check\tau(s) \odot \check\tau(t).
        $$
        We continue our calculation as follows
        $$
        \qquad\quad
        \begin{array}{rcl}
          \big[\check\tau(s \odot t)\big] &=& \big[\check\tau(s) \odot \check\tau(t)\big] 
                                          ~=~ \big[\check\tau(o(\bar{d})) \odot 
                                                   \check\tau(p(\bar{e}))\big] 
                                          ~=~ \big[o(\tau(\bar{d})) \odot 
                                                   p(\tau(\bar{e}))\big]          \\[1ex]
                                          &=& \big[o(\tau(\bar{d}))\big] \cdot 
                                              \big[p(\tau(\bar{e}))\big] 
                                          ~=~ \hat\tau([s]) \cdot \hat\tau([t]) \ .
        \end{array}
        $$
        Combining the two equations, we get 
        $\hat\tau([s] \cdot [t]) = \hat\tau([s]) \cdot \hat\tau([t])$.
  \item We can finally claim that $\cM=(M,\cdot,\hat\blank)$ 
        is an orbit-finite data monoids: this follows immediately from the
        fact that the set $O$ of orbit names is finite.
\end{enumerate} 

\medskip\noindent
The proof of the second part of the proposition is more tedious, but not really difficult.
Let $\cM=(M,\cdot,\hat\blank)$ be an orbit-finite data monoid over $C$. If $C$ is 
infinite, then we assume without loss of generality that the data values in $C$ are 
the positive natural numbers. If $C$ is finite, then we assume that $C$ is a prefix 
of the natural numbers. We first define a term-based representation system 
$\cS=(T,\odot,\check\blank,\dsim)$ and later we show that $\cS$ represents $\cM$.

\smallskip\noindent
{\em Definition of $\cS$. }
Let $O$ be a finite set of orbit names that contains exactly one orbit name $o$ for each 
orbit of $\cM$. The arity of $o$ is the size of the memories of the elements of $o$ 
(recall that memories of elements from the same orbit have the same cardinality). 
Define $T$ to be the set of all terms that are build up from orbit symbols in $O$ and data 
values in $C$. Recall that, in a term-based representation system, the action $\check\blank$
of the renamings is naturally defined as follows: 
$\check\tau(o(\bar{d})) = o(\tau(\bar{d}))$ for all data renamings $\tau\in\cG_C$.

Below we define the operation $\odot$ on $T$. Since each element of $\cM$ can be represented
by several terms in $T$, we need to commit to a specific mapping of elements in $\cM$ to 
terms in $T$. The rough idea is as follows. We begin by fixing some representatives of the 
orbits of $\cM$ and an isomorphism between these representatives and some canonical terms 
in $T$. To compute the product of two terms $s,t\in T$, we first apply a renaming so as to 
map them to the canonical terms $\tilde{s}$ and $\tilde{t}$; then, we exploit the isomorphism 
between the canonical terms and the representatives of the orbits of $\cM$ to compute the 
product of $\tilde{s}$ and $\tilde{t}$ inside $\cM$; finally, we apply the inverse isomorphism
and renaming to obtain the desired product $s\odot t$.

More precisely, we fix a representative $m_o$ inside each orbit $o$ of $\cM$ 
in such a way that $\mem(m_o)$ is a prefix of the natural numbers, namely, 
$\mem(m_o)=\{1,\ldots,\mathsf{arity}(o)\}$. 
We associate with each sequence of data values $\bar{d}=d_1,\ldots,d_k$ 
a renaming $\sigma_{\bar{d}}$ that maps the numbers $1,\ldots,k$ to the 
values $d_1,\ldots,d_k$, and vice versa, and that is the identity on  
$D\setminus\{1,\ldots,k,d_1,\ldots,d_k\}$. 
We then define the function $f$ from $T$ to $M$ such that, for every term $o(\bar{d})$,
$$
  f(o(\bar{d})) ~\eqdef~ \hat\sigma_{\bar{d}}(m_o).
$$
Note that $f$ is not injective in general. 
This allows us to define the equivalence $\dsim$ over terms by $s \dsim t$ iff $f(s)=f(t)$. 

For each element $m\in M$, we need to choose in a canonical way a term $g(m)$ that belongs 
to the set $f^{-1}(m)$. This can be accomplished by letting $g(m)$ be the term in $f^{-1}(m)$ 
with the minimal tuple of data values according to the lexicographical order. 
In a similar way, we can associate with each pair $(s,t)$ of terms in $T$ a 
\emph{canonical renaming} $\sigma_{s,t}$ as follows. First, we say that a pair 
$(s',t')$ of terms is \emph{minimal} if $s'$ is of the form $o(1,\ldots,k)$, $t'$ 
is of the form $p(d_1,\ldots,d_h)$, and, for all $1\le i<j\le h$, $d_i,d_j\nin \{1,\ldots,k\}$ 
implies $d_i<d_j$. Then, we define the canonical renaming $\sigma_{s,t}$ as the unique renaming 
$\sigma$ such that $(\sigma(s),\sigma(t))$ is a minimal pair $(\sigma(s),\sigma(t))$.

We can now define the product $\odot$ of two terms $s,t\in T$ as follows:
$$
  s\odot t ~\eqdef~ \check\sigma^{-1}_{s,t} \circ g\,
                    \Big(\big(f \circ \hat\sigma_{s,t}\,(s)\big) \:\cdot\:
                         \big(f \circ \hat\sigma_{s,t}\,(t)\big)\Big).
$$
Note that the term $s\odot t$ belongs to the set $f^{-1}(f(s) \cdot f(t))$.
Accordingly, we define the identity term $1_T$ to be $g(1_{\cM})$, where 
$1_{\cM}$ is the identity element of $\cM$. This completes the definition
of $\cS=(T,\odot,\check\blank,\dsim)$.

\smallskip\noindent
{\em $\cS$ is a term-based presentation system.}
Before we prove that $\cS$ satisfies the conditions of Definition \ref{def:term-rep-system}, we establish the following claim.

\begin{claim}
Let $s,t \in T$ and $\tau\in\cG_C$. Then
\begin{condlist}
  \item $f(\check\tau(s)) ~=~ \hat\tau(f(s))$,
        \label{presentation-first-cond}
        \label{presentation-cond1}
  \item $f(s \odot t) ~=~ f(s) \cdot f(t)$
        \label{presentation-cond2}
  \item $\tau \circ \sigma_{s,t}\,(d) ~=~ \sigma_{\check\tau(s),\check\tau(t)}\,(d)$ 
        for all $d \in \mem(s) \cup \mem(t)$,
        \label{presentation-cond3}
  \item $\sigma_{\check\tau(s),\check\tau(t)} \circ \tau \,(d) ~=~ \sigma_{s,t}\,(d)$ 
        for all $d \in \mem(s) \cup \mem(t)$.
        \label{presentation-cond4}
        \label{presentation-last-cond}
\end{condlist}               
\end{claim}     

\begin{proof}[Proof of claim]
We first prove Condition \refcond{presentation-cond1}.
Recall that if $\bar{d}=d_1,\ldots,d_k$ is a tuple of data values, 
then $\sigma_{\bar{d}}$ is the data renaming that maps the numbers 
$1,\ldots,k$ to the values $d_1,\ldots,d_k$, and vice versa, and that 
is the identity on $D\setminus\{1,\ldots,k,d_1,\ldots,d_k\}$. 
For all renamings $\tau$ and numbers $i\le k$, we have that
\begin{align*}
  \tau \circ \sigma_{\bar{d}}\,(i) ~=~ \tau(d_i) ~=~ \sigma_{\tau(\bar{d})}\,(i) \ . 
  \tag{$\star$}
\end{align*}
For every term $s=o(\bar{d})$, with $\bar{d}=d_1,\ldots,d_k$, we verify that
\begin{align*}
  f(\check\tau(o(\bar{d}))) 
  &~=~ f(o(\tau(\bar{d})))                        \tag{by definition of $\check\blank$} \\[1ex]
  &~=~ \hat\sigma_{\tau(\bar{d})}(m_o)            \tag{by definition of $f$} \\[1ex]
  &~=~ \widehat{\tau \circ \sigma}_{\bar{d}}(m_o) \tag{by $\star$ and 
                                                       by $\mem(m_o) \subseteq \{1, \ldots, k\}$} \\[1ex]
  &~=~ \hat\tau \circ \hat\sigma_{\bar{d}}~(m_o)  \tag{since $\cM$ is a data monoid} \\[1ex]
  &~=~ \hat\tau (f(o(\bar{d}))) \ .               \tag{by definition of $f$}
\end{align*}
Next, we verify Condition \refcond{presentation-cond2}:
\begin{align*}
  \mspace{-15mu}
  f(s \odot t) 
  &~=~ \hat\sigma_{s,t}^{-1} \circ \hat\sigma_{s,t} \circ f \,(s \odot t) 
       \tag{since $\hat\sigma_{s,t}^{-1} \circ \hat\sigma_{s,t}$ is the identity} \\[1ex]
  &~=~ \hat\sigma_{s,t}^{-1} \circ f \circ \check\sigma_{s,t} \,(s \odot t) 
       \tag{by Condition \refcond{presentation-cond1}} \\[1ex]
  &~=~ \hat\sigma_{s,t}^{-1} \circ f \circ \check\sigma_{s,t} 
       \,\Big(\check\sigma^{-1}_{s,t} \circ g \big(f \circ \hat\sigma_{s,t}(s) ~\cdot~ 
                                                   f \circ \hat\sigma_{s,t}(t) \big)\Big) 
       \tag{by definition of $\odot$} \\[1ex]
  &~=~ \hat\sigma_{s,t}^{-1} \circ f 
       \,\Big(g \big(f \circ \check\sigma_{s,t}(s) ~\cdot~ 
                     f \circ \check\sigma_{s,t}(t) \big)\Big) 
       \tag{since $\check\sigma_{s,t} \circ \check\sigma^{-1}_{s,t}$ is the identity} \\[1ex]
  &~=~ \hat\sigma_{s,t}^{-1} 
       \,\big(f \circ \check\sigma_{s,t} \,(s) ~\cdot~ 
              f \circ \check\sigma_{s,t} \,(t) \big) 
       \tag{since $f \circ g$ is the identity} \\[1ex]            
  &~=~ \hat\sigma_{s,t}^{-1} 
       \,\big(\hat\sigma_{s,t} \circ f \,(s) ~\cdot~ 
              \hat\sigma_{s,t} \circ f \,(t) \big)
       \tag{by Condition \refcond{presentation-cond1}} \\[1ex]            
  &~=~ \hat\sigma_{s,t}^{-1} \circ \hat\sigma_{s,t} 
       \,\big(f(s) \cdot f(t) \big)
       \tag{since $\cM$ is a data monoid} \\[1ex]        
  &~=~ f(s) \cdot f (t) \ .
       \tag{since $\hat\sigma_{s,t}^{-1} \circ \hat\sigma_{s,t}$ is the identity}
\end{align*}
As for Condition \refcond{presentation-cond3}, 
suppose that $s=o(d_1,\ldots,d_k)$ and $t=p(e_1,\ldots,e_h)$.
We first consider the case of a data value $d \in \mem(s)$, namely, 
$d=d_i$ for some $1\le i\le k$. By definition of $\sigma_{s,t}$, 
we have $\sigma_{s,t}(d) = i$. Hence 
$$
  \tau \circ \sigma_{s,t} \,(d) ~=~ \tau(i) 
                                ~=~ \sigma_{\check\tau(s),\check\tau(t)}\,(d_i) 
                                ~=~ \sigma_{\check\tau(s),\check\tau(t)}\,(d) \ .
$$
Next, we consider the case of a data value $d\in \mem(t)$, namely, 
$d=e_i$ for some $1\le i\le h$. By definition of $\sigma_{s,t}$,
we have $\sigma_{s,t}(d) = \len{\{e_1, \ldots, e_i\} \setminus \{d_1, \ldots, d_k\}}$.
From this we derive
$$
  \tau \circ \sigma_{s,t} \,(d) ~=~ \tau\big(\len{\{e_1, \ldots, e_i\} \setminus 
                                                  \{d_1, \ldots, d_k\}}\big) 
                                ~=~ \sigma_{\check\tau(s),\check\tau(t)}\,(e_i) 
                                ~=~ \sigma_{\check\tau(s),\check\tau(t)}\,(d) \ .
$$
The proof of the last condition 
$\sigma_{\check\tau(s),\check\tau(t)} \circ \tau\,(d) ~=~ \sigma_{s,t}\,(d)$ is similar.
\end{proof}     

Turning to the main proof of the proposition, we show that $\cS$ is
indeed a valid presentation system by verifying that all the conditions 
of Definition \ref{def:term-rep-system} are satisfied:
\begin{enumerate}
  \item {\em Identity.~}
        Recall that we defined the identity term to be $1_T = g(1_{\cM})$. 
        As $1_{\cM}$ has empty memory, we have $g(1_{\cM}) = f^{-1}(1_{\cM})$. 
        For a generic $t\in T$, we get
        \begin{align*}
          1_T \odot t 
          &~=~ \check\sigma^{-1}_{1_{\cS},t} \circ g 
               \,\big(f \circ \hat\sigma_{1_{\cS},t}\,(1_{\cS}) ~\cdot~ 
                      f \circ \hat\sigma_{1_{\cS},t}\,(t) \big) \\[1ex]
          &~=~ \check\sigma^{-1}_{1_{\cS},t} \circ g
               \,\big(f \circ \hat\sigma_{1_{\cS},t}\,(f^{-1}(1_{\cM})) ~\cdot~ 
                      f \circ \hat\sigma_{1_{\cS},t}\,(t) \big) \\[1ex]
          &~=~ \check\sigma^{-1}_{1_{\cS},t} \circ g
               \,\big(1_{\cM} ~\cdot~ f \circ \hat\sigma_{1_{\cS},t}\,(t) \big) \\[1ex]
          &~=~ \check\sigma^{-1}_{1_{\cS},t} \circ g
               \,\big(f \circ \hat\sigma_{1_{\cS},t}\,(t)\big) \\[1ex]
          &~=~ t \ .
        \end{align*}
  \item {\em Equivariance.~} 
        We verify that $\check\tau(s) \odot \check\tau(t) = \check\tau(s \odot t)$ 
        for all $s,t \in T$ and $\tau\in\cG_C$:
        \begin{align*}
          \mspace{-30mu}
          \check\tau(s) \odot \check\tau(t) 
          &~=~ \check\sigma^{-1}_{\check\tau(s),\check\tau(t)} \circ g
               \,\big(f \circ \hat\sigma_{\check\tau(s),\check\tau(t)}\,(\check\tau(s)) ~\cdot~ 
                      f \circ \hat\sigma_{\check\tau(s),\check\tau(t)}\,(\check\tau(t)) \big) \\[1ex]
          &~=~ \check\sigma^{-1}_{\check\tau(s),\check\tau(t)} \circ g
               \,\big(f \circ \hat\sigma_{s,t}\,(s) ~\cdot~ 
                      f \circ \hat\sigma_{s,t}\,(t) \big) 
               \tag{by Condition \refcond{presentation-cond4}} \\[1ex]
          &~=~ \check\tau \circ \check\sigma^{-1}_{s,t} \circ g
               \,\big(f \circ \hat\sigma_{s,t}\,(s) ~\cdot~ 
                      f \circ \hat\sigma_{s,t}\,(t) \big)
               \tag{by Condition \refcond{presentation-cond3}} \\[1ex]
          &~=~ \check\tau(s \odot t).
        \end{align*}    
  \item {\it Associativity up to $\dsim$.~}
        Recall that two terms are $\dsim$-equivalent iff $f$ maps them to the same 
        monoid element. We consider some terms $s,t,u$ and we prove that 
        $(s \odot t) \odot u ~\dsim~ s \odot (t \odot u)$ as follows:
        \begin{align*}
          f((s \odot t) \odot u) 
          &~=~ (f(s) \cdot f(t)) \cdot f(u) 
               \tag{by Condition \refcond{presentation-cond2}} \\[1ex]
          &~=~ f(s) \cdot (f(t) \cdot f(u)) 
               \tag{by associativity of $\cdot$} \\[1ex]
          &~=~ f(s \odot (t \odot u)) \ .
               \tag{by Condition \refcond{presentation-cond2}}
        \end{align*}
  \item {\em Congruence for products.~} 
        Assume that $s \dsim t$ and $u \dsim v$. 
        By exploiting Condition \refcond{presentation-cond2}
        we easily verify that $s \odot u ~\dsim~ t \odot v$: 
        \begin{align*}
          f(s \odot u) ~=~ f(s) \cdot f(u) ~=~ f(t) \cdot f(v) ~=~ f(t \odot v) \ .
        \end{align*}    
  \item {\em Congruence for renamings.~}
        Assume that $s \dsim t$ and let $\sigma$ be a renaming. 
        We need to prove that $\check\sigma(s) \dsim \check\sigma(t)$. 
        We know that $f(s) = f(t)$. Moreover, since $\hat\sigma$ is a function on $M$,
        we know that $\hat\sigma(f(s)) = \hat\sigma(f(t))$. Finally, we know from
        Condition \refcond{presentation-cond1} that 
        $f(\check\sigma(s)) = f(\check\sigma(t))$, whence
        $\check\sigma(s) \dsim \check\sigma(t)$.
\end{enumerate} 
We have just proved that $\cS$ is a term-based presentation system. 

\medskip\noindent
{\it The term-based system represents $\cM$.}
It remain to verify that $\cS=(T,\odot,\check\blank)$ represents 
the data monoid $\cM=(M,\cdot,\hat\blank)$. 
Let $\widetilde{\cM}=(\widetilde{M},\tilde{\cdot},\tilde{\hat\blank})$ be the structure
represented by $\cS$, where $\widetilde{M}$ is the set of equivalence classes of $\dsim$
and the product $\tilde\cdot$, the action $\tilde{\hat\blank}$, and the identity 
$1_{\widetilde{\cM}}$ are defined by
$$
  [s] ~\tilde\cdot~ [t] ~=~ [s \odot t] 
  \qquad\qquad
  \tilde{\hat\tau}([s]) ~=~ [\check\tau(s)]
  \qquad\qquad
  1_{\widetilde{\cM}} ~=~ [1_{\cS}] \ .
$$
We know form the first part of the proposition that $\widetilde{\cM}$ is a 
data monoid. We need to show that $\widetilde{\cM}$ and $\cM$ are isomorphic. 
For this, we consider the function $h:\widetilde{\cM} \then \cM$ defined by
$$
  h([s]) ~\eqdef~ f(s)
$$
and we show that $h$ is a data monoid isomorphism. 
We first check that $h$ is a data monoid morphism. There are three properties to check:
\begin{enumerate}
  \item We need to check that $h$ commutes with products. 
        Using Condition \refcond{presentation-cond2}, we can calculate 
        $$  
          h([s] ~\tilde\cdot~ [t]) ~=~ h([s \odot t]) 
                                   ~=~ f(s \odot t) 
                                   ~=~ f(s) \cdot f(t) 
                                   ~=~ h([s]) \cdot h([t]) \ .
        $$
  \item Next, we verify that $h$ preserves the identity: 
        $$
          h(1_{\widetilde{\cM}}) ~=~ h([1_{\cS}]) ~=~ h(g(1_{\cM})) ~=~ 1_{\cM} \ .
        $$
  \item Finally, we verify that $h$ commutes with the renamings:
        $$
          h(\tilde{\hat\tau}([s])) ~=~ h([\check\tau(s)]) 
                                   ~=~ f(\check\tau(s)) 
                                   ~=~ \hat\tau(f(s)) 
                                   ~=~ \hat\tau(h([s])) \ .
        $$
\end{enumerate} 
Furthermore, $h$ is injective by construction. It remains to show that $h$ is surjective. 
Let some $m \in M$ be given. We will show that there is a term $t \in T$ such that $f(t)=m$.
This will imply that $m$ is the image via $h$ of the element $[t]\in\tilde{M}$: indeed, 
we have $h([t]) = f(t) = m$.

Let $o$ be the orbit of $m$ and assume that it has arity $k$. 
Recall that we fixed a representative for each orbit of $\cM$, in particular, the 
representative of the orbit $o$ is $m_o$. As $m$ and $m_o$ are in the same orbit 
there must exist a data renaming $\tau$ such that $\hat\tau(m_o) = m$. Moreover, 
recall that Condition \ref{presentation-cond1} implies
$f \circ \check\tau = \hat\tau \circ f$. 
By multiplying with $\check\tau^{-1}$ to the right, 
we get $f = \hat\tau \circ f \circ \check\tau^{-1}$. 
Towards a conclusion, define $t = \check\tau(o(1, \ldots,k))$ and observe that
\begin{align*}
  f(\check\tau(o(1, \ldots,k))) 
  &~=~ \hat\tau \circ f \circ \check\tau^{-1}\,(\check\tau(o(1, \ldots,k))) \\[1ex]
  &~=~ \hat\tau \circ f \,(o(1,\ldots,k)) \\[1ex]
  &~=~ \hat\tau (\hat\sigma_{1,\ldots,k}(m_o)) 
       \tag{by the definition of $f$} \\[1ex]
  &~=~ \hat\tau (m_o)                            
       \tag{since $\sigma_{1,\ldots,k}$ is the identity} \\[1ex]
  &~=~ m \ .
\end{align*}
We have just shown that $f(t)=m$ and hence $h$ is surjective. 
This completes the proof of the proposition.
\end{proof}

\subsection{Green's relations and memorable values}\label{subsec:green-theory}

In Section \ref{sec:monoids2logic} we will show how recognizability by an orbit 
finite data monoid corresponds to definability by a formula of rigidly guarded 
MSO logic. Like in the theorem of Schützenberger \cite{aperiodic_monoids}, the 
translation from a monoid to a formula exploits an induction on certain ideals 
of the monoid that are induced by the so-called Green's relations 
\cite{green_relations,mathematical_foundations_of_automata}. 
The goal of this section is to recall the basic ingredients of this theory
and further develop it in order to ease the inductive constructions on 
orbit-finite data monoids.

As already noticed in \cite{data_monoids}, a relevant part of the theory 
of Green's relations, which holds for finite monoids, can be lifted to 
{\sl locally finite} monoids, namely, to monoids such that all finitely
generated sub-monoids are finite. In particular, this applies to orbit-finite 
data monoids. The basic Green's relations 
$\le_\cR$, $\le_\cL$, $\le_\cJ$ associated with an orbit-finite data 
monoid $\cM$ are the {\sl preorders} defined by:
$$
\begin{array}{rclcrcl}
  s &\le_\cR& t &\qquad\text{iff}&\qquad s\cdot M        &\subseteq&t\cdot M \qquad \\[1ex]
  s &\le_\cL& t &\qquad\text{iff}&\qquad M\cdot s        &\subseteq& M\cdot t \qquad \\[1ex]
  s &\le_\cJ& t &\qquad\text{iff}&\qquad M\cdot s\cdot M &\subseteq& M\cdot t\cdot M \ .
\end{array}
$$
We remark the following crucial property: for every orbit-finite data monoid,
the preorder $\le_\cJ$ is well-founded (for a proof of this result, see Lemma 
9.3 in \cite{nominal_monoids}). This provides the inductive principle that will
be used in our proofs.

We also denote by $=_\cR$, $=_\cL$, $=_\cJ$ the corresponding equivalence relations 
(e.g., $s =_\cJ t$ iff $s \le_\cJ t$ and $t \le_\cJ s$) and we introduce an additional 
fourth equivalence $=_\cH$ defined by 
$$
  s ~=_\cH~ t 
  \qquad\text{iff}\qquad 
  s ~=_\cR~ t 
  ~\text{ and }~ 
  s ~=_\cL~ t \ . 
  \mspace{8mu}\ 
$$ 
Given an element $s$ of a data monoid $\cM$, we denote by $\cR(s)$ (resp., $\cL(s)$, $\cJ(s)$, 
$\cH(s)$) the $=_\cR$-class (resp., $=_\cL$-class, $=_\cJ$-class, $=_\cH$-class) of $s$. 
We remark that the equivalence relation $=_\cR$ (resp., $=_\cL$) is a congruence with
respect to products on the left (resp., right). For example, we have that $s =_\cR t$ 
implies $u\cdot s =_\cR u\cdot t$. 

We naturally lift the above relations to orbits. Specifically, for each $\cK$ 
among $\cR$, $\cL$, $\cJ$, we denote by $\le_{\orbit{\cK}}$ the preorder relation 
such that $s \le_{\orbit{\cK} }t$ iff $s \le_{\cK} \tau(t)$ for some renaming 
$\tau\in\cG_D$. 
We do the same for the equivalence relations $=_\cR$, $=_\cL$, $=_\cJ$, 
$=_\cH$, thus obtaining the relations $=_{\orbit{\cR}}$, $=_{\orbit{\cL}}$, 
$=_{\orbit{\cJ}}$, $=_{\orbit{\cH}}$.

\medskip
The last part of this section is devoted to an analysis of the types of data values 
that can occur in the memory of an element of an orbit-finite data monoid.
We begin by distinguishing two types of data values.

\begin{definition}\label{def:l-r-values}
Given an element $s$ of an orbit-finite data monoid $\cM$, we define $\memr(s)$
(resp., $\meml(s)$) to be the intersection of the memories of the elements in the 
$=_\cR$-class (resp., $=_\cL$-class) of $s$:
$$
  \memr(s) ~\eqdef~ \displaystyle\bigcap\limits_{t \:\in\, \cR(s)}\mem(t) 
  \qquad\qquad
  \meml(s) ~\eqdef~ \displaystyle\bigcap\limits_{t \:\in\, \cL(s)}\mem(t) \ . 
$$
We call \emph{$\cR$-memorable} (resp., \emph{$\cL$-memorable}) values of $s$ the 
values in $\memr(s)$ (resp., $\meml(s)$).
\end{definition}

Quite surprisingly, it turns out that the memory of every element of an orbit-finite
data monoids consists only of $\cR$-memorable and $\cL$-memorable values:

\begin{proposition}\label{prop:memorable-values}
For every element $s$ of an orbit-finite data monoid, 
we have $\mem(s) = \memr(s) \cup \meml(s)$.
\end{proposition}

Before turning to the proof Proposition~\ref{prop:memorable-values}, let us show that 
a similar result fails for data monoids with infinitely many data orbits. 

\begin{example}\label{ex:infinite-data-group}
Consider the data language $L_{\mathsf{even}} \subseteq D^*$ of all words where 
every value occurs an even number of times. The syntactic data monoid of the language 
$L_{\mathsf{even}}$ consists of one element $s_C$ for each finite subset $C$ of $D$. 
The product corresponds to the symmetric difference of sets.
It is easy to see that the memorable values of $s_C$ are exactly the values in $C$,
which are neither $\cL$-memorable nor $\cR$-memorable (the syntactic data monoid is 
indeed a group).
\end{example}

In order to prove Proposition \ref{prop:memorable-values}, we need to introduce a couple 
of other concepts. An \emph{inverse} of an element $s$ of a monoid, is an element $t$ 
such that $s\cdot t=t\cdot s=1_\cM$. If the inverse of $s$ exists, then it can be easily 
proven to be unique and hence it can be denoted by $s^{-1}$. A \emph{data group} is 
simply a \emph{data monoid} where all elements have an inverse. The next lemma shows 
that orbit-finiteness is a severe restriction for data groups.

\begin{lemma}\label{lemma:data-groups}
Every orbit-finite data group is finite.
\end{lemma}     

\begin{proof}
We begin by proving the following claim:

\begin{claim}
If $s,t,u$ are elements of a data group (not necessarily orbit-finite), then
$$
  \mem(s) = \mem(s^{-1})
  \qquad\text{and}\qquad
  \mem(s\cdot t\cdot u) ~\supseteq~ \mem(t) \setminus \mem(s) \setminus \mem(u) \ .
$$
\end{claim}

\begin{proof}[Proof of claim]
We first prove the equality on the left.
More precisely, we prove that $\mem(s)\subseteq\mem(s^{-1})$
(by symmetric arguments, one can prove that $\mem(s^{-1}) \subseteq \mem(s)$
holds as well).
Recall that, by Definition \ref{def:memory}, the memory $\mem(t)$ 
of an element $t$ contains a data value $d$ iff, for all sets 
$C\subseteq D\setminus\{d\}$, there is a renaming $\tau$ 
that is the identity on $C$ and such that $t\neq \tau(t)$.
Let $d$ be a data value in $\mem(s)$.
To prove that $d\in\mem(s^{-1})$, we consider a generic  
set $C\subseteq D\setminus\{d\}$. Since $d\in\mem(s)$,
we know that there is a renaming $\tau$ that is the identity
on $C$ and such that $s\neq \tau(s)$.
Moreover, because the identity $1$ of the data group has empty memory,
we have that $1=\tau(1)$, and hence
$$
  s \cdot s^{-1} ~=~ 1 
                 ~=~ \tau(1) 
                 ~=~ \tau(s\cdot s^{-1}) 
                 ~=~ \tau(s)\cdot \tau(s^{-1}) \ .
$$ 
Finally, because $s\neq\tau(s)$ and because each element of the data group
has exactly one inverse, we derive $s^{-1}\neq\tau(s^{-1})$.
This proves that $d\in\mem(s^{-1})$ and hence $\mem(s)\subseteq\mem(s^{-1})$.

We conclude by proving the containment on the right. For this, it is sufficient to
observe that
$\mem(t) = \mem(s^{-1}\cdot s\cdot t\cdot u\cdot u^{-1})
         \subseteq \mem(s\cdot t\cdot u) \cup \mem(s^{-1}) \cup \mem(u^{-1})
         = \mem(s\cdot t\cdot u) \cup \mem(s) \cup \mem(u)$,
and hence 
$\mem(t) \setminus \mem(s) \setminus \mem(u) \subseteq \mem(s\cdot t\cdot u)$.
\end{proof}

To prove the lemma assume, towards a contradiction, 
that $\cG$ is an infinite data group with finitely many orbits. 
$\cG$ must contain an infinite orbit $o$, and hence
we can inductively construct an infinite subset $G=\{g_1, g_2, \ldots\}$ of $o$ 
such that each element $g_i$ has a distinguished memorable value $d_i$ that is 
not memorable in any other element of $G$, namely, for all $i$, we have
$d_i \in \mem(g_i) \setminus \bigcup_{j\neq i}\mem(g_j)$. 
Observe that, for all $i\le k$,
$\bigcup_{j<i} \mem(g_j) \supseteq \mem(g_1\cdot\ldots\cdot g_{i-1})$ and
$\bigcup_{j>i} \mem(g_j) \supseteq \mem(g_{i+1}\cdot\ldots\cdot g_k)$
(this holds in any data monoid, not necessarily in a data group).
In particular, we have that for all $i\le k$,
$d_i \in \mem(g_i) \setminus \mem(g_1\cdot\ldots\cdot g_{i-1}) 
                   \setminus \mem(g_{i+1}\cdot\ldots\cdot g_k)$.
Finally, using the previous claim, we derive that, for all $k$,
$$
  \{d_1,\dots,d_k\} ~\subseteq~ \mem(g_i\cdot\ldots\cdot g_k) \ .
$$
Becasue the values $d_1,\ldots,d_k$ are pairwise distinct, 
this contradicts the finite memory axiom.
\end{proof}

With each $=_\cH$-class $H$ of a monoid one can associate a group $\Gamma(H)$, called the 
Sch\"utzenberger group \cite{mathematical_foundations_of_automata} (in fact there exist two such 
groups, but we will only consider one of them here). To define $\Gamma(H)$, we first introduce 
the set $T(H)$ of all elements $t \in H$ such that $t\cdot H$ is a subset of $H$. 
For each $t\in T(H)$, we then let $\gamma_t$ be the transformation on $H$ that maps 
$h\in H$ to $t\cdot h$. Finally, we define the Sch\"utzenberger group $\Gamma(H)$ as
the set of all transformations $\gamma_t$, with $t\in T(H)$, equipped with the functional 
composition $\circ$ as binary product. 

There is also a natural way to extend the 
action $\hat\blank$ on the data monoid $\cM$ 
to an action $\tilde\blank$ on $\Gamma(H)$ by simply 
letting $\tilde\tau(\gamma_s) = \gamma_{\hat{\tau}(s)}$ for all renamings 
$\tau\in\Gamma_{D \setminus (\memr(H)\,\cup\,\meml(H))}$, where $\memr(H)=\memr(h)$ 
and $\meml(H)=\meml(h)$ for some arbitrary element $h\in H$ (note that all elements 
of $H$ have the same set of $\cR$-memorable values and the same set of $\cL$-memorable 
values). The following lemma shows that $\tilde\blank$ is indeed a group action on 
the Sch\"utzenberger group $\Gamma(H)$.

\begin{lemma}\label{lemma:schutzenberger-data-group}
If $\cM=(M,\cdot,\hat\blank)$ is a data monoid over $D$ and $H$ is an $=_\cH$-class of $\cM$, 
then $(\Gamma(H),\circ,\tilde\blank)$ is a data group over $D\setminus (\memr(H)\,\cup\,\meml(H))$.
Moreover, if $\cM$ is orbit-finite, then so is $(\Gamma(H),\circ,\tilde\blank)$.
\end{lemma}

\begin{proof}
It is known that $(\Gamma(H),\circ)$ is a group. We only need to verify that $\tilde\blank$ 
is an action on $\Gamma(H)$. We first show that $\Gamma(H)$ is closed under the action 
$\tilde\blank$ induced by the renamings over $D\setminus (\memr(H)\cup\meml(H))$. 
By definition of $\tilde\blank$, this is equivalent to verifying that $H$ is closed under 
the action $\hat\blank$ of renamings over $D\setminus (\memr(H)\cup\meml(H))$. The proof 
is thus similar to proof of the Memory Theorem for $=_\cJ$-classes in \cite{data_monoids}; 
however, we give a complete proof here for the sake of self-containment. 

Suppose that $H$ is the intersection of an $=_\cR$-class $R$ and an $=_\cL$-class $L$. 
Since a renaming is a permutation that is the identity on all but finite many values, 
any renaming over $D\setminus (\memr(H)\cup\meml(H))$ can be decomposed into a sequence 
of transpositions of pairs of values from $D\setminus (\memr(H)\cup\meml(H))$. 
Therefore, in order to prove the closure of $H=R\cap L$ under the action of renamings over 
$D\setminus (\memr(H)\cup\meml(H))$, it is sufficient to prove a similar closure property for
the transpositions $\pi_{de}$ of pair of elements $d,e\nin\memr(H)\cup\meml(H)$. We first show 
that $R$ is closed under such transpositions. Let $d$ and $e$ be two values outside $\memr(H)$ 
and let $\pi_{de}$ be their transposition. Since $d,e\nin\memr(H)$, we know that there exist 
two elements $s,t\in R$ such that $d\nin\mem(s)$ and $e\nin\mem(t)$. 
Let $f$ be a data value outside $\mem(s) \cup \mem(t)$. By definition of memory, we know 
that $\hat\pi_{df}$, where $\pi_{df}$ is the transposition of $d$ and $f$, is a stabilizer 
of $s$ and, similarly, $\hat\pi_{ef}$ is a stabilizer of $t$. Now, consider an element $s'$ 
that is $\cR$-equivalent to $s$. There must exist some elements $u$ and $u'$ in $\cM$ such 
that $s\cdot u=s'$ and $s'\cdot u'=s$. Since $\hat\blank$ commutes with the product of $\cM$, 
we obtain $\hat\pi_{df}(s') = \hat\pi_{df}(s\cdot u)$ and hence 
$\hat\pi_{df}(s') \le_{\cR} \hat\pi_{df}(s)$. By similar arguments, 
we obtain $\hat\pi_{df}(s') \ge_\cR \hat\pi_{df}(s)$. 
We thus have $\hat\pi_{df}(s') \,=_\cR\, \hat\pi_{df}(s) \,=\, s \in R$. 
A symmetric argument shows that $\hat\pi_{ef}(s') \,=_\cR\, \hat\pi_{ef}(s) \,=\, s \in R$. 
Moreover, since $\pi_{de}=\pi_{df}\circ\pi_{ef}\circ\pi_{df}$, we conclude that 
$\hat\pi_{de}(s) \in R$. Finally, a similar proof shows that $\hat\pi_{de}(s) \in L$. 
Putting all together, we have that for every $s\in H=R\cap L$ and every renaming $\tau$ 
over $D\setminus (\memr(H)\cup\meml(H))$, $\hat\tau(s)\in R\cap L=H$. This shows that $H$ 
is closed under the action $\hat\blank$ of renamings over $D\setminus (\memr(H)\cup\meml(H))$.

Below, we verify that $\tilde\blank$ is a group morphism from the group of renamings 
over $D\setminus (\memr(H)\,\cup\,\meml(H))$ to the group of automorphisms on $\Gamma(H)$. 
Clearly, the function $\tilde\blank$ maps the identity $\iota$ on 
$\cG_{D\setminus (\memr(H)\cup\meml(H))}$ to the trivial automorphism $\tilde\iota$ on 
$\Gamma(H)$ (i.e., $\tilde\iota(\gamma_s)=\gamma_{\hat\iota(s)}=\gamma_s$). Moreover, 
$\tilde\blank$ is a morphism because
$$
  \widetilde{\tau \circ \pi}(\gamma_s) ~=~
  \gamma_{\widehat{\tau \circ \pi}(s)} ~=~
  \gamma_{\hat\tau \circ \hat\pi(s)} ~=~
  (\tilde\tau \circ \tilde\pi)(\gamma_s) \ .
$$
Finally, we observe that $\gamma_{s \cdot t} = \gamma_s \circ \gamma_t$ 
(indeed, for every $h\in H$, we have
$\gamma_{s \cdot t}(h) \,=\, (s \cdot t) \cdot h 
                       \,=\, s \cdot (t \cdot h) 
                       \,=\, \gamma_s(t \cdot h) 
                       \,=\, \gamma_s \circ \gamma_t(h)$)
and hence
$$
  \hat{\tau}(\gamma_s) \circ \hat{\tau}(\gamma_t) ~=~
  \gamma_{\hat{\tau}(s)} \circ \gamma_{\hat{\tau}(t)} ~=~
  \gamma_{\hat{\tau}(s) \cdot \hat{\tau}(t)} ~=~ 
  \gamma_{\hat{\tau}(s \cdot t)} \;=\; 
  \hat{\tau}(\gamma_{s \cdot t}) ~=~
  \hat{\tau}(\gamma_s \circ \gamma_t) \ .
$$

\noindent
To complete the proof of the lemma, we need to show that $(\Gamma(H),\circ,\tilde\blank)$ 
is orbit-finite when $\cM$ is orbit-finite. Let us consider two elements $s,t \in H$ and 
suppose that $s$ and $t$ are in the same orbit, namely, that there is 
$\tau\in\cG_{D\setminus (\memr(H)\cup\meml(H))}$ such that $t=\hat\tau(s)$. 
Since $\tilde\blank$ is a group action, we know that $\gamma_t=\gamma_{\hat\tau(s)}=\tilde\tau(\gamma_s)$. 
This shows that the two elements $\gamma_s$ and $\gamma_t$ of $\Gamma(H)$ are on the same orbit.
\end{proof}     

It is known that any $=_\cH$-class $H$ of a monoid has the same cardinality 
of the associated Sch\"utzenberger group $\Gamma(H)$ (see, for instance, 
\cite{mathematical_foundations_of_automata}). This implies the following 
crucial property:

\begin{corollary}\label{cor:finite-H-classes}
All $=_\cH$-classes of an orbit-finite data monoid are finite.
\end{corollary}

\begin{proof}
Let $H$ be an $=_\cH$-class of an orbit-finite data monoid $\cM$. 
By Lemma \ref{lemma:schutzenberger-data-group}, we can associate 
with $H$ an orbit-finite data group $(\Gamma(H),\circ,\tilde\blank)$, 
where $\Gamma(H)$ is the Sch\"utzenberger group of $H$. Lemma 
\ref{lemma:data-groups} implies that $\Gamma(H)$ is finite.
Finally, since the $=_\cH$-class $H$ has the same cardinality as
its Sch\"utzenberger group $\Gamma(H)$, we conclude that $H$ 
is finite.
\end{proof}

We are now ready to prove that every memorable value is either
$\cR$-memorable or $\cL$-memorable:

\begin{proof}[Proof of Proposition \ref{prop:memorable-values}]
Let $s$ be an element of an orbit-finite data monoids $\cM$.
We aim at proving that $\mem(s)\subseteq\memr(s) \cup \meml(s)$. 
Assume towards a contradiction that there is a value 
$d \in \mem(s)\setminus( \memr(s) \cup \meml(s))$. Since $d\nin\memr(s)$, 
there is an $s'$ in the $=_\cR$-class of $s$ such that $d\nin\mem(s')$. 
Then there are elements $u,u' \in \cM$ such that $s\cdot u=s'$ and $s'\cdot u'=s$. 
Symmetrically, as $d\nin\meml(s)$, $s$ has an $\cL$-equivalent element $s''$ 
such that $d\nin\mem(s'')$, and there are $v,v'' \in \cM$ such that 
$v\cdot s=s''$ and $v''\cdot s''=s$. 

Let $d_1,d_2,\ldots$ be an infinite sequence of pairwise distinct values 
that are not in the memory of either $s$, $s'$, or $s''$. We denote by 
$\pi_i$ the transposition of $d$ with $d_i$. As neither $d$ nor $d_1,d_2,\ldots$ 
are in the memory of $s'$, $\tau_i(s')=s'$ and hence 
$\tau_i(s\cdot u)=\tau_i(s')=s'$. Combining this with $s'\cdot u'=s$, 
we obtain $\tau_i(s)\cdot\tau_i(u)\cdot u' = s'\cdot u' = s$ and hence 
$\tau_i(s) \ge_{\cR} s$. Similarly, one proves that $\tau_i(s) \le_{\cR} s$ 
and hence $\tau_i(s) =_\gR s$. By symmetry, one gets $\tau_i(s) =_\gL s$. 
We have just shown that $\tau_i(s)$ belongs to the $=_\cH$-class of $s$. 

Towards a conclusion, we recall that $d\in\mem(s)$ and $d\nin\mem(\tau_i(s))$, and 
hence $s$ is different from $\tau_i(s)$. Similarly, for all $j\neq i$, we have 
that $d_i\in\mem(\tau_i(s))$ and $d_i\nin\mem(\tau_j(s))$, and hence 
$\tau_i(s)$ is different from $\tau_j(s)$. We must conclude that 
the $=_\cH$-class of $s$ is infinite, contradicting Corollary 
\ref{cor:finite-H-classes}.
\end{proof}

\section{Rigidly guarded \MSO and its variants}\label{sec:logics}

From now on, we abbreviate by \emph{\MSO} the monadic second-order logic extended 
with data equality tests. Formally, \MSO formulas are built up from atoms of the form 
$x<y$, $x\in X$, or $x\sim y$, $a(x)$, where $a$ ranges over a fixed finite alphabet, 
using boolean connectives and existential quantifications over first-order variables 
(e.g., $x,y,z,\ldots$) and monadic second-order variables (e.g., $X,Y,Z,\ldots$). 
The meaning of the atom $x\sim y$ is that the data values 
at the two positions that correspond to the interpretation of the variables $x$ and $y$ 
must be the same. The meaning of the other predicates is as usual. 
We write $u\models\varphi$ whenever a formula $\varphi$ holds over the data word $u$.
Moreover, we write $\varphi(x_1,\ldots,y_n,X_1,\ldots,X_m)$ whenever we want to make explicit 
that the free variables of $\varphi$ are among $x_1,\ldots,x_n,X_1,\ldots,X_m$. 

We recall that the satisfiability problem for \MSO interpreted over data words is undecidable.
Intuitively, the source of undecidability lies in the fact that the equality relationships
between the values in a data word may form a grid-like structure onto which one can encode 
computations of Turing machines \cite{undecidability_data_languages}. Here we pursue the idea 
of restricting the use of data equality tests in \MSO formulas with the general goal of excluding 
the possibility of logically defining grids inside the input data words. We will see that this approach
can be used to rule out the source of undecidability of MSO even when the logic is interpreted 
over classes of graphs that contain grid-minors of unbounded size \cite{graph_minors_treewidth}.

The general idea is to guard the data equality predicate $x\sim y$ by a formula $\varphi(x,y)$ 
that is \emph{rigid}, in the sense that it defines a partial bijection on the positions of every
input data word. This gives rise to a fragment of \MSO that we call rigidly guarded \MSO:

\begin{definition}\label{def:rigid-MSO}
The logic \emph{rigidly guarded \MSO} consists of formulas generated according to the following grammar:
\begin{align*}
  \varphi ~~:=~~
  \exists x~ \varphi ~\mid~ \exists Y~ \varphi ~\mid~ a(x) ~\mid~ x<y 
                     ~\mid~ x\in Y ~\mid~ \neg\varphi ~\mid~ \varphi \:\et\: \varphi 
                     ~\mid~ \varphi_\rigid(x,y) \:\et\: x \sim y
\end{align*}
where $a$ ranges over a fixed finite alphabet $A$ and 
$\varphi_\rigid(x,y)$ denotes a formula generated by the 
same grammar that in addition satisfies the rigidity constraint, 
that is, for all data words $u\in(D\times A)^*$ and all positions 
$x$ (resp., $y$) in $u$, there is {\sl at most one} position $y$ 
(resp., $x$) in $u$ such that $u\models\varphi_\rigid(x,y)$. 
\par\noindent
We call \emph{rigidly guarded \FO} the first-order fragment of rigidly guarded \MSO.
\end{definition}

The notion of rigidity is a semantic property, and this may seem problematic.
However, we can enforce rigidity syntactically as follows. Instead of guarding
a data test by a generic rigid formula $\varphi_\rigid(x,y)$, one uses the new guard
$$
  \tilde\varphi_\rigid(x,y) 
  ~\eqdef~ \varphi_\rigid(x,y) ~\et~ \forall x',y'~ \varphi_\rigid(x',y') \then (x=x' \iff y=y') \ .
$$
It is easy to check that $\tilde\varphi_\rigid$ is always rigid and, furthermore,
if the original formula $\varphi_\rigid$ is rigid, then it is also equivalent to 
$\tilde\varphi_\rigid$. This trick allows us to enforce rigidity syntactically. 
We will prove later in Corollary \ref{cor:decidability} that one can decide if 
a given formula respects the semantic assumption of rigidity in all its guards 
(the problem is of course undecidable when data tests are not guarded).

We also remark that in rigidly guarded \MSO, the similar constructions 
$\varphi_\rigid(x,y) \et x \nsim y$, $\varphi_\rigid(x,y) \rightarrow x\sim y$, 
and $\varphi_\rigid(x,y) \rightarrow x\nsim y$ can be derived. 
This is thanks to the Boolean equivalences $\alpha\rightarrow\beta$ iff 
$\alpha\then(\alpha\et\beta)$, $\alpha\et \neg\beta$ iff $\neg(\alpha\then\beta)$, 
and $\alpha\then\neg\beta$ iff $\neg(\alpha\et\beta)$.

\begin{example}\label{ex:rigidguard}
We show how to define in rigidly guarded \FO the language $L_{\ge k}$ of all data 
words that contain at least $k$ different data values. If~$k=1$ we just need to check 
that the input data word is not empty, e.g., by the sentence $\exists x ~ \true$.
If~$k=2$ it is sufficient to check the existence of two distinct {\sl consecutive} 
data values, e.g., by the sentence $\exists x,y ~ (x+1=y) \,\et\, x\nsim y$. 
For~$k>2$, one can proceed by induction as follows. One first observes that if a word 
has at least $k$ distinct data values, then there is a {\sl minimal} factor witnessing 
this property, say $[x,y]$. A closer inspection reveals that, in this case, $[x+1,y-1]$ 
is a maximal factor that uses exactly $k-2$ data values and thus belongs to the language 
$L_{\ge k-2}\setminus L_{\ge k-1}$, which is definable in rigidly guarded \FO thanks
to the inductive hypothesis. Moreover, the formula $\varphi(x',y')$ that defines the 
endpoints $x'=x+1$ and $y'=y-1$ of a {\sl maximal} factor in $L_{\ge k-2}\setminus L_{\ge k-1}$ 
is rigid. We can thus define the language $L_{\ge k}$ by means of the rigidly guarded 
\FO sentence $\exists x,y ~ \varphi(x+1,y-1) \,\et\, x\nsim y$.
\end{example}

Rigidly guarded \MSO will be the main object of our study. As we already mentioned, in Sections
\ref{sec:logic2monoids} and \ref{sec:monoids2logic} we will show that the data languages 
definable in rigidly guarded \MSO are exactly the languages recognizable by orbit-finite 
data monoids. 
This result, which is interesting in its own right, also implies that rigidly guarded \MSO 
has a decidable satisfiability problem over the class of data words: satisfiability indeed 
reduces to the problem of checking emptiness of data languages recognized by orbit-finite 
data monoids. However, one can prove decidability of rigidly guarded \MSO by a more direct
(and general) argument. 
In the following, we outline the argument underlying decidability of a logic that is
even more expressive than rigidly-guarded \MSO and that is interpreted over generic 
classes of \emph{relational structures with data values}.
Formally, for a fixed signature consisting of $m$ relational symbols $R_1,\ldots,R_m$,
we consider structures of the form $\cS=(U,R^\cS_1,\ldots,R^\cS_m,\lambda)$, where $U$ 
is the universe of the structure, $R^\cS_i$ is a relation over $U$ of the same arity 
as $R_i$, say $R^\cS_i \subseteq U^{k_i}$, and $\lambda:U\then D$ is a labelling 
function associating data values with the elements of $\cS$. Similar structures 
have been considered for example in \cite{expressive_queries_on_data_graphs,querying_data_graphs}. 

To define the variant of rigidly-guarded \MSO, we relax the rigidity constraint. 
Given a class $\sC$ of relational structures with data values and given a generic 
formula $\varphi(x,y)$ interpretable over $\sC$, we say that $\varphi(x,y)$ is 
\emph{semi-rigid} (with respect to $\sC$) if for every relational structure 
$\cS=(S,R_1,\ldots,R_m,\lambda)$ in $\sC$ and every element $x\in S$, there 
is {\sl at most one} vertex $y\in S$ such that $\cS \models \varphi(x,y)$.
As we already seen before, semi-rigidity can be enforced syntactically over 
any class of data graphs, so it is not really important here under which class
of structures the guards are assumed to be semi-rigid. We define below the 
variant of \MSO in which data tests are guarded by semi-rigid formulas.

\begin{definition}\label{def:semi-rigid-MSO}
The logic \emph{semi-rigidly guarded \MSO}, interpreted over a class $\sC$ 
of relational structures with data values, consists of formulas generated 
by the following grammar:
\begin{align*}
  \varphi ~~:=~~
  \exists x~ \varphi &~\mid~ \exists Y~ \varphi ~\mid~ R_i(x_1,\ldots,x_{k_i}) ~\mid~ 
                             x\in Y \:\mid\: \neg\varphi ~\mid~ \varphi \:\et\: \varphi \\[0.5ex]
                     &~\mid~ \varphi_\semirigid(x,y) \:\et\: \varphi_\semirigid(x,z) \:\et\: y \sim z
\end{align*}
where the occurrences of $\varphi_\semirigid(x,y)$ and $\varphi_\semirigid(x,z)$
above denote possibly different formulas generated by the same grammar, 
sharing a common variable $x$, and being semi-rigid (w.r.t. $\sC$). 
\end{definition}

Clearly, rigidly guarded \MSO can be seen as a fragment of semi-rigidly guarded \MSO
over the class of data words.
The interesting feature of these logics is that their satisfiability problems can be 
reduced to the satisfiability problem for classical MSO over the same classes of 
structures, where of course data values become immaterial. For example, one can decide 
whether a given semi-rigidly guarded \MSO formula is satisfiable over the class of all 
(finite or infinite) data words/trees. 

\begin{theorem}\label{thm:decidability}
Let $\sC$ be a class of relational structures for which membership of a structure $\cS$
in $\sC$ does not depend on labelling of the elements of $\cS$ by data values.
The satisfiability problem of semi-rigidly guarded \MSO over $\sC$ is reducible 
to the satisfiability problem of classical MSO over the same class $\sC$.
\end{theorem}

\proof
We begin by making the following crucial remark: every semi-rigidly guarded data test 
$\alpha(x,y) ~\et~ \beta(x,z) ~\et~ y\sim z$ can be normalized into the formula
$$
  \alpha(x,y) ~\et~ \beta(x,z) ~\et~ 
  \underbrace{(\exists y',z'~ \alpha(x,y') \:\et\: \beta(x,z') \:\et\: y'\sim z')}%
  _{\text{call it }\gamma^\sim_{\alpha,\beta}(x)} \ .
$$
This formalizes the idea that semi-rigidly guarded data tests behave almost like unary 
predicates. It is indeed possible to transform a given semi-rigidly guarded \MSO formula 
$\varphi$ into a classical MSO formula $\varphi^-$ inductively as follows:
$$
\begin{array}{rclrcl}
  (\exists x~ \psi)^-         &\eqdef& \exists x~ \psi^-      & \qquad\qquad
  (x\in Y)^*                  &\eqdef& x\in Y                    \\[1ex]
  (\exists Y~ \psi)^-         &\eqdef& \exists Y~ \psi^-      & \qquad\qquad
  (\neg\psi)^-                &\eqdef& \neg\psi^-                \\[1ex]
  (R_i(x_1,\ldots,x_{k_i}))^- &\eqdef& R_i(x_1,\ldots,x_{k_i}) & \qquad\qquad
  (\psi_2\et\psi_2)^-         &\eqdef& \psi_1^- \et \psi_2^-
\end{array}
$$
and, most importantly,
$$
  \big(\alpha(x,y) ~\et~ \beta(x,z) ~\et~ y\sim z\big)^- 
  ~\eqdef~\; 
  \alpha^-(x,y) ~\et~ \beta^-(x,z) ~\et~ x\in c^\sim_{\alpha,\beta} 
$$
where $c^\sim_{\alpha,\beta}$ is a fresh unary predicate.

Let $\varphi$ be a sentence for which we want to decide satisfiability.
For the sake of brevity, let $C=\{c_1,\ldots,c_n\}$ be the set of unary 
predicates $c^\sim_{\alpha,\beta}$ that correspond to the normalized 
semi-rigidly guarded data tests $\gamma^\sim_{\alpha,\beta}(x)$ occurring 
in $\varphi$.
Given a relational structure with data values $\cS$, we denote by $\cS^-$ 
the relational structure without data values that is obtained from $\cS$ 
by removing the labelling function and by expanding the structure with
the predicates $c^\sim_{\alpha,\beta}\in C$ in such a way that 
$$
  \cS^- \models c^\sim_{\alpha,\beta}(x)
  \qquad\text{iff}\qquad 
  \cS \models \gamma^\sim_{\alpha,\beta}(x) \ .
$$
Clearly, for every relational structure $\cS$ with data values, we have that 
$\cS \models \varphi$ iff $\cS^- \models \varphi^-$. 

Now, it is possible to characterize the class of relational structures without
data values of the form $\cS^-$ without taking into account the data values in $\cS$. 
For this it is sufficient to test whether the universe of the structure can be 
partitioned into classes in such a way that, for every element $x$, 
$x$ satisfies $c^\sim_{\alpha,\beta}$ iff there exist (unique) elements
$y$ and $z$ in the same class that satisfy $\alpha^-(x,y)$ and $\beta^-(x,z)$, 
respectively.
If such a partition exists, then one can reconstruct the corresponding relational
structure with data values $\cS$ by assigning different data values to elements 
of different classes. Conversely, if the relational structure is known to be of 
the form $\cS^-$, then a partition can be found by simply grouping the elements
of $\cS$ having the same data value.
Moreover, one can easily see that the coarsest partitions satisfying the above
property contain at most $n$ classes (recall that $n$ is the number of occurrences 
of semi-rigidly guarded data tests in our sentence $\varphi$). This allows us to 
define the class of relational structures of the form $\cS^-$ by means of a simple 
MSO formula:
\begin{align*}
  \varphi_{\mathsf{partition}} 
  ~~\eqdef~~& \exists Z_1,\ldots,Z_n~ \displaystyle\bigwedge_{1\le i<j\le n} 
                                      (Z_i\cap Z_j=\emptyset)  \\
       \et~~& \forall x~ \Big(\: c^\sim_{\alpha,\beta}(x) ~\iff~
                              \exists y,z ~ \alpha^-(x,y) \,\et\, \beta^-(x,z) \,\et 
                              \displaystyle\bigvee_{1\le i\le n}
                              (y\in Z_i \et z\in Z_i) \:\Big) \ .
\end{align*}
Putting everything together, we have that a relational structure with data values
$\cS\in\sC$ satisfies the semi-rigidly guarded \MSO sentence $\varphi$ iff $\cS$ 
(or any other relational structure that differs from $\cS$ only in the data values) 
satisfies the MSO sentence
$$
  \varphi^= ~\eqdef~ \exists c_1,\ldots,c_n~ \varphi^- \,\et\, \varphi_{\mathsf{partition}} \ .\eqno{\qEd}
$$
\medskip

\begin{corollary}\label{cor:decidability}
The satisfiability problem for semi-rigidly guarded \MSO over the class of 
data trees is decidable. Moreover, one can decide whether a given formula
belongs to semi-rigidly guarded \MSO, or even belongs to rigidly guarded \MSO.
\end{corollary}

\begin{proof}
By Theorem \ref{thm:decidability}, the satisfiability problem for semi-rigidly 
guarded \MSO over the class of all data trees is reduced to the satisfiability 
problem for classical MSO over trees, which is known to be decidable \cite{s2s}.

We explain how to decide whether a given formula $\varphi$ belongs to
semi-rigidly guarded \MSO (a similar argument can be used to test membership 
in rigidly guarded \MSO). For this it is sufficient to check, in bottom-up
manner, that every sub-formula of $\varphi$ satisfies the syntactic and semantic
restrictions enforced by the grammar of semi-rigidly guarded \MSO. In particular, 
if $y\sim z$ is a data test that occurs as a sub-formula of $\varphi$, one needs
to check that (i) this test is guarded by a conjunction of two formulas $\alpha(x,y)$ 
and $\beta(x,z)$ and (ii) both sub-formulas $\alpha(x,y)$ and $\beta(x,z)$ are
semi-rigid. Assuming that the sub-formula $\alpha(x,y)$ is already known to 
belong to semi-rigidly guarded \MSO, one can decide semi-rigidity of $\alpha(x,y)$
by testing the validity of the sentence 
$$
  \alpha_{\mathsf{semirigid?}} 
  ~~\eqdef~~ \forall x,y,y' ~ \alpha(x,y) \,\et\, \alpha(x,y') \:\then\: y=y' \ .
$$
A similar test can be performed on the sub-formula $\beta(x,z)$.
\end{proof}

\section{From rigidly guarded \MSO to orbit-finite monoids}\label{sec:logic2monoids}
 
In this section, we show that every data language defined by a rigidly guarded \MSO sentence 
is recognized by an orbit-finite data monoid. Our proof follows the classical technique 
for showing that MSO definable languages over standard words can be recognized by monoids. 
Namely, we show that each construction in the logic corresponds to a closure under some 
operation on recognizable languages: disjunction corresponds to union, negation corresponds 
to complement, existential quantification corresponds to projection, etc.

To simplify the notation, it is sometimes convenient to think of a first-order variable $x$ 
as a second-order variable $X$ interpreted as a singleton set. Therefore, by a slight abuse of 
notation, we shall often write variables in upper-case letters, without explicitly saying whether 
these are first-order or second-order variables (their correct types can be inferred from the 
atoms they appear in). 
%
As usual, given a formula~$\varphi(\bar{X})$ with some free (first-order or monadic second-order) 
variables $X_1,\ldots, X_m$, one can see it as defining the language 
$$
  \intr{\varphi}
  ~=~ \big\{\ang{w,U_1,\dots,U_m} ~:~ w\in (D\times A)^*,~ 
                                      U_1,\ldots,U_m\subseteq\dom(w),~ 
                                      w\models\varphi(U_1,\ldots,U_m)\big\}
$$
where $\ang{w,U_1,\dots,U_m}$ is the data word over $D\times A\times B^m$,
with $B=\{0,1\}$, that associates the letter $(d,a,b_1,\dots,b_m)$ with 
each position $i$ iff $(d,a)$ is the $i$-th letter of $w$, and 
for all~$j=1\dots m$, $b_j$ is~$1$ if~$i\in U_j$, and~$0$ otherwise.

The principle of the proof is to establish that, given a rigidly guarded 
\MSO formula~$\varphi(\bar{X})$, the language~$\intr{\varphi}$ is recognized 
by an orbit-finite data monoid. 
Though this statement is true, it is convenient to strengthen it in order 
to be able to use it as the invariant of an inductive proof based on $\varphi$. 
The problem is that the operation that corresponds to existential quantification (i.e., projection)
transforms an orbit-finite data monoid into a data monoid which is not orbit-finite, in general. 
This is why our induction hypothesis needs to be stronger, namely, need to state that $\intr{\varphi}$ 
is recognized by an orbit-finite data monoid via a {\sl projectable} morphism, as defined just below.

\begin{definition}\label{def:projectable}
Let $h$ be a morphism from the free data monoid $(D\times A\times B^m)^*$ to 
a data monoid $\cM$. We say that $h$ is \emph{projectable} (\emph{over $B^m$}) 
if for all data words $w\in(D\times A)^*$ and all tuples of predicates 
$\bar{U}=(U_1,\ldots,U_m)$ and $\bar{V}=(V_1,\ldots,V_m)$, 
$$
  h(\ang{w,\bar{U}}) ~\orbiteq~ h(\ang{w,\bar{V}}) 
  \qquad\text{implies}\qquad 
  h(\ang{w,\bar{U}}) ~=~ h(\ang{w,\bar{V}})
$$
where $s\orbiteq t$ means that the elements $s$ and $t$ are in the same orbit of $\cM$.
\end{definition}

\noindent
We now state the theorem, which is at the same time our induction hypothesis:

\begin{theorem}\label{thm:formula-to-monoid}
For all rigidly guarded \MSO formulas $\varphi(\bar{X})$, the language $\intr{\varphi}$ 
is effectively recognized by an orbit-finite data monoid via a projectable morphism.
\end{theorem}

From the above theorem we obtain, in particular, the following corollary.

\begin{corollary}\label{cor:sentence-to-monoid}
Every data language definable in rigidly guarded \MSO (resp., rigidly guarded \FO) is 
effectively recognized by an orbit-finite data monoid (resp., aperiodic orbit-finite data monoid).
\end{corollary}

\begin{proof}
The case of rigidly guarded \MSO corresponds just to 
Theorem \ref{thm:formula-to-monoid} in the case of a sentence~$\varphi$. 
The case of rigidly guarded \FO could be proved by establishing the aperiodicity at the 
same time. However, in our case, it is sufficient to recall a result from \cite{data_monoids}
which proves that every data language definable in \FO (non-necessarily rigidly guarded \FO) 
is recognized by an aperiodic data monoid. In particular, if we consider the syntactic data 
monoid of a language definable in rigidly guarded \FO, we easily see that it is aperiodic 
thanks to the result in \cite{data_monoids} and, moreover, has finitely many orbits since
it is the quotient of an orbit-finite data monoid obtained from 
Theorem~\ref{thm:formula-to-monoid}. 
\end{proof}

\medskip
Before entering the details, we give an overview of the proof of Theorem~\ref{thm:formula-to-monoid}, 
which is by structural induction on the rigidly guarded \MSO formulas. The translation of the atomic 
formulas $x<y$, $a(x)$, $x\in Y$ are easy and the translations of the Boolean connectives are as in the classical case.

The translation of the existential closures uses a powerset construction on orbit-finite 
data monoids. We recall that the standard powerset construction returns new elements that
are sets of elements from the original monoid. In general, due to the presence of infinitely 
many elements in a data monoid, the standard powerset construction may remember sets of
unbounded size, possibly resulting in a data monoid that has infinitely many orbits, 
even if the original data monoid has finitely many of them. In our case, however, because 
the morphism is projectable, it is sufficient to apply a variant of the powerset construction 
that remembers at most one element for each orbit of the original monoid, thus producing an 
orbit-finite data monoid.

The most technical part of the proof concerns the translation of the rigidly guarded 
data tests $\varphi(x,y) \:\et\: x\sim y$. The rigidity assumption on the guard $\varphi(x,y)$ 
is crucial: if $\varphi(x,y)$ were not rigid, then the data monoid recognizing 
$\intr{\varphi(x,y) \:\et\: x\sim y}$ would still be orbit-finite, but the morphism 
would in general not be projectable. The proof that $\intr{\varphi(x,y) \:\et\: x\sim y}$ 
is recognized via a projectable morphism requires a bit of analysis, since rigidity is a 
semantic assumption and hence one cannot directly deduce from it a property for the data 
monoid. In this case, we use the rigidity property for ``normalizing'' the data monoid, 
allowing the construction to go through.

In addition, for the translation to be effective, we need to compute the results of 
algebraic operations on orbit-finite data monoids, such as product, projection, and subset
construction. It turns out that all the operations that are needed in the proof are compatible with 
the operation of {\sl finite restriction} that we introduced in Definition \ref{def:restriction}. 
As an example, for every two orbit-finite data monoids $\cM$ and $\cN$ and for every finite subset
$C$ of $D$, we have that the product $\cM\times\cN$ is an orbit-finite data monoid and, moreover, 
$(\cM\times\cN)|_C=\cM|_C \, \times \, \cN|_C$. 
It follows from Proposition \ref{prop:restriction} that we can compute a (finite) representation 
of the result of an algebraic construction starting from some given (finite) representations of 
the input orbit-finite data monoids. 
In view of the above arguments, in the rest of the proof, we shall often skip the details 
about how the representations of the various orbit-finite data monoids are computed and we 
shall focus instead on purely algebraic constructions. In particular, by a slight abuse of 
terminology, we will say that an orbit-finite data monoid $\cN$ is \emph{computed} 
from other orbit-finite data monoids $\cM_1,\ldots,\cM_n$ when a representation of $\cN$ 
can be obtained effectively from some representations of $\cM_1,\ldots,\cM_n$. In a similar 
way, since morphisms from free data monoids to orbit-finite data monoids are uniquely determined 
by the images of the singleton data words, we say that a morphism $g$ can be computed from 
other morphisms $h_1,\ldots,h_n$ when the images via $g$ of all singleton words can be obtained
effectively from the images via $h_1,\ldots,h_n$ of all singleton words (note that there exist
only finitely many images of singleton words up to renamings).

\medskip
We begin by describing the translation of the existential closures of formulas (hereafter,
all formulas are meant to be rigidly guarded \MSO formulas).

\begin{lemma}\label{lemma:projectable-to-powerset}
Let $\psi(\bar{X},X_{m+1})$ be a formula and let 
$\varphi(\bar{X}) = \exists X_{m+1} ~ \psi(\bar{X},X_{m+1})$. 
If $\intr{\psi}$ is recognized by an orbit-finite data monoid $\cM$ 
via a projectable morphism $h:(D\times A\times B^{m+1})^*\then\cM$, 
then one can compute an orbit-finite data monoid $\cN$ and a projectable 
morphism $g:(D\times A\times B^m)^*\then\cN$ recognizing $\intr{\varphi}$.
\end{lemma}

\begin{proof}
For the sake of brevity, we denote by $L$ the language over $D\times A\times B^{m+1}$ 
that is defined by the formula $\psi(\bar{X},X_{m+1})$, and by $\exists L$ the language 
over $D\times A\times B^m$ that is defined by 
$\varphi(\bar{X}) = \exists X_{m+1} ~ \psi(\bar{X},X_{m+1})$. 
We assume that $L$ is recognized by an orbit-finite data monoid $\cM$ via a morphism $h$.
We will apply a variant of the powerset construction to the orbit-finite data monoid 
$\cM$ to obtain an orbit-finite data monoid $\cN$ that recognizes $\exists L$. 
The same construction can be applied to any finite restriction $\cM|_C$ that represents 
$\cM$, so as to compute a restriction $\cN|_C$ that represents $\cN$. We observe, however, 
that the cardinality of the set $C$ must be at least twice the maximal size of the memories 
of the elements of $\cN$.

\smallskip\noindent
{\em The powerset construction.}
Let $\cM=(M,\cdot,\hat\blank)$. Define $\cN=(N,\odot,\check\blank)$ as follows:
\begin{itemize}
  \item the elements of $N$ are the subsets of $M$ that contain only pairwise 
        orbit-distinct elements, namely, those sets $S \subseteq M$ such that 
        for all $s,s'\in S$, $s\orbiteq s'$ implies $s=s'$;
  \item the product $\odot$ is defined on pairs of sets $S,T\in N$ by 
        $$
          \qquad
          S\odot T ~=~ \begin{cases}
                         S\cdot T  & \text{if for all $s,s'\in S$ and $t,t'\in T$,
                                           $s\cdot t \orbiteq s'\cdot t'$ implies 
                                           $s\cdot t=s'\cdot t'$} \\[1ex]
                         \emptyset & \text{otherwise}
                       \end{cases}
        $$
        where $S\cdot T$ denotes the set $\{s\cdot t ~:~ s\in S,\, t\in T\}$;
  \item the function $\check\blank$ maps any renaming $\tau$ to the automorphism 
        $\check\tau$ such that, for all $S\in M'$, 
        $$
          \check\tau(S) ~=~ \{\hat\tau(s) ~:~ s\in S\} \ .
        $$
\end{itemize}
It is routine to check that the product $\odot$ is associative, the function
$\check\blank$ is a group action, the empty set $\emptyset$ is a null element 
in $\cN$, and the singleton $\{1_\cM\}$ is the identity element in $\cN$. 

Below, we verify that the data monoid $\cN$ is orbit-finite.
Let $n$ be the number of orbits of $\cM$. We observe 
that every set $S\in N$ has cardinality at most $n$
(indeed, if this were not the case, then $S$ would contain two distinct 
elements $s$ and $t$ such that $s\orbiteq t$, which would contradict 
the definition of $N$). 
From this property it follows that $\cN$ is the projection of 
$\cM^{\le n}$ under some equivariant mapping, where 
$\cM^{\le n} = \biguplus_{i\le n} \cM^i$ and each 
$\cM^i$ is the $i$-fold product of $\cM$ with itself.
Because orbit-finite sets are closed under products, finite disjoint unions, 
and images under equivariant mappings, we have that $\cN$ is orbit-finite.

\medskip\noindent
{\it The morphism.}
We now define a morphism $g$ from the the free data monoid $(D\times A\times B^m)^*$ 
to the orbit-finite data monoid $\cN$. For every expanded data word 
$\ang{w,U_1,\ldots,U_m}$, we let 
$$
  g(\ang{w,U_1,\ldots,U_m}) ~\eqdef~ 
  \big\{h(\ang{w,U_1,\ldots,U_m,U_{m+1}}) ~:~ U_{m+1}\subseteq\dom(w)\big\}
$$
(note that, since $h$ is projectable, then $g(\ang{w,U_1,\ldots,U_m})$ contains only 
pairwise orbit-distinct elements and hence it is an element of the data monoid $\cN$).

We verify that the morphism $g$ is projectable. Consider a data word $w$ and some 
tuples of predicates $\bar{U}=U_1,\ldots,U_m$ and $\bar{V}=V_1,\ldots,V_m$, and 
suppose that $g(\ang{w,\bar{U}}) \:\orbiteq\: g(\ang{w,\bar{V}})$. 
This means that there is a renaming $\tau$ such that 
$$
  g(\ang{w,\bar{V}}) ~=~ \check\tau\big(g(\ang{w,\bar{U}})\big).
$$
Moreover, by definition of $g$, we have
$$
  \big\{h(\ang{w,\bar{V},V_{m+1}}) ~:~ V_{m+1}\subseteq\dom(w)\big\} 
  ~=~ 
  \big\{\hat\tau(h(\ang{w,\bar{U},U_{m+1}})) ~:~ U_{m+1}\subseteq\dom(w)\big\}.
$$
Since $h$ is projectable, we have that the two sets $g(\ang{w,\bar{U}})$ and 
$g(\ang{w,\bar{V}})$ coincide, which proves that $g$ is projectable as well.

\medskip\noindent
{\it Recognizability.}
It remains to prove that the language $\exists L$ is recognized by $\cN$ via 
the morphism $g$. For the sake of brevity, let $F=h(L)$ and $G=g(\exists L)$. 
We consider an expanded data word $\ang{w,\bar{U}}\in (D\times A\times B^m)^*$
and we prove that $\ang{w,\bar{U}}\in \exists L$ iff $g(\ang{w,\bar{U}})\in G$. 
The left-to-right implication is trivial, so we prove the converse implication. 
Suppose that $g(\ang{w,\bar{U}})\in G$. Since $G=g(\exists L)$, we know that 
there is an expanded data word $\ang{w',\bar{V}}\in \exists L$ such that 
$g(\ang{w,\bar{U}}) \:=\: g(\ang{w',\bar{V}})$. From the definition of $g$ 
we also know that 
$$
  \big\{h(\ang{w,\bar{U},U_{m+1}}) ~:~ U_{m+1}\subseteq\dom(w)\big\} ~=~ \big\{h(\ang{w',\bar{V},V_{m+1}}) ~:~ V_{m+1}\subseteq\dom(w')\big\} \ .
$$
Moreover, from the definition of $\exists L$ we know that 
$(\ang{w',\bar{V},V_{m+1}})\in L$ for some unary predicate $V_{m+1}\subseteq\dom(w')$.
Finally, since $L$ is recognized by $\cM$ via the morphism $h$ and since
$\ang{w',\bar{V},V_{m+1}}$ belongs to $L$, we have $h(\ang{w',\bar{V},V_{m+1}})\in F$ 
and hence $h(\ang{w,\bar{U},U_{m+1}})\in F$ for some unary predicate $U_{m+1}\subseteq\dom(w)$. 
This shows that $\ang{w,\bar{U},U_{m+1}}\in L$ and hence $\ang{w,\bar{U}}\in \exists L$. 
\end{proof}

We now turn to the translation of rigidly guarded data tests.

\begin{lemma}\label{lemma:reduced-to-projectable}
Given a rigid formula $\varphi(x,y)$, an orbit-finite data monoid $\cM$ and a 
projectable morphism $h$ that recognizes $\intr{\varphi}$, one can compute an 
orbit-finite data monoid $\cM'$ and a projectable morphism $h'$ that recognizes 
$\intr{\varphi(x,y)\et x\sim y}$. 
\end{lemma}

\begin{proof}
Let $\cM$ be an orbit-finite data monoid and let 
$h:(D\times A\times B^2)^*\then\cM$ be a projectable morphism that recognizes 
$L=\intr{\varphi(x,y)}$. We first show that the image via $h$ of the free data 
monoid $(D\times A\times B^2)^*$, which is a data sub-monoid of $\cM$, can be 
computed from $\cM$ and $h$:
       
\begin{claim}
From the orbit-finite data monoid $\cM$ and the morphism 
$h:(D\times A\times B^2)^*\then\cM$, one can compute the 
data sub-monoid $h((D\times A\times B^2)^*)$. 
\end{claim}
        
\begin{proof}[Proof of claim]
Suppose that the orbit-finite data monoid $\cM$ is represented by 
its restriction $\cM|_C$, for some finite subset $C$ of $D$ such 
that $\len{C}\ge 2\dlen{\cM}$.
Let $\cM'=h\big((D\times A\times B^2)^*\big)$ be the data 
sub-monoid induced by $h$. 
Clearly, we have $\dlen{\cM'}\le\dlen{\cM}$ and hence, 
by Proposition \ref{prop:restriction}, the data sub-monoid 
$\cM'$ is uniquely determined by its restriction $\cM'|_C$. 
Moreover, the domain of $\cM'|_C$ is the finite set 
$h\big((C\times A\times B^2)^*\big)$, which is computable from 
$\cM|_C$ and $h|_{C\times A\times B^2}$. Finally, the product 
and the group action of the data sub-monoid $\cM'|_C$ are the 
restrictions of the product and the group action of $\cM$ 
to the finite set $h\big((C\times A\times B^2)^*\big)$. 
This shows that $\cM'|_C$ can be computed from $\cM|_C$ 
and $h|_{C\times A\times B^2}$.
\end{proof}             

Thanks to above claim, we can assume, without loss of generality, that $h$ is
a surjective morphism. Unfortunately, even under this assumption, the property 
of projectability is not straightforwardly preserved when we translate the morphism 
$h$ for the rigid guard $\varphi(x,y)$ to a morphism for the rigidly guarded comparison 
$\varphi(x,y) \:\et\: x \sim y$. For this, we must derive from the rigidity assumption 
on $\varphi(x,y)$ a stronger notion of projectability, which is defined below and which 
is called $0$-reduced projectability.

An element $s$ of a data monoid $\cN$ is a \emph{null} if 
$s\cdot t = t\cdot s = s$ for all elements $t$ of $\cN$. 
If a data monoid has a null element, then this element 
is unique, and in this case it is denoted by $0_\cN$.
Moreover, it is easy to see that if a language $L$ is 
recognized by an orbit-finite data monoid, then $L$ is 
also recognized by an orbit-finite data monoid with a 
null element. 
        
\begin{definition}\label{def:0-reduced}
Let $h$ be a morphism from the free data monoid $(D\times A\times B^2)^*$ 
to a data monoid $\cN$ with null element $0_{\cN}$. 
We say that $h$ is \emph{$0$-reduced} if for all data words $w\in(D\times A)^*$ 
and positions $x,x',y,y'\in\dom(w)$, the following implications hold:
\begin{itemize}
  \item if $h(\ang{w,\!\{x\},\emptyset}) = h(\ang{w,\!\{x'\},\emptyset})$,
        then $x \:=\: x'$ or $0_\cN = h(\ang{w,\{x\},\emptyset}) = h(\ang{w,\!\{x'\},\emptyset})$
  \item if $h(\ang{w,\emptyset,\!\{y\}}) = h(\ang{w,\emptyset,\!\{y'\}})$, 
        then $y \:=\: y'$ or $0_\cN = h(\ang{w,\emptyset,\!\{y\}}) = h(\ang{w,\emptyset,\!\{y'\}})$.
\end{itemize}
\end{definition}

Below, we show that the data language $L$ is equally recognized by an orbit-finite
data monoid with null element and a morphism that is surjective, projectable, and 
$0$-reduced (for simplicity, we call it a $0$-reduced projectable morphism).

\begin{claim}
From the orbit-finite data monoid $\cM$ and the projectable surjective morphism 
$h:(D\times A\times B^2)^* \then \cM$ recognizing the language $L=\intr{\varphi(x,y)}$
of the rigid formula $\varphi(x,y)$, one can compute an orbit-finite data monoid 
$\cN$ with null element $0_{\cN}$ and a $0$-reduced projectable morphism 
$g:(D\times A\times B^2)^* \then \cN$ recognizing the same language $L$.
\end{claim}

\begin{proof}[Proof of claim]
The desired orbit-finite data monoid $\cN$ is obtained from a suitable quotient 
of $\cM$, precisely, by collapsing those elements of $\cM$ that do not represent 
factors of data words in $L$. As usual, the construction can be applied effectively
to a restriction $\cM|_C$ that represents $\cM$, thus obtaining a representation 
$\cN|_C$ of $\cN$.

\smallskip\noindent
{\it Collapsing bad elements.}
Let $F=h(L)$ and let $G$ be the maximal set of all elements such that 
$M \cdot G \cdot M \:\cap\: F \,=\, \emptyset$. Intuitively, $G$ contains 
those elements of $\cM$ that cannot be extended to elements in $F$ by 
concatenating elements to the left, to the right, or both. Note that 
$G$ is an ideal of $\cM$, namely, $M\,\cdot\, G \,\cdot\, M \,\subseteq\, G$, 
and, furthermore, it is closed under the action of renamings, namely, 
$\tau(G) \subseteq G$ for all renamings $\tau$. 
We now introduce the equivalence $\dsim_G$ that groups 
any two elements $s,t\in M$ whenever we have either $s=t$ or $s,t\in G$. 
Note that $\dsim_G$ is a congruence with respect to the product of $\cM$, 
namely, if $s \:\dsim_G\: t$ and $u \:\dsim_G\: v$, then 
$s\cdot u \:\dsim_G\: t\cdot v$. 
The equivalence $\dsim_G$ is also compatible with the action of renamings, 
namely, if $s\dsim_G t$, then $\tau(s) \:\dsim_G\: \tau(t)$ for all 
renamings $\tau$. 
This allows us to define $\cN$ as the quotient of $\cM$ with respect to 
$\dsim_G$, where the elements are the $\dsim_G$-equivalence classes, the 
product is defined by
$$
  [s]_{\dsim_G} \,\odot\, [t]_{\dsim_G} ~\eqdef~ [s\cdot t]_{\dsim_G}
$$
and the action of renamings is defined by
$$
  \tau([s]_{\dsim_G}) ~\eqdef~ [\tau(s)]_{\dsim_G}
$$
(note that the above functions are well defined).

Clearly $\cN$ is an orbit-finite data monoid. Moreover, for all $s\in M\setminus G$, 
the $\dsim_G$-equivalence class of $s$ is the singleton $\{s\}$. 
The only other element of $\cN$ is the entire set $G$, which is also 
the null element, and is thus denoted by $0_{\cN}$.

\smallskip\noindent
{\it The morphism.}
We now define the morphism $g:(D\times A\times B^2)^* \then \cN$ that recognizes $L$. 
This is nothing but the functional composition $h_G\circ h$ of the morphism $h$ from 
$(D\times A\times B^2)^*$ to $\cM$ and the morphism $h_G$ from $\cM$ to $\cN$ defined by 
$$
  h_G(s) ~\eqdef~ [s]_{\dsim_G} \ .
$$
We recall that $h$ and $h_G$ are surjective morphisms, so $h$ is also surjective. 
Moreover, since $h_G^{-1}\circ h_G$ is the identity on $F=h(L)$, we have 
$$
  L ~=~ h^{-1}(h(L)) ~=~ h^{-1}(F) ~=~ h^{-1}(h_G^{-1}(h_G(F))) ~=~ g^{-1}(g(L)) \ .
$$
This shows that $g$ is a surjective morphism recognizing the data language $L$.

\smallskip\noindent
{\it Projectability.}
Next, we verify that the morphism $g$ is projectable. 
Consider a data word $w\in (D\times A)^*$ and some 
predicates $U_1,U_2,V_1,V_2\subseteq\dom(w)$ and suppose 
that $g(\ang{w,U_1,U_2})$ and $g(\ang{w,V_1,V_2})$ are 
in the same orbit, namely, that there is a renaming $\tau$ 
such that $g(\ang{w,V_1,V_2}) \:=\: \tau(g(\ang{w,U_1,U_2}))$. 
We distinguish two cases depending on whether one among 
the two elements $g(\ang{w,U_1,U_2})$ 
and $g(\ang{w,V_1,V_2})$ coincides with $0_{\cN}$ or not. 
If $g(\ang{w,U_1,U_2}) \:=\: 0_{\cN}$, 
then we recall that $0_{\cN}$ has empty memory and hence we obtain
$$
  g(\ang{w,V_1,V_2}) ~=~ \tau(g(\ang{w,U_1,U_2})) 
                     ~=~ \tau(0_{\cN}) 
                     ~=~ 0_{\cN} 
                     ~=~ g(\ang{w,U_1,U_2}) \ .
$$
A similar conclusion can be obtained when $g(\ang{w,V_1,V_2}) \:=\: 0_{\cN}$. 

In the remaining case, we assume that neither $g(\ang{w,U_1,U_2})$ 
nor $g(\ang{w,V_1,V_2})$ are the null element. 
We know from the definition of $\cN$ that neither 
$h(\ang{w,U_1,U_2})$ nor $h(\ang{w,V_1,V_2})$ belong 
to the ideal $G$ and hence 
$g(\ang{w,U_1,U_2}) \:=\: \big\{h(\ang{w,U_1,U_2})\big\}$ and  
$g(\ang{w,V_1,V_2}) \:=\: \big\{h(\ang{w,V_1,V_2})\big\}$. 
Moreover, we have
$$
  g(\ang{w,V_1,V_2}) ~=~ 
  \tau(g(\ang{w,U_1,U_2})) ~=~ 
  \big\{\tau(h(\ang{w,U_1,U_2}))\big\}
$$
and hence $h(\ang{w,V_1,V_2}) \:=\: \hat\tau(\ang{w,U_1,U_2})$. 
Finally, since $h$ is projectable, we obtain 
$h(\ang{w,U_1,U_2}) \:=\: h(\ang{w,V_1,V_2})$ 
and therefore $g(\ang{w,U_1,U_2}) \:=\: g(\ang{w,V_1,V_2})$. 
This shows that $g$ is projectable.

\smallskip\noindent
{\it $0$-Reduced.}
It remains to prove that $g$ is $0$-reduced. 
Here, we exploit the fact that the language $L$ is defined 
by a {\sl rigid} formula $\varphi(x,y)$. Let $w\in (D\times A)^*$ 
be a data word and let $x,x'\in\dom(w)$ be two positions in it. 
By way of contradiction, assume that $x\neq x'$ and 
$g(\ang{w,\{x\},\emptyset}) \:=\: g(\ang{w,\{x'\},\emptyset}) \:\neq\: 0_{\cN}$.
We need to derive that $\varphi(x,y)$ is not rigid (the same conclusion 
can be obtained from the assumption that there exist two positions 
$y,y\in\dom(w)$ such that $y\neq y'$ and 
$g(\ang{w,\emptyset,\{y\}}) \:=\: g(\ang{w,\emptyset,\{y'\}}) \:\neq\: 0_{\cN}$). 
Since $g(\ang{w,\{x\},\emptyset}) \:=\: g(\ang{w,\{x'\},\emptyset}) \:\neq\: 0_{\cN}$, 
we know that $h(\ang{w,\{x\},\emptyset}) = h(\ang{w,\{x'\},\emptyset}) \nin G$ and hence 
there exist $s,t\in M$ such that 
$s\cdot h(\ang{w,\{x\},\emptyset})\cdot t = s\cdot h(\ang{w,\{x'\},\emptyset})\cdot t\in F$.
Moreover, since $h$ is surjective, we know that there exist two expanded data words
$\ang{w',U_1,U_2}$ and $\ang{w'',V_1,V_2}$ such that $h(\ang{w',U_1,U_2})=s$ and 
$h(\ang{w'',V_1,V_2})=t$.
Since $\cM$ and $h$ recognize the language $L=\intr{\varphi(x,y)}$ and $F=h(L)$, 
we have that $\varphi(x,y)$ is satisfied both by the data word 
$\ang{w',U_1,U_2}~\ang{w,\{x\},\emptyset}~\ang{w'',V_1,V_2}$ and by
the data word $\ang{w',U_1,U_2}~\ang{w,\{x'\},\emptyset}~\ang{w'',V_1,V_2}$.
Finally, since $x\neq x'$, we conclude that $\varphi(x,y)$ is not rigid. 
This completes the proof of our claim.
\end{proof}

\smallskip
We can now turn back to the proof of Lemma \ref{lemma:reduced-to-projectable}. 
By the above claim, we can assume that the language $L=\intr{\varphi(x,y)}$ 
is recognized by an orbit-finite data monoid $\cM_1$ with null element $0_{\cM_1}$ 
via a $0$-reduced projectable morphism $h_1:(D\times A\times B^2)^*\then\cM_1$.
Moreover, we can construct the syntactic data monoid $\cM_2$ of the language 
defined by $x\sim y$ and the corresponding morphism 
$h_2:(D\times A\times B^2)^*\then\cM_2$ recognizing $\intr{x\sim y}$.
We observe that the data monoid $\cM_2$ has finitely many orbits and 
its elements can be identified with terms of one the following forms:
\begin{enumerate}
  \item $o(\emptystr)$, 
        which plays the role of the identity in $\cM_2$ and corresponds 
        to the image under $h_2$ of the empty data word;
  \item $p(d)$, for any $d\in D$, which corresponds to the image under $h_2$ 
        of data words $w$ expanded with a singleton predicate $U=\{x\}$, 
        where $w(x)=d$, and with the empty predicate $V=\emptyset$;
  \item $q(d)$, for any $d\in D$, which corresponds to the image under 
        $h_2$ of data words expanded with the empty predicate $U=\emptyset$ 
        and with a singleton predicate $V=\{y\}$, where $w(y)=d$; 
  \item $r(\emptystr)$, with corresponds to the image under $h_2$ of 
        the expanded data words that satisfy $x\sim y$;
  \item $s(\emptystr)$, which plays the role of the null element in $\cM_2$ 
        and corresponds to the image under $h_2$ of data words expanded with 
        two non-empty predicates $U,V$ that do not satisfy $x\sim y$.
\end{enumerate}
For example, we have $p(d)\odot q(d)=r(\emptystr)$ and 
$p(d)\odot q(e)=s(\emptystr)$, for all pairs of distinct 
values $d,e\in D$. We also observe that the morphism $h_2$ 
is not projectable, which explains why, in order to recognize 
the intersection of the data languages $L=\intr{\varphi(x,y)}$ 
and $\intr{x\sim y}$, we introduce below a variant of the product 
of data monoids.

\smallskip\noindent
{\it The $0$-collapse product.}
The orbit-finite data monoid $\cM'$ for the formula 
$\varphi(x,y) \:\et\: x\sim y$ is defined using a 
suitable variant of the product of data monoids with null elements, 
which we call \emph{$0$-collapse product} 
(strictly speaking, the $0$-collapse product 
is a special form of semi-direct product). 
Formally, let $\cM_1=(M_1,\cdot,\hat\blank)$ and $\cM_2=(M_2,\odot,\check\blank)$.
We define $\cM'=(M',\circledcirc,\tilde\blank)$, where
\begin{itemize}
  \item $M'$ consists of all pairs $(s_1,s_2)\in M_1\times M_2$ 
        such that $s_1=0_{\cM_1}$ implies $s_2=0_{\cM_2}$;
  \item for every $(s_1,s_2),(t_1,t_2)\in M'$, 
        the product $(s_1,s_2)\circledcirc(t_1,t_2)$ is either the
        pair $(s_1\cdot t_1,s_2\odot t_2)$ or the pair 
        $(0_{\cM_1},0_{\cM_2})$, depending on whether 
        $s_1\cdot t_1 \neq 0_{\cM_1}$ or $s_1\cdot t_1 = 0_{\cM_1}$;
  \item $\tilde\tau(s_1,s_2)=(\hat\tau(s_1),\check\tau(s_2))$ 
        for all all renamings $\tau$ and all $(s_1,s_2)\in M'$.
\end{itemize}
Clearly, $\cM'$ is an orbit-finite data monoid.

\smallskip\noindent
{\it The morphism.}
Accordingly, we define the morphism $h'$ that maps any expanded data word 
$w\in (D\times A\times B^2)^*$ either to the pair $(h_1(w),h_2(w))$ or 
to the pair $(0_{\cM_1},0_{\cM_2})$, depending on whether 
$h_1(w) \neq 0_{\cM_1}$ or $h_1(w) = 0_{\cM_1}$. 
Clearly, $h'$ recognizes the language $\intr{\varphi(x,y) \:\et\: x\sim y}$.

\smallskip\noindent
{\it Projectability.}
It remains to prove that the morphism $h'$ is projectable. Consider a data word 
$w\in(D\times A)^*$ and some predicates $U_1,U_2,V_1,V_2\subseteq\dom(w)$, and 
suppose that the elements $h'(\ang{w,U_1,U_2})$ and $h'(\ang{w,V_1,V_2})$ are in 
the same orbit. We distinguish between the case where 
$h_1(\ang{w,U_1,U_2}) = 0_{\cM_1}$ (and hence $h_1(\ang{w,V_1,V_2}) = 0_{\cM_1}$ 
as well) and the case where $h_1(\ang{w,U_1,U_2})\neq 0_{\cM_1}$ 
(and hence $h_1(\ang{w,V_1,V_2})\neq 0_{\cM_1}$ as well). 
In the former case, we immediately get 
$$
  h'(\ang{w,U_1,U_2}) ~=~ (0_{\cM_1},0_{\cM_2}) ~=~ h'(\ang{w,V_1,V_2}) \ .
$$
In the latter case, we have 
$h'(\ang{w,U_1,U_2}) \,=\, \big(h_1(\ang{w,U_1,U_2}),h_2(\ang{w,U_1,U_2})\big)$ and 
$h'(\ang{w,V_1,V_2}) \,=\, \big(h_1(\ang{w,V_1,V_2}),h_2(\ang{w,V_1,V_2})\big)$.
From $h'(\ang{w,U_1,U_2}) \,\orbiteq\, h'(\ang{w,V_1,V_2})$, 
we obtain 
$h_1(\ang{w,U_1,U_2}) \,\orbiteq\, h_1(\ang{w,V_1,V_2})$ and 
$h_2(\ang{w,U_1,U_2}) \,\orbiteq\, h_2(\ang{w,V_1,V_2})$.
Moreover, since $h_1$ is projectable, we know that 
$h_1(\ang{w,U_1,U_2}) \,=\, h_1(\ang{w,V_1,V_2})$. 
It remains to prove that $h_2(\ang{w,U_1,U_2}) \,=\, h_2(\ang{w,V_1,V_2})$. 
To do so, we distinguish between the following subcases:
\begin{enumerate}
  \item $U_1=U_2=\emptyset$. 
        We have $h_2(\ang{w,U_1,U_2}) = 1_{\cM_2}$ and hence, since 
        $1_{\cM_2}$ has empty memory and 
        $h_2(\ang{w,U_1,U_2}) \,\orbiteq\, h_2(\ang{w,V_1,V_2})$, 
        we get $h_2(\ang{w,V_1,V_2}) = 1_{\cM_2}$.
  \item Both $U_1$ and $U_2$ are non-empty. 
        In this case $h_2(\ang{w,U_1,U_2})$ must be either 
        the null element $0_{\cM_2}$ or the term $r(\emptystr)$ 
        (recall that this term represents all expanded data 
        words that satisfy $x\sim y$). 
        Both elements have empty memory and hence from 
        $h_2(\ang{w,U_1,U_2}) \,\orbiteq\, h_2(\ang{w,V_1,V_2})$
        we get
        $h_2(\ang{w,U_1,U_2}) \,=\, h_2(\ang{w,V_1,V_2})$.
  \item $U_1\neq\emptyset$ and $U_2=\emptyset$. 
        Clearly, $U_1$ is a singleton of the form $\{x\}$. Similarly, since 
        $h_2(\ang{w,U_1,U_2}) \,\orbiteq\, h_2(\ang{w,V_1,V_2})$, we have 
        that $V_1$ is a singleton of the form $\{x'\}$ and $V_2=\emptyset$. 
        We also recall that 
        $h_1(\ang{w,U_1,U_2}) \,=\, h_1(\ang{w,V_1,V_2}) \,\neq\, 0_{\cM_1}$ 
        and that the morphism $h_1$ is $0$-reduced, which implies $x=x'$. 
        This shows that $h_2(\ang{w,U_1,U_2}) \,=\, h_2(\ang{w,V_1,V_2})$.
  \item $U_1=\emptyset$ and $U_2\neq\emptyset$. 
        This case is symmetric to the previous one.
\end{enumerate}
We have just shown that $h'$ is a projectable morphism 
recognizing $\intr{\varphi(x,y) \et x\sim y}$.
\end{proof}\smallskip\enlargethispage{\baselineskip}

We are now ready to prove the main theorem of this section, that is, that every
language $\intr{\varphi}$ defined by a rigidly guarded \MSO formula $\varphi(\bar{X})$
is effectively recognized by an orbit-finite data monoid via a projectable morphism.

\begin{proof}[Proof of Theorem \ref{thm:formula-to-monoid}]
As already mentioned, the proof is by induction on the structure of the rigidly guarded \MSO 
formula $\varphi(\bar{X})$. As for the base cases, the languages defined by the atomic formulas 
$x<y$, $a(x)$, and $x\in Y$ are clearly recognized by orbit-finite data monoids via projectable 
morphisms. 

As for the inductive step, suppose that a formula $\varphi$ with $m$ free variables 
$X_1,\ldots,X_m$ is given and that one can compute an orbit-finite data monoid $\cM$ 
and a projectable morphism $h:(D\times A\times B^m)^*\then\cM$ recognizing the 
language defined by $\varphi$. It follows that the complement language defined by 
$\neg\varphi$ is recognized by the same orbit-finite data monoid $\cM$ via the 
same projectable morphism $h$. 

Similarly, for the disjunction of two formulas, suppose that $\varphi_1$ and $\varphi_2$
are given. Without loss of generality (namely, by introducing dummy free variables
via cylindrification), we can assume that the two formulas $\varphi_1$ and $\varphi_2$ 
have the same free variables $X_1,\ldots,X_m$. Furthermore, suppose that one can compute 
two orbit-finite data monoids $\cM_1$ and $\cM_2$ and two projectable morphisms 
$h_1:(D\times A\times B^m)^*\then\cM_1$ and $h_2:(D\times A\times B^m)^*\then\cM_2$
recognizing the languages defined by $\varphi_1$ and $\varphi_2$. As these languages
are over the same alphabet, we can construct an orbit-finite data monoid 
$\cM_1\times\cM_2$ and a projectable morphism $h_1\times h_2$ that recognize
the language defined by $\varphi_1 \vel \varphi_2$.

As for the existential closure, Lemma \ref{lemma:projectable-to-powerset} implies that 
the language defined by the formula $\exists X_m ~ \varphi$ is recognized by a suitable
orbit-finite data monoid $\cN$ via a projectable morphism $g$, both computable from $\cM$ 
and $h$. 

Finally, if $m=2$ and $\varphi(x_1,x_2)$ is a rigid formula, then we know from 
Lemma \ref{lemma:reduced-to-projectable} how to compute an orbit-finite data monoid $\cN$ 
and a projectable morphism $g$ that recognizes the language defined by 
$\varphi(x_1,x_2) \:\et\: x_1\sim x_2$. This concludes the proof of the theorem.
\end{proof}

\section{From orbit-finite monoids to rigidly guarded \MSO}\label{sec:monoids2logic}

Having shown that every language defined by a rigidly guarded \MSO (resp., \FO) 
sentence is recognized by an orbit-finite data monoid (resp., by an aperiodic 
orbit-finite data monoid), we now prove the converse. This is the most technical 
result of the paper. 

We remark that in the classical theory of regular languages, the analogous of 
this result (at least the part involving only MSO) is straightforward: indeed, 
a monoid can be used as an automaton, and in this case it is sufficient to write 
an MSO formula that guesses a run of such an automaton and checks that it is valid 
and accepting. We cannot use such an approach with data monoids: not only there is 
no equivalent automaton model, but furthermore, the above approach is intrinsically 
not compatible with the notion of rigidity. Another consequence is that, as 
opposed to the classical case, the proof is significantly more involved for 
rigidly guarded \MSO than for rigidly guarded \FO.

We also recall that data languages are invariant under renamings. This means
that every data language that is recognized by an orbit-finite data monoid $\cM$ 
via a morphism $h$ can be described as the union over some orbits $o$ of $\cM$ 
of the inverse images $h^{-1}(o)$. The result we aim to prove is thus the following:

\begin{theorem}\label{thm:monoids-to-logic}
Given an orbit-finite data monoid $\cM$, a morphism $h$ from a free data monoid 
to $\cM$, and an orbit $o$ of $\cM$, one can compute a rigidly guarded \MSO sentence 
$\varphi$ that defines the data language $h^{-1}(o)$. Moreover, if $\cM$ is aperiodic, 
then $\varphi$ is a rigidly guarded \FO sentence.
\end{theorem}

The proof of the theorem follows a structure similar to Sch\"utzenberger's proof 
that languages recognized by aperiodic monoids are definable by star-free expressions 
(i.e., in FO logic). The objective of our proof is to find suitable formulas that, 
given some positions $x\le y$ in a data word~$w$, determine the orbit of the
image via $h$ of the infix of $w$ between $x$ and $y$, that is, determine 
the monoid element $h(w[x,y])$. 
We will reach this objective by exploiting an induction on a well-founded partial order 
that is defined on the $=_{\orbit\cJ}$-classes of $\cM$ and that is induced by the preorder 
$\le_\cJ$ (refer to Section \ref{subsec:green-theory} for an account of these orders).

Roughly speaking, we first construct the desired formulas for shorter infixes of the 
data word and then we move up towards longer infixes, until we determine the orbit of 
the entire word. 
To do so, we need to be able to compute the orbit of an infix $w[x,y]$ on the basis 
of some {\sl bounded} amount of information related to some factors of it, e.g.,
$w[x,z]$ and $w[z+1,y]$, for some $z$ between $x$ and $y$. 
Moreover, since the product of two elements depends not only on the orbits, but 
also on the memorable values, we need to be able to compute the latter as well.
Here, by ``computing the memorable values'' of $w[x,y]$ we mean being able to 
locate some positions in $w[x,y]$ that carry the memorable values of the element $h(w[x,y])$.
For this, we use formulas of the form $\varphi(x,y,z_1,\ldots,z_n)$ which determine,
not only the orbit of $h(w[x,y])$, but also some positions $z_1,\ldots,z_n$ witnessing 
the memorable values. This must be done with care, however, as in our logic positions 
with memorable values can be compared only if they are guarded by rigid formulas.

We tacitly assume that all formulas defined hereafter are either rigidly guarded \FO 
formulas or rigidly guarded \MSO formulas, depending on whether $\cM$ is aperiodic or not.

\medskip
We begin by generalizing the notion of rigidity to formulas with more than two variables. 

\begin{definition}\label{def:inward-rigid}
We say that \emph{$x_i$ determines $x_j$} in a formula $\varphi(x_1,\ldots,x_n)$ 
if for all data words $w$ and all positions~$x_1,\ldots,x_n$, $x'_1,\ldots,x'_n$ $\in \dom(w)$,
$$
\begin{cases}
  w \models \varphi(x_1,\ldots,x_n) \\[0.5ex] 
  w \models \varphi(x'_1,\ldots,x'_n) \\[0.5ex] 
  x_i = x'_i
\end{cases}
\quad\text{implies}\qquad
x_j = x'_j
$$
The formula $\varphi(x_1,\ldots,x_n)$ is \emph{rigid} 
if $x_i$ determines $x_j$ for all~$i,j\in\{1,\ldots,n\}$. 
Similarly, a formula~$\varphi(x,z_1,\ldots,z_k,y)$ is \emph{inward-rigid} if 
$w\models\varphi(x,z_1,\ldots,z_k,y)$ implies $x\le z_1,\ldots,z_k\leq y$,
and, in addition, $x$ determines $z_1,\ldots,z_k,y$ and $y$ determines 
$x,z_1,\ldots,z_k$ in it.
\end{definition}

We will mainly work with formulas $\varphi(x,z_1,\ldots,z_k,y)$ that are inward-rigid. 
Under certain conditions, these formulas can be used to determine the orbit and the positions 
of some memorable values of factors of a data word (we will describe formally what 
this means in Definition \ref{def:computing-types}). However, in order 
to compare memorable values and simulate products in $\cM$, we need to be 
able to turn an inward-rigid formula into a fully rigid formula, or a
finite disjunction of such formulas. The following crucial lemma shows how to do so.
We remark that the lemma can be read with formulas meaning either rigidly 
guarded \MSO formulas or rigidly guarded \FO formulas: both results hold, 
and the proof is in fact the same.

\begin{lemma}[Sub-definability]\label{lemma:subdefinability}
For all formulas~$\varphi(x,y)$ where $x$ determines $y$,
there exist finitely many formulas~$\beta_i(z,y)$ where 
$z$ determines $y$ such that for all~$x\leq z\leq y$, 
$$
  w \models \varphi(x,y) 
  \qquad\text{implies}\qquad 
  w \models \beta_i(z,y) \text{ for some $i$}.
$$
\end{lemma}

\begin{proof}
We begin by recalling a result that originates from the composition methods 
developed by Feferman-Vaught and Shelah \cite{feferman_vaught_theorem,composition_method_shelah}:

\begin{claim}\label{claim:composition}
Given a classical MSO / FO formula $\varphi(x,y)$ that only uses the order $<$ and 
some unary predicates, but no data tests, and that entails $x\le y$, there exist 
finitely many pairs of formulas 
$(\varphi^L_i(x),\varphi^R_i(y))_{i=1\ldots n}$ such that, for all words $w=u\,v$ and
all positions $x$ in $u$ and $y$ in $v$,
$$
  u\,v\sat\varphi(x,|u|+y)
  \qquad\text{iff}\qquad 
  u\sat\varphi^L_i(x)
  \text{ and } 
  v\sat\varphi^R_i(y)
  \text{ for some $i\in\{1,\ldots,n\}$}.
$$
\end{claim}

\noindent
By relativising quantifications, we can then obtain formulas in two variables
$\alpha_i(x,z)$, $\beta_i(z,y)$ such that
$$
\begin{array}{rcll}
  w\sat\alpha_i(x,z) &\text{iff}&  w[1,\ldots,z-1]\sat\varphi^L_i(x) & \qquad\text{and} \\[1ex]
  w\sat\beta_i(z,y)  &\text{iff}&  w[z,\ldots,y]\sat\varphi^R_i(y).
\end{array}
$$
Our sub-definability lemma for the classical MSO / FO formula $\varphi(x,y)$
follows easily from the above result, since $\varphi(x,y)$ is equivalent to 
$$
  \bigvee_{i=1\ldots n}
  \exists z ~ \alpha_i(x,z) ~\et~ \beta_i(z,y) ~\et~ x\le z\le y.
$$
Below, we generalize this argument to formulas that use rigidly guarded data tests.

Let $\varphi(x,y)$ be a formula of rigidly guarded \MSO / \FO.
We use a technique similar to that of the proof of Theorem \ref{thm:decidability} to 
syntactically replace in $\varphi$ every occurrence of a data test $x'\sim y'$ with 
a fresh unary predicate $c^\sim_\alpha(x')$, where $\alpha(x',y')$ is the rigid formula 
that guards the occurrence of $x'\sim y'$ in $\varphi$ and $c^\sim_\alpha(x')$ encodes 
the existence of a (unique) position $y'$ satisfying $\alpha(x',y') \:\et\: x'\sim y'$.
We denote by $\varphi^-(x,y)$ the resulting formula of classical MSO / FO and, for every
data word $w$, we denote by $w^-$ the word obtained from~$w$ by removing all data values 
and by adding the predicates~$c^\sim_\alpha$ at positions $x'$ in such a way that 
$$
  w^- \sat c^\sim_\alpha(x')
  \qquad\text{iff}\qquad
  w \sat \exists y'~ \alpha(x',y') ~\et~ x'\sim y' \ .
$$
Clearly, for all data words $w$ and all positions $x,y$ in it, we have 
$$
  w \sat \varphi(x,y)
  \qquad\text{iff}\qquad
  w^-\sat\varphi^-(x,y) \ .
$$
Now, suppose that $x$ determines $y$ in $\varphi(x,y)$.
It can happen that $x$ does not determine~$y$ in~$\varphi^-(x,y)$, since the unary predicates 
$c^\sim_\alpha$ could be chosen in a way that is inconsistent with any choice of data values.
This can be easily corrected by `rigidifying' $\varphi^-$, namely, by letting
\begin{align*}
  \varphi^=(x,y) ~~\eqdef~~ \varphi^-(x,y) ~~\et~~ \forall y' ~ \varphi^-(x,y') ~\then~ y'=y \ .
\end{align*}
Indeed, when interpreted on a generic data word $w$ and a position $x$ in it,
the formula~$\varphi^=(x,y)$ is equivalent to~$\varphi^-(x,y)$ as long as 
there is at most one position $y$ in $w$ that satisfies $\varphi^-$. Otherwise, 
$\varphi^=(x,y)$ simply does not hold. 

Knowing that $\varphi^=(x,y)$ is a classical MSO / FO formula and that, by construction,
$x$ determines $y$ in it, we can apply the sub-definability lemma to $\varphi^=(x,y)$,
thus obtaining finitely many formulas $\beta^-_i(z,y)$ where $z$ determines $y$ and such 
that, for all~$x\leq z\leq y$, 
$$
  w \models \varphi^=(x,y) 
  \qquad\text{implies}\qquad 
  w \models \beta^-_i(z,y) \text{ for some $i$}.
$$
From each formula~$\beta^-_i(z,y)$, we reconstruct a formula of rigidly guarded \MSO / \FO
formula $\beta_i(z,y)$ by syntactically replacing every occurrence of unary predicate 
$c^\sim_\alpha(x')$ with $\exists y'~ \alpha(x',y') \:\et\: x'\sim y'$. 
It is clear that, since $\beta^-_i(z,y)$ defines a unique $y$ from $z$, so does~$\beta_i(z,y)$.
Furthermore, for all data words~$w$ and all positions~$x\leq z\leq y$, we have
$$
  w \models \varphi(x,y)
  \qquad\text{iff}\qquad 
  w^- \models \varphi^-(x,y)
  \qquad\text{iff}\qquad 
  w^- \models \varphi^=(x,y)
$$
and hence, there is~$i$ such that~$w^- \models \beta^-_i(z,y)$ and~$w \models \beta_i(z,y)$.
\end{proof}

An immediate consequence of the above lemma is the following:

\begin{corollary}\label{cor:rigidification}
Every inward-rigid formula is equivalent to a finite disjunction of rigid formulas.
\end{corollary}

\begin{proof}
Consider an inward-rigid formula $\varphi(x,z_1,\ldots,z_k,y)$ and define
$\phi(x,y) = \exists z_1,\ldots,z_k ~ \varphi(x,z_1,\ldots,z_k,y)$.
Since $x$ determines $y$ in $\phi(x,y)$, we can apply Lemma~\ref{lemma:subdefinability},
thus obtaining the formulas $\alpha_1(z,y)$, $\ldots$, $\alpha_n(z,y)$.
Accordingly, the desired rigid formulas are defined by
$$
  \phi_{i_1,\ldots,i_k}(x,z_1,\ldots,z_k,y) 
  ~\eqdef~ \varphi(x,z_1,\ldots,z_k,y) ~\et~ \alpha_{i_1}(z_1,y) ~\et~ \ldots ~\et~ \alpha_{i_k}(z_k,y)
$$
where $i_1,\ldots,i_k$ are indices ranging over $\{1,\ldots,n\}$.

One easily checks that the formulas~$\phi_{i_1,\ldots,i_k}(x,z_1,\ldots,z_k,y)$ are rigid.
Indeed, $x$ determines $z_i$, which in its turn determines~$y$, and $y$ determines $x$.
Of course, $\phi_{i_1,\ldots,i_k}$ entails $\varphi$, by construction. Conversely, 
given some positions $x,z_1,\ldots,z_k,y$ such that~$w\models\varphi(x,z_1,\dots,z_k,y)$,
we know that $x\le z_1,\ldots,z_k\le y$ and, by Lemma~\ref{lemma:subdefinability}, there 
exist~$i_1,\ldots,i_k$ such that $w\models\alpha_{i_1}(z_1,y) ~\et~ \ldots ~\et~ \alpha_{i_k}(z_k,y)$,
and hence $w\models\phi_{i_1,\dots,i_k}(x,z_1,\dots,z_k,y)$.
We have just proved that $\varphi(x,z_1,\ldots,z_k,y)$ is equivalent to the finite
disjunction $\bigvee_{1\le i_1,\dots,i_k\le n}\phi_{i_1,\dots,i_k}(x,z_1,\dots,z_k,y)$ 
of rigid formulas.
\end{proof}

\medskip
We now formalize the meaning of ``computing the type under a guard''.
For this it is convenient to fix an orbit-finite data monoid $\cM$ that is 
given by a term-based presentation system $\cS=(T,\odot,\check\blank,\dsim)$.
This means that the elements of $\cM$ are the $\dsim$-equivalence classes of the 
terms in $T$. However, by a slight abuse of notation, we shall often identify 
the terms $o(d_1,\ldots,d_k)$ in $T$ with the corresponding elements 
$[o(d_1,\ldots,d_k)]_\dsim$ of $\cM$. For example, we can write
$h(w[x,y]) \,=\, o(d_1,\ldots,d_k)$.

\begin{definition}\label{def:computing-types}
Let $o$ be an orbit of the data monoid $\cM$ having memory size $k$.
A formula $\varphi(x,z_1,\ldots,z_k,y)$ \emph{witnesses the orbit $o$} 
if it is inward-rigid and
$$
  w \models \varphi(x,z_1,\ldots,z_n,y) 
  \qquad\text{implies}\qquad 
  h(w[x,y]) \,=\, o(w[z_1],\ldots,w[z_n]) \ .
$$
A family of formulas $F=(\varphi_o)_{o\in O}$ 
\emph{computes the types under the guard~$\alpha(x,y)$} if each 
formula~$\varphi_o(x,\bar{z},y)$ witnesses the orbit~$o$ and, 
moreover, the guard $\alpha(x,y)$ is logically equivalent to 
$\bigvee_{\!o\,\in O} \exists \bar{z} ~ \varphi_o(x,\bar{z},y)$. 
We say that \emph{one can compute the types under the guard~$\alpha(x,y)$} 
if there exists such a family of formulas.
\end{definition}

We aim at proving that for every rigid formula~$\alpha(x,y)$ (and, in particular, 
for the rigid formula $\alpha(x,y) \:=\: (\neg\exists z ~ z<x) \et (\neg\exists z ~ z>y)$), 
one can compute the types under $\alpha(x,y)$. As we mentioned, the proof of Theorem 
\ref{thm:monoids-to-logic} exploits an induction on the partial order $\le_{\orbit\cJ}$ 
of the $=_{\orbit\cJ}$-classes of $\cM$. The invariant of the induction is given in 
the following lemma. 

\begin{lemma}[Inductive statement]\label{lemma:inductive-statement}
For every $=_{\orbit\cJ}$-class $O$ of $\cM$:
\begin{condlist}
  \item \label{claim1}
        there exists a formula $\varphi_O(x,y)$ such that 
        $w \sat \varphi_O(x,y)$ iff $h(w[x,y])\in O$;
  \item \label{claim2}
        for all rigid guards $\alpha(x,y)$ such that~$w\sat\alpha(x,y)$ 
        implies $h(w[x,y]) \ge_{\orbit\cJ} O$, one can compute the types 
        under~$\alpha$.
\end{condlist}
\end{lemma}

We will prove the above lemma first under the assumption that $\cM$ is {\sl aperiodic}, 
constructing formulas of the rigidly guarded \FO logic. In the aperiodic case, we use the 
fact that the orbit of an infix is determined by its $\gLO$-class and its $\gRO$-class and 
by the equality relationships between the memorable values in these two classes (this 
follows basically from the fact that the $\cH$-classes of an aperiodic monoid are 
singletons). In the second part, we will reprove the same lemma without the assumption 
of aperiodicity, constructing formulas of the rigidly guarded \MSO logic. In this case
different objects have to be guessed by quantifying over monadic second-order predicates. 

\subsection{The translation in the aperiodic case}\label{subsec:monoids2logic-fo}

In this section we assume that $\cM$ is an {\sl aperiodic} orbit-finite data monoid 
and we prove the inductive statement given in Lemma \ref{lemma:inductive-statement},
where formulas are meant to be rigidly-guarded \FO formulas.

For the sake of brevity, we can fill the parameters of a formula~$\varphi(x_1,\dots,x_n)$ 
with $\star$ to denote the fact that the corresponding variables are existentially quantified. 
With this notation, if~$\varphi(x,y,z)$ is rigid according to Definition \ref{def:inward-rigid}, 
then so is~$\varphi(\star,y,z)$, as well as~$\varphi(x,\star,z)$ and~$\varphi(x,y,\star)$.

\medskip
We begin by presenting a special, but important, case of Lemma \ref{lemma:inductive-statement}, 
which shows that the types of infixes of length~$1$ can be computed (this will serve as our 
base case for the inductive construction).

\begin{lemma}\label{lemma:base-case}
Let~$\alpha_1(x,y) \:\eqdef\: (x=y)$. One can compute the types under the guard~$\alpha_1$.
\end{lemma}

\begin{proof}
Note that the morphism~$h$ maps singleton words to orbits that have memory size at most~$1$.
A family~$F_1$ that computes the types under~$\alpha_1$ consists of formulas 
$\phi_o$, for all orbits $o\in O$, defined by
\begin{align*}
  \phi_o(x,y)     &~~\eqdef~ \bigvee_{\!\!h((d,a)) \,=\, o(\emptystr)\!\!} 
                             a(x) ~\et~ x=y                  
  \tag{if~$o$ has memory size~$0$} \\[1ex]
  \phi_o(x,z_1,y) &~~\eqdef~ \bigvee_{\!\!h((d,a)) \,=\, o(d)\!\!} 
                             a(x) ~\et~ x=y ~\et~ x=z_1
  \tag{if~$o$ has memory size~$1$} \\[1ex]
  \phi_o(x,y)     &~~\eqdef~ \false
  \tag{otherwise}
\end{align*}
\end{proof}

The next lemma shows how to compose families of formulas that compute the types under 
some given guards. This is one of the places where the products of the data monoid $\cM$ are 
simulated by comparing memorable values. In particular, the lemma depends on the 
fact that any inward-rigid formula used to witness an orbit can be written as a 
finite disjunction of rigid formulas (Corollary \ref{cor:rigidification}).

\begin{lemma}\label{lemma:products}
Given two rigid guards $\alpha(x,y)$ and $\alpha'(x',y')$, let $\alpha\cdot\alpha'$
be the rigid guard defined by 
$(\alpha\cdot\alpha')(x,y) ~\eqdef~ \exists z ~ \alpha(x,z) ~\et~ \alpha'(z+1,y)$.
Given two families of formulas $F$ and~$F'$ that compute the types under the 
guards~$\alpha$ and~$\alpha'$, respectively, one can construct a family~$F\cdot F'$ 
that computes the types under the guard $\alpha\cdot\alpha'$.
\end{lemma}

\begin{proof}
Let~$F=(\varphi_o)_{o\in O}$ and~$F'=(\varphi'_o)_{o\in O}$.
We aim at constructing~$F\cdot F'=(\psi_o)_{o\in O}$ that 
computes the types under~$\alpha\cdot\alpha'$. We begin by
recalling that the orbit that results from the product of an 
element in orbit~$o$ with an element in orbit~$o'$ depends 
on the respecive memorable values. The equality relationships 
between memorable values will be represented by pairs of terms 
with data values (up to renaming, there are only finitely many 
pairs) and, for each such pair, we will produce a corresponding formula. 

Consider two terms~$t \:=\: o(d_1,\dots,d_k)$ and~$t' \:=\: o'(e_1,\dots,e_h)$ 
and let $t\cdot t' \:=\: o''(f_1,\dots,f_\ell)$ be their product according to $\cM$.
Let~$\bigvee_{1\le p\le n} \: \phi_{o,p}(x,z_1,\ldots,z_k,y)$ 
and~$\bigvee_{1\le q\le m} \: \phi'_{o',q}(x,z_1,\ldots,z_h,y)$ 
be the finite disjunctions of rigid formulas, equivalent to $\varphi_o$ and $\varphi'_{o'}$,
respectively, that are obtained from Corollary~\ref{cor:rigidification}.
Define 
$$
  \beta_{p,q}^{i,j}(z_i,z'_j) 
  ~~\eqdef~~ \exists y.~ \phi_{o,p}(\star,\bar\star,z_i,\bar\star,y) ~\et~ 
                         \phi_{o',q}(y+1,\star,\bar\star,z'_j,\bar\star,\star)
$$
and
\begin{align*}
  \psi_{t\cdot t'}(x,z''_1\dots,z''_\ell,y) 
  ~~\eqdef&~ ~\exists \: \xi, 
             ~\exists \: \bar{z}=z_1\dots z_k, 
             ~\exists \: \bar{z}'=z'_1\dots z'_h. 
             \tag{a} \\[1ex]
  ~~\et~~ &~ \bigvee_{p,q} ~ 
             \Big(~ 
             \phi_{o,p}(x,\bar{z},\xi,) ~~\et~~ \phi'_{o',q}(\xi+1,\bar{z}',y) 
             \tag{b} \\
          &~ \qquad~\et~~ 
             \bigwedge_{\!\!\!\!d_i=e_j\!\!\!\!} ~ 
             \beta_{p,q}^{i,j}(z_i,z'_j) ~\et~ z_i\sim z'_j 
             \tag{c} \\
          &~ \qquad~\et~~     
             \bigwedge_{\!\!\!\!d_i\neq e_j\!\!\!\!} ~
             \beta_{p,q}^{i,j}(z_i,z'_j) ~\et~ z_i\nsim z'_j ~\Big)
             \tag{d} \\[1ex]
  ~~\et~~ &~ \bigwedge_{\!\!\!\!f_i=d_j\!\!\!\!} ~ 
             z''_i=z_j 
  \quad\et   \bigwedge_{\!\!\!\!\!\begin{smallmatrix} 
                                    f_i=e_j \\ 
                                    f_i \,\nin\, \{d_1,\ldots,d_k\} 
                                  \end{smallmatrix}\!\!\!\!\!\!\!}
             z''_i=z'_j \ .
             \tag{e}
\end{align*}
Given $x$ and $y$, the formula $\psi_{t\cdot t'}(x,z''_1\dots,z''_\ell,y)$ 
first guesses the intermediate position~$\xi$ and the variables $\bar{z}$
and $\bar{z}'$ that contain the memorable values of $h(w[x,\xi])$ and of 
$h(w[\xi+1,y])$ (a). It then guesses the indices $p,q$ for the rigid 
formulas $\phi_{o,p}$ and $\phi'_{o',q}$ that hold over the factors 
$w[x,\xi]$ and $w[\xi+1,y]$ (b). Line (c) checks that, whenever a memorable
value of~$t$ and a memorable value of~$t'$ are equal, then the corresponding
positions in the factors share the same data value. Note that this comparison 
is done under the rigid guard~$\beta_{p,q}^{i,j}(z_i,z'_j)$, which of course
holds between $z_i$ and $z'_j$ whenever (b) holds. Similar conditions for disequalities
are verified in line (d). Finally, line (e) uniquely determines the positions 
of the memorable values of $t\cdot t'$ (in case a memorable value is shared 
between the left and the right term, priority is given to the leftmost position).

Overall, the formula~$\psi_{t\cdot t'}(x,z''_1\dots,z''_\ell,y)$ is inward-rigid 
and witnesses the orbit~$o''$. Furthermore, if~$x,\xi,y$ are positions such that 
$w \models \alpha(x,\xi)$ and $w \models \alpha'(\xi+1,y)$ and $\tau$ is a 
renaming such that $\tau(t) = h(w[x,\xi])$ and $\tau(t') = h(w[\xi+1,y])$, 
then $w \models \psi_{t\cdot t'}(x,\bar{z}'',y)$ for some tuples of positions
$\bar{z}''$. Therefore, the family $F\cdot F'=(\psi_o)_{o\in O}$ of formulas 
that computes the types under the guard $\alpha\cdot\alpha'$ can be obtained 
by associating with each orbit $o''$ the formula
$$
  \psi_o(x,\bar{z}'',y) ~~\eqdef~~ \bigvee_{t\cdot t' \,\in\, o} \psi_{t\cdot t'}(x,\bar{z}''y,).
$$
(this tries every possible pair of terms $t,t'$, among the finitely many 
different possibilities up to renamings, whose product yields the orbit $o$).
\end{proof}

Using Lemmas \ref{lemma:base-case} and \ref{lemma:products}, 
one can compute the orbits of infixes of fixed length:

\begin{corollary}\label{cor:fixed-length-infixes}
Let~$\alpha_k(x,y) \:\eqdef\: (x+k-1=y)$. One can compute the types under the guard~$\alpha_k$.
\end{corollary}

%
%
%
We now prove a technical lemma that is similar to Theorem V.1.9 from 
\cite{mathematical_foundations_of_automata}. The difference here is that
the hypothesis of the lemma uses the coarser equivalence $=_{\orbit\cJ}$ 
in place of $=_\cJ$. We will exploit this result several times in the paper, 
for instance, in the proof of Theorem \ref{lemma:cond2-fo}.

\begin{lemma}\label{lemma:decomposition-stairs}
For every pair of elements $s,t$ of $\cM$, 
if $s =_{\orbit\cJ} s\cdot t$, then $s =_\cR s\cdot t$.
Similarly, if $t =_{\orbit\cJ} s\cdot t$, then $t =_\cL s\cdot t$.
\end{lemma}

\begin{proof}
Suppose that $s =_{\orbit\cJ} s\cdot t$ (symmetric arguments can be used when
$t =_{\orbit\cJ} s\cdot t$ and with $\cL$ in place of $\cR$). 
By definition of $=_{\orbit\cJ}$, we know that there is a renaming $\tau$
such that $s \in M\cdot \tau(s\cdot t)\cdot M$, and hence there exist 
some elements $u,v$ of $\cM$ a such that 
$$
  s ~=~ \tau(u\cdot s\cdot t\cdot v) 
    ~=~ \tau(u)\cdot\tau(s)\cdot\tau(t\cdot v) \ .
$$
By repeatedly applying the mapping $\tau$ and substituting $s$ with
$\tau(u)\cdot\tau(s)\cdot\tau(t\cdot v)$, we obtain
$$
  s ~=~ \underbrace{\tau(u) \cdot \ldots \cdot \tau^n(u)}_{u_n} \:\cdot\:
        \underbrace{\tau^n(s)}_{s_n} \:\cdot\:
        \underbrace{\tau^n(t\cdot v) \cdot \ldots \cdot \tau(t\cdot v)}_{z_n} \ .
$$
Since $\tau$ is a permutation on $D$ that is the identity 
on all but finite many data values, we have that $\tau^{n_0}$ 
is the identity for some $n_0\ge1$. In particular, for all 
multiples $m\cdot n_0$ of $n_0$, we have 
$$
  u_{m\cdot n_0} ~=~ u_{n_0}^m
  \qquad\qquad\quad
  s_{m\cdot n_0} ~=~ s
  \qquad\qquad\quad
  z_{m\cdot n_0} ~=~ z_{n_0}^m \ .
$$
Moreover, since $\cM$ is locally finite, we can fix $m\ge 1$ 
large enough in such a way that $z_{n_0}^m$ is an idempotent. 
We thus obtain 
$$
  s ~=~ u_{n_0}^m \cdot s \cdot z_{n_0}^m
    ~=~ u_{n_0}^m \cdot s \cdot z_{n_0}^m \cdot z_{n_0}^m
    ~=~ s \cdot z_{n_0}^m
$$
whence 
$$
  s ~=~ s \cdot \tau^{n_0}(t\cdot v) \cdot \big(\tau^{n_0}(t\cdot v)\big)^{m-1}
    ~=~ s \cdot t \cdot v \cdot \big(\tau^{n_0}(t\cdot v)\big)^{m-1} \ .
$$
We have just shown that $s\cdot t \ge_\cR s$. As the converse relation $s \ge_\cR s\cdot t$
holds trivially, we conclude that $s =_\cR s\cdot t$.
\end{proof}

From the above lemma we easily obtain the following result:

\begin{lemma}\label{lemma:witness-not-J-above}
Let $O$ be a $=_{\orbit\cJ}$-class of $\cM$ and let $[x,y]$ be a {\sl minimal} 
interval such that~$h(w[x,y]) \not\geq_{\orbit\cJ} O$. We have that
\begin{enumerate}
  \item either $x=y$,
  \item or $x+1=y$,
  \item or $[x+1,y-1]$ is an interval such that $h(w[x+1,y-1]) >_{\orbit\cJ} O$.
\end{enumerate}
In particular, in the third case, there exists a $=_{\orbit\cJ}$-class $P$ that 
is strictly above $O$ (i.e., $P >_{\orbit\cJ} O$) and such that $[x+1,y-1]$ is a 
{\sl maximal} interval satisfying $h(w[x+1,y-1]) \in P$.
\end{lemma}

\begin{proof}
Let $[x,y]$ be a minimal interval such that~$h(w[x,y]) \not\geq_{\orbit\cJ} O$ 
and suppose that neither the first case nor the second case holds, namely, 
$x<x+1\le y-1<y$.
For the sake of brevity, let $s \,=\, h(w(x))$, $t \,=\, h(w[x+1,y-1])$, 
and $u \,=\, h(w(y))$.
We begin by noting that the minimality of $[x,y]$ implies 
$s\cdot t \ge_{\orbit\cJ} O$ and $t\cdot u \ge_{\orbit\cJ} O$.

Below, we aim at proving that $t >_{\orbit\cJ} s\cdot t$ and 
$t >_{\orbit\cJ} t\cdot u$, as this would imply that $t >_{\orbit\cJ} O$ 
and that $[x+1,y-1]$ is a maximal interval such that $h(w[x+1,y-1])$ 
belongs to the $=_{\orbit\cJ}$-class of $t$.
Suppose, by contradiction, that $t \not>_{\orbit\cJ} s\cdot t$. 
Since $t \ge_\cJ s\cdot t$,
we derive $t =_\cJ s\cdot t$. By applying Lemma \ref{lemma:decomposition-stairs}
we obtain $t =_\cL s\cdot t$. Moreover, since $\cL$ is a congruence with respect 
to products on the right, we derive $t\cdot u =_\cL s\cdot t\cdot u$. Finally, 
since $=_\cL$ refines $=_\cJ$, we obtain $t\cdot u =_\cJ s\cdot t\cdot u$, which 
contradicts the minimality of the interval $[x,y]$. We must conclude that 
$t >_{\orbit\cJ} s\cdot t$ and, by symmetric arguments, $t >_{\orbit\cJ} t\cdot u$.
\end{proof}

From now on, we assume that $O$ is a $=_{\orbit\cJ}$-class of $\cM$ and that both
claims \refcond{claim1} and \refcond{claim2} of Lemma \ref{lemma:inductive-statement} 
hold for all $=_{\orbit\cJ}$-classes $P$ that are strictly above~$O$ (we refer to this 
assumption as our inductive hypothesis). 

\begin{lemma}\label{lemma:not-J-above}
There exists a formula $\alpha^\min_{\not\geq O}(x,y)$ such that 
$w\models\alpha^\min_{\not\geq O}(x,y)$ iff $[x,y]$ is a minimal 
interval such that~$h(w[x,y]) \not\ge_{\orbit\cJ} O$.
Furthermore, the formula $\alpha^\min_{\not\geq O}(x,y)$ 
is rigid and one can compute the types under it.
\end{lemma}

\begin{proof}
Lemma~\ref{lemma:witness-not-J-above} describes three types of intervals
$[x,y]$ such that $h(w[x,y]) \not\ge_{\orbit\cJ} O$. 
For the first type of intervals, one simply lets $\alpha_1(x,y) \:\eqdef\: (x=y)$ 
and accordingly constructs the family $F_1$ that computes the types under $\alpha_1$
using Lemma \ref{lemma:base-case}. Similarly, for the second type of intervals, 
one lets $\alpha_2(x,y) \:\eqdef\: (x+1=y)$ and uses Corollary \ref{cor:fixed-length-infixes} 
to construct a family $F_2$ computing the types under $\alpha_2$.

We now focus on the most interesting type of intervals, which are of the form 
$[x,y]$, where $[x+1,y-1]$ is maximal such that $h(w[x+1,y-1]) \in P$ and $P$
is a specific $=_{\orbit\cJ}$-class strictly above $O$.
Let~$\alpha^\max_P(x,y)$ be a formula stating that $x\le y$ and $[x,y]$ is a 
maximal interval such that $h(w[x,y]) \in P$ (this formula exists thanks to 
the inductive hypothesis \refcond{claim1}). Note that the formula $\alpha^\max_P(x,y)$ 
is rigid by construction. Hence, by using this time the inductive hypothesis 
\refcond{claim2}, one can construct a family $F^\max_P$ that computes the types 
under the guard $\alpha^\max_P(x,y)$.

The desired formula~$\alpha^\min_{\not\geq O}(x,y)$ can be defined as follows:
\begin{align*} 
  \alpha^\min_{\not\geq O}(x,y) 
  ~~\eqdef&~~ \alpha_1(x,y) ~\et~ F_1(x,y) \,\not\geq_{\orbit\cJ} O \\[1.5ex]
    \vel~~&~~ \alpha_2(x,y) ~\et~ F_2(x,y) \,\not\geq_{\orbit\cJ} O \\[1ex]
    \vel~~&~~ \bigvee_{\!\!\!P \,>_{\orbit\cJ} O} 
              \Big(\, (\alpha_1\cdot\alpha^\max_P\cdot\alpha_1)(x,y) ~\et~ 
                      (F_1\cdot F^\max_P\cdot F_1)(x,y) \,\not\geq_{\orbit\cJ} O \,\Big)
\end{align*}
where the families $F_1,F_2,F^\max_P$ are used as if they were functions computing 
types and the operations of compositions are those outlined in Lemma \ref{lemma:products} 
(this shorthand of notation should be clear to understand and can be transformed 
into standard formulas by unfolding the finitely many cases). We also observe that,
since $\cM$ is orbit-finite, the disjunction over all $=_{\orbit\cJ}$-classed $P$
strictly above $O$ is finite. 

It is easy to see that the above formula $\alpha^\min_{\not\geq O}(x,y)$ correctly 
defines the minimal intervals $[x,y]$ such that $h(w[x,y]) \not\geq_{\orbit\cJ} O$.
Finally, the formula is rigid by construction and a family computing the types 
under this guard can be easily obtained using the same kind of constructions.
\end{proof}

We are now ready to prove the induction steps for claims \refcond{claim1} and 
\refcond{claim2} of Lemma \ref{lemma:inductive-statement} with respect to the 
${\orbit\cJ}$-class $O$. We remark that only the proof of Claim \refcond{claim2} 
relies on the assumption that the monoid $\cM$ is aperiodic, as well as on 
the properties of memorable values that we outlined in Section \ref{subsec:green-theory}.

\begin{lemma}[Induction step for \refcond{claim1}]\label{lemma:cond1-fo}
There exists a formula~$\varphi_O(x,y)$ such that
$w \models \varphi_O(x,y)$ iff~$h(w[x,y]) \in O$.
\end{lemma}

\begin{proof}
One first disproves the existence of an interval $[x',y']$ included 
in~$[x,y]$ and satisfying $\alpha^\min_{\not\geq O}(x',y')$. 
This property implies $h(w[x,y]) \ge_{\orbit\cJ} O$ and can be easily 
defined by a formula obtained from Lemma~\ref{lemma:not-J-above}.
One then excludes the case $h(w[x,y]) >_{\orbit\cJ} O$ by verifying
the conjunction of the properties $h(w[x,y]) \nin P$ over all $=_{\orbit\cJ}$-classes 
$P$ strictly above $O$. The latter properties can be defined thanks to the
inductive hypothesis~\refcond{claim1}.
\end{proof}

\begin{lemma}[Induction step for \refcond{claim2}]\label{lemma:cond2-fo}
For all rigid guards $\alpha(x,y)$ such that~$w\models\alpha(x,y)$ 
implies $h(w[x,y]) \ge_{\orbit\cJ} O$, one can compute the types under~$\alpha$.
\end{lemma}

\begin{proof}
Thanks to the inductive hypothesis, for each of the finitely many $=_{\orbit\cJ}$-classes
$P$ that are strictly above $O$, one can construct a formula $\varphi_P(x,y)$
that checks whether $h(w[x,y]) \in P$ and in this case compute the types 
under the rigid guard $\alpha(x,y) \et \varphi_P(x,y)$. Thus, to prove the 
lemma, it is sufficient to consider the case where 
$w\models\alpha(x,y)$ implies~$h(w[x,y])\in O$.

We begin by introducing the formula $\alpha^\min_O(x',y')$ that expresses the property 
that $[x',y']$ is a minimal interval satisfying $h(w[x',y']) \in O$ -- such a formula 
exists thanks to Lemma~\ref{lemma:cond1-fo} and, moreover, it is rigid. 
Next, we assume that $\alpha(x,y)$ holds and we consider the intervals $[x',y']$ 
that are included in $[x,y]$ and satisfy $\alpha^\min_O(x',y')$; we call these 
intervals \emph{blocks}. We focus in particular on the block $[x_1,y_1]$ whose 
left endpoint $x_1$ is as close as possible to $x$, as well as on the block 
$[x_2,y_2]$ whose right endpoint $y_2$ is as close as possible to $y$. These 
two special blocks can be defined from $x$ and $y$ by the following formula
$$
\begin{array}{rclccccl}
  \beta(x,x_1,y_1,x_2,y_2,y) 
  &\eqdef& \alpha(x,y) ~\et~ x\le x_1 ~\et~ y_2\le y ~\et~ \\[1ex]
  &      & \alpha^\min_O(x_1,y_1) ~\et~ 
           \forall x' ~ 
           \big(\, \alpha^\min_O(x',\star) &\then& x'<x &\,\vel&\, x_1\le x' &\!\!\!\!\big) 
           \\[1ex]
  &      & \alpha^\min_O(x_2,y_2) ~\et~ 
           \forall y' ~ 
           \big(\, \alpha^\min_O(\star,y') &\then& y<y' &\,\vel&\, y'\le y_2 &\!\!\!\!\big)
\end{array}
$$
(note that the formula implies $h(w[x,y])\in O$ and hence $[x_1,y_1],[x_2,y_2]\subseteq [x,y]$).
It is easy to see that $\beta$ is an inward-rigid formula: indeed, $x$ determines $x_1$,
which determines $y_1$, and $y$ determines $y_2$, which determines $x_2$.
Thus, by Corollary~\ref{cor:rigidification}, the formula is equivalent to 
a finite disjunction of rigid formulas, say~$\beta_1,\ldots,\beta_n$.

We can now describe the steps for computing the types under~$\alpha(x,y)$:
\begin{enumerate}
  \item Guess an index $i\in\{1,\ldots,n\}$ and some positions $x_1$, $y_1$, $x_2$, and $y_2$
        such that $w \models \beta_i(x,x_1,y_1,x_2,y_2,y)$.
  \item Compute the orbits under the rigid guard 
        $\beta_i(x,\star,y_1,\star,\star,\star)$.
        This is doable since 
        (i) $\beta_i(x,\star,y_1,\star,\star,\star)$ entails $\beta(x,\star,y_1,\star,\star,\star)$,
            which in its turn entails
            $\beta(x,x_1-1,\star,\star,\star,\star) \,\cdot\, \alpha^\min_O(x_1,y_1)$,
        (ii) thanks to the fact that $w \models \beta(x,x_1-1,\star,\star,\star,\star)$ 
             implies $h(w[x,x_1-1]) >_{\orbit\cJ} O$, one can exploit the inductive 
             hypothesis \refcond{claim2} to compute the types under the guard 
             $\beta(x,x_1-1,\star,\star,\star,\star)$,
        (iii) by Lemma~\ref{lemma:not-J-above}, 
              one can compute the types under the guard $\alpha^\min_O(x_1,y_1)$, and
        (iv) by Lemma \ref{lemma:products},
             one can compute the types under the guard
             $\beta(x,x_1-1,\star,\star,\star,\star) \,\cdot\, \alpha^\min_O(x_1,y_1)$.
        \par\noindent
        We also claim that the element $h(w[x,y_1])$, which is determined by the guard 
        $\beta_i(x,\star,y_1,\star,\star,\star)$, belongs to the same $=_\cR$-class as
        the element $h(w[x,y])$. 
        Indeed, we have $h(w[x,y]) \,=\, h(w[x,y_1]) \,\cdot\, h(w[y_1+1,y])$.
        Moreover, by construction, both elements $h(w[x,y])$ and $h(w[x,y_1])$ 
        belong to the same $=_{\orbit\cJ}$-class $O$. By Lemma \ref{lemma:decomposition-stairs}
        it follows that $h(w[x,y]) =_\cR h(w[x,y_1])$.
  \item In a similar way, compute the types under the rigid guard
        $\beta_i(\star,\star,\star,x_2,\star,y)$. By symmetric arguments, we know that
        the element $h(w[x_2,y])$ is in the same $=_\cL$-class as the element $h(w[x,y])$.
  \item Compute the orbits under the rigid guard $\alpha(x,y)$, as follows.
        First, recall that $h(w[x,y]) =_\cR h(w[x,y_1])$ and $h(w[x,y]) =_\cL h(w[x_2,y])$.
        Moreover, since $\cM$ is aperiodic, all its $=_\cH$-classes are singletons. 
        In particular, the intersection of the $=_\cR$-class of $h(w[x,y_1])$ and 
        the $=_\cL$-class of $h(w[x_2,y])$ is the singleton that contains precisely 
        the element $h(w[x,y])$.
        It follows that the orbit of $h(w[x,y])$ can be determined from the orbits and
        from the memorable values of the elements $h(w[x,y_1])$ and $h(w[x_2,y])$. This 
        information is available from to the previous constructions. In particular, one 
        can compare the memorable values of $h(w[x,y_1])$ and $h(w[x_2,y])$ using suitable 
        rigid guards, in a way that is similar to the proof of Lemma \ref{lemma:products}
        (note that this requires applying Corollary \ref{cor:rigidification} to the 
        inward-rigid formulas that witness the orbits of $h(w[x,y_1])$ and $h(w[x_2,y])$). 
        It remains to determine from the endpoints $x$ and $y$ some positions 
        that contain the memorable values of $h(w[x,y])$. 
        For this, one recalls that the $\cR$-memorable values of $h(w[x,y])$ 
        are the same as the $\cR$-memorable values of $h(w[x,y_1])$, for which 
        some witnessing positions can be determined thanks to the previous constructions.
        Similarly, one determines some positions for the $\cL$-memorable values of 
        $h(w[x_2,y])$, which are known to be the same as the $\cL$-memorable values 
        of $h(w[x,y])$. Finally, by Proposition~\ref{prop:memorable-values}, one 
        knows that there are no other memorable values for $h(w[x,y])$.
\end{enumerate}
It is routine to translate the above steps into a family of rigidly guarded \FO formulas
that compute the types of $h(w[x,y])$ under the guard $\alpha(x,y)$.
\end{proof}

The above arguments prove Lemma \ref{lemma:inductive-statement} under the assumption that
the orbit-finite data monoid $\cM$ is aperiodic. We conclude this part by proving the claim
of Theorem \ref{thm:monoids-to-logic} that deals with the aperiodic case.

\begin{corollary}\label{cor:monoids-to-logic-fo}
Every data language recognized by a morphism into an orbit-finite aperiodic data monoid
is effectively definable by a rigidly guarded \FO sentence.
\end{corollary}

\begin{proof}
Since the image of the data language via the recognizing morphism $h$ is a finite union
of orbits, it is sufficient to construct, for each orbit~$o$, a corresponding sentence 
$\varphi_o$ that holds over a data word~$w$ iff~$h(w)\in o$. 
For this, we consider the guard 
$\alpha(x,y) \:\eqdef\: (\neg\exists z ~ z<x) \,\et\, (\neg\exists z ~ z>y)$, which 
holds over $w$ iff~$x$ is the first position and $y$ the last position of $w$.
By claim \refcond{claim2} of Lemma \ref{lemma:inductive-statement}, we can 
construct a family of formulas $\varphi_o(x,\bar{z},y)$ that compute the types 
under~$\alpha$. The language~$h^{-1}(o)$ is thus defined by the sentence 
$\exists x,\bar{z},y ~ \alpha(x,y) \,\et\, \varphi_o(x,\bar{z},y)$.
\end{proof}

\subsection{The translation in the non-aperiodic case}\label{subsec:monoids2logic-mso}

In the previous section we have seen how to establish Theorem \ref{thm:monoids-to-logic}
in the aperiodic case. The remaining claim, stating that every data language recognized by an 
orbit-finite data monoid is definable in rigidly guarded \MSO logic, is proved by following 
the same structure, namely, by relying on the same induction on $=_{\orbit\cJ}$-classes and 
on similar constructions. 

We fix for the rest of this section an orbit-finite data monoid $\cM$
over a set $D$ of data values, and a morphism $h$ from the free data 
monoid $(D\times A)^*$ to $\cM$. We assume that all formulas defined 
hereafter are of rigidly-guarded \MSO.

The goal is to reprove Lemma \ref{lemma:inductive-statement}, but this time without assuming 
that the monoid $\cM$ is aperiodic. We recall that the proof of claim \refcond{claim1} 
(Lemma \ref{lemma:cond1-fo}) does not exploit the assumption of aperiodicity (as far 
as the induction hypothesis is admitted). Hence we can reuse this part of the proof 
for the monoid $\cM$. The only proof that needs to be changed is that of 
Lemma~\ref{lemma:cond2-fo}, and more precisely the constructions described in step 4. 
Below, we focus only on this part of the proof, assuming that $O$ is a $=_{\orbit\cJ}$-class
of $\cM$ and that the inductive hypothesis holds, namely, the claim of 
Lemma \ref{lemma:inductive-statement} for all $=_{\orbit\cJ}$-classes $P$ strictly above $O$.

\medskip
To compute the types under a rigid guard $\alpha(x,y)$, 
we will divide the infix $w[x,y]$
into several blocks. That is, given a data word $w$ and two positions $x\le y$ such that
$w\models\alpha(x,y)$, a formula will first guess a factorization of $w[x,y]$ into 
some infixes $w_1,\ldots,w_n$ which are small enough that they can be handled by the 
inductive hypothesis. Then, the formula will perform sub-computations that determine 
the orbit of each factor, as well as some positions carrying the memorable values in 
it. Finally, it will recursively compute the types of the partial products 
$h(w_1)\cdot\ldots\cdot h(w_i)$, for $i=1,\ldots,n$, eventually determining 
the orbit the entire product $h(w[x,y]) = h(w_1)\cdot\ldots\cdot h(w_n)$.

We begin by describing the factorizations we are mainly interested in (for the 
sake of simplicity, the definitions are given with respect to the whole word $w$, 
as if the rigid guard $\alpha(x,y)$ held over $w$ with $x=1$ and $y=|w|$).

\begin{definition}\label{def:J-factorization}
A \emph{factorization} of a data word $w$ is a sequence $w_1,\ldots,w_n$ 
of non-empty infixes such that $w \:=\: w_1\cdot\ldots\cdot w_n$.
This factorization is called an \emph{$O$-factorization}, for 
some $=_{\orbit\cJ}$-class $O$, if we have:
\begin{itemize}
  \item $h(w_1\cdot\ldots\cdot w_n)\in O$, 
  \item $h(w_i)\in O$, for all $1\le i\le n$.
\end{itemize}
Similarly, the factorization is called an \emph{$O$-prefactorization} if we have:
\begin{itemize}
  \item $h(w_1\cdot\ldots\cdot w_n)\in O$, 
  \item $h(w_i)\in O$ or $h(w_{i+1})\in O$, for all $1\le i\le n-1$,
  \item $h(w_i)\in O$ implies $h(u)\nin O$, for all proper infixes $u$ of $w_i$ 
        and for all $1\le i\le n$.
\end{itemize}
Finally, we call \emph{left endpoint} (resp., \emph{right endpoint}) of $w_i$
the position in $w$ where the factor $w_i$ begins (resp., ends).
\end{definition}

We remark that any factorization $w_1=w[x_1,y_1]$, $\ldots$, $w_n=w[x_n,y_n]$ of $w$ 
can be \emph{represented} in MSO by the pair $(X,Y)$ of monadic predicates that 
contain the left endpoints and the right endpoints of the factors, respectively, 
i.e.~$X=\{x_1,\ldots,x_n\}$ and $Y=\{y_1,\ldots,y_n\}$.

The second definition concerns special forms of factorizations where the endpoints 
of every factor can be determined, one from the other, by means of a rigid formula. 
Assuming that such a factorization exists, one can move from an endpoint to another 
adjacent endpoint, either to the left or to the right, in a deterministic manner.

\begin{definition}\label{def:traversable-factorization}
Let $G=\{\gamma_j(x',y')\}_{j=1,\ldots,k}$ be a finite family of rigid formulas.
We say that a factorization $w_1 \ldots w_n$ of $w$ is \emph{rigidly traversable by $G$} 
if for every $1\le i\le n$, there exists $1\le j\le k$ such that $w \models \gamma_j(x_i,y_i)$, 
where $x_i$ and $y_i$ are, respectively, the left and the right endpoints of the factor $w_i$.
\end{definition}

\begin{lemma}\label{lemma:J-prefactorization}
Let $O$ be a $=_{\orbit\cJ}$-class and let $\alpha(x,y)$ be a rigid 
formula such that $w\models\alpha(x,y)$ implies $h(w[x,y]) \in O$. 
One can construct a formula $\varphi^\fact_\prej(x,y,X,Y)$ and a 
finite family $G=\{\gamma_j(x',y')\}_{j=1,\ldots,k}$ of rigid 
formulas such that:
\begin{enumerate}
  \item if $w \models \alpha(x,y)$, 
        then $w \models \varphi^\fact_\prej(x,y,X,Y)$ for some sets $X,Y\subseteq[x,y]$;
  \item if $w \models \varphi^\fact_\prej(x,y,X,Y)$, then the pair $(X,Y)$ represents an 
        $O$-prefactorization of $w[x,y]$ that is rigidly traversable by $G$;
  \item one can compute the types under each guard $\gamma_j\in G$.
\end{enumerate}
\end{lemma}

\begin{proof}
We begin by describing the formula $\varphi^\fact_\prej(x,y,X,Y)$. The main idea is to 
verify that the pair $(X,Y)$ of monadic predicates represents an $O$-prefactorization
$w_1,\ldots,w_n$ of $w[x,y]$, namely, $X$ contains the left endpoints and $Y$ contains the 
right endpoints of the factors $w_1,\ldots,w_n$.
This is easy to do since we can use the formula $\varphi_O(x',y')$ from claim \refcond{claim1} 
of Lemma \ref{lemma:inductive-statement} to verify all the properties that define an 
$O$-prefactorization (recall that the proof of claim \refcond{claim1} that we gave 
in Section \ref{subsec:monoids2logic-fo} holds even when the data monoid $\cM$ is not aperiodic). 
However, this is not yet the definition of $\varphi^\fact_\prej(x,y,X,Y)$.
Because in general there exist many $O$-prefactorizations for the same word $w[x,y]$ 
and because we need to be able to move across the endpoints of the factors in a deterministic 
manner, it is convenient to commit to a single $O$-prefactorization of $w[x,y]$, which 
can be uniquely determined from $x$ and $y$. More precisely, if we denote by 
$\tilde\varphi^\fact_\prej(x,y,X,Y)$ the formula that checks that $(X,Y)$ 
represents an $O$-prefactorization of $w[x,y]$, then we can uniquely determine one such 
pair $(X,Y)$ from $x$ and $y$ by selecting, for instance, the least one according to some 
fixed MSO-definable total order $\le$. This is what the formula $\varphi^\fact_\prej(x,y,X,Y)$ 
does. In addition, the formula verifies that the guard $\alpha(x,y)$ holds between $x$ and $y$
(this will be used later and is not a restriction). Summing up, we let
$$
\begin{array}{rrl}
  \varphi^\fact_\prej(x,y,X,Y) 
  &~\eqdef&~ \alpha(x,y) ~\et~ \tilde\varphi^\fact_\prej(x,y,X,Y) \\[1.5ex]
  &\et    &~  \forall X',Y' ~ \tilde\varphi^\fact_\prej(x,y,X',Y') ~\then~ (X,Y) \le (X',Y') \ .
\end{array}
$$
For short, we call \emph{canonical prefactorization of $w[x,y]$} the factorization that is
represented by the unique pair $(X,Y)$ that satisfies $\varphi^\fact_\prej(x,y,X,Y)$.

We now construct the family $G$ of rigid formulas that can be used to move across the endpoints
of the canonical prefactorization of $w[x,y]$. The problem can be reduced to determining the positions 
$x$ and $y$ from a given endpoint of a factor: once we know $x$ and $y$, we can reconstruct the 
canonical prefactorization of $w[x,y]$ and then determine the other endpoint of the factor.
We use the sub-definability Lemma \ref{lemma:subdefinability} to find, for each endpoint $z$, 
a finite number of possible choices for $x$ and $y$. More precisely, we recall that 
$\varphi^\fact_\prej(x,y,X,Y)$ entails the rigid formula $\alpha(x,y)$ and we apply to
this formula the sub-definability Lemma \ref{lemma:subdefinability}. In this way we obtain 
a finite set of formulas $\beta_1(z,y)$, $\ldots$, $\beta_k(z,y)$ such that
\begin{itemize}
  \item $z$ determines $y$ in $\beta_j(z,y)$, for all $1\le j\le k$, 
  \item if $z$ is an endpoint of the canonical prefactorization of $w[x,y]$, 
        then there exists $1\le j\le k$ such that $w \models \beta_j(z,y)$.
\end{itemize}
We can thus define the family $G$ as the set of formulas 
$$
\begin{array}{rcrl}
  \gamma_j(x',y') &~\eqdef&~ \exists\, x,y,X,Y ~&
                             \beta_j(x',y) ~\et~ \varphi^\fact_\prej(x,y,X,Y) \\[1ex]
                  &       &  \et~& x'\le y' ~\et~ x'\in X ~\et~ y'\in Y ~\et~ \\[1ex]
                  &       &  \et~& \forall z ~ \big( x'<z<y' ~\then~ z\nin X\cup Y \big) \ .
\end{array}
$$
It is easy to see that, in every formula $\gamma_j(x',y')$ of $G$, the variable $x'$
determines $y'$. Indeed, $\beta_j(x',y)$ determines $y$ from $x'$, 
$\varphi^\fact_\prej(x,y,X,Y)$ determines $x,X,Y$ from $y$, and the remaining 
of conjuncts of $\gamma_j(x',y')$ determine $y'$ from $X,Y$. By symmetric arguments,
$y'$ determines $x'$ in $\gamma_j(x',y')$, and hence $G$ contains only rigid formulas.
Moreover, it is clear from the previous constructions that if $w_i=w[x_i,y_i]$ is a
factor of the canonical prefactorization of $w[x,y]$, then $w \models \gamma_j(x',y')$
for some formula $\gamma_j\in G$.

It remains to show how to compute the types under the guards in $G$.
For this, we recall that every formula $\gamma_j(x',y')$ determines 
a factor $w_i$ of a possible $O$-prefactorization. By Definition 
\ref{def:J-factorization} we know that either $w_i$ is a minimal infix 
such that $h(w_i)\in O$ (hence $w \models \alpha^\min_{\not\geq O}(x',y')$), 
or $h(w_i)\in P$ for some $=_{\orbit\cJ}$-class $P$ strictly above $O$. 
Furthermore, we can check with the formula $\varphi_O(x',y')$ which of the 
two cases holds. We can finally exploit Lemma \ref{lemma:not-J-above} and 
our inductive hypothesis (claim \refcond{claim2}) to construct a family 
$F_{\gamma_j}$ of formulas that compute the types under the guard $\gamma_j(x',y')$.
\end{proof}

The analogous lemma for $J$-factorizations is shown below.

\begin{lemma}\label{lemma:J-factorization}
Let $O$ be a $=_{\orbit\cJ}$-class and let $\alpha(x,y)$ be a rigid 
formula such that $w\models\alpha(x,y)$ implies $h(w[x,y]) \in O$. 
One can construct a formula $\varphi^\fact_O(x,y,X',Y')$ and a 
finite family $G'$ of rigid formulas such that:
\begin{enumerate}
  \item if $w \models \alpha(x,y)$, 
        then $w \models \varphi^\fact_O(x,y,X',Y')$ for some sets $X',Y'\subseteq[x,y]$;
  \item if $w \models \varphi^\fact_O(x,y,X',Y')$, then the pair $(X',Y')$ represents a
        $O$-factorization of $w[x,y]$ that is rigidly traversable by $G'$;
  \item one can compute the types under each guard $\gamma'\in G'$.
\end{enumerate}
\end{lemma}

\begin{proof}
Let $\varphi^\fact_\prej$ and $G$ be as in Lemma \ref{lemma:J-prefactorization}
and let $X=\{x_1,\ldots,x_n\}$ and $Y=\{y_1,\ldots,y_n\}$ be some monadic predicates 
satisfying $\varphi^\fact_\prej(x,y,X',Y')$. The pair $(X',Y')$ represents an 
$O$-prefactorization of $w[x,y]$ with factors $w_1=w[x_1,y_1]$, $\ldots$, $w_n=w[x_n,y_n]$.
For the sake of simplicity, we assume that~$n$ is even (otherwise, the last 
factor has to be treated differently). 

We define a new factorization by grouping the factors $w_1,\ldots,w_n$ two by two, 
namely, we let $v_i \eqdef w_{2i-1} \cdot w_{2i}$ for all $1\le i\le \frac{n}{2}$ 
and, accordingly, $X' \eqdef \{x_1,x_3,\ldots,x_{n-1}\}$ and $Y' \eqdef \{y_2,y_4,\ldots,y_n\}$. 
Clearly, the sequence $v_1,\ldots,v_{\frac{n}{2}}$ is a $O$-factorization of $w[x,y]$ 
and the two monadic predicates $X',Y'$ can be defined from $X,Y$ by means of an MSO formula. 
This is exactly how the formula $\varphi^\fact_O(x,y,X',Y')$ defines a $O$-factorization 
of $w[x,y]$.

It remains to show that the defined $O$-factorization is rigidly traversable by a 
family $G'$ and that one can compute the types under each guard $\gamma'\in G'$. 
For this, we recall that the original factorization $w_1,\ldots,w_n$ is rigidly 
traversable by $G$, and we let $G'$ be the set of all formulas 
$\gamma_1\cdot\gamma_2 \,\eqdef\, \exists z' ~ \gamma_1(x',z') ~\et~ \gamma_2(z'+1,y')$, 
with $\gamma_1,\gamma_2\in G$. Clearly, for every factor 
$v_i = w_{2i-1}\cdot w_{2i} = w[x_{2i-1},y_{2i}]$, there 
there are formulas $\gamma_1,\gamma_2 \in G$ such that
$w \sat \gamma_1(x_{2i-1},y_{2i-1})$ and $w \sat \gamma_2(x_{2i},y_{2i})$,
whence $w \sat (\gamma_1\cdot\gamma_2)(x_{2i-1},y_{2i})$.
Finally, by Lemma \ref{lemma:products} we conclude that one can 
compute the types under each guard $\gamma_1\cdot\gamma_2\in G'$.
\end{proof}

\medskip
What remains to be done is to find a suitable mechanism for determining 
the orbits and the memorable values of $h(w[x,y])$ on the basis of the orbits 
and the memorable values of $h(w_1)$, $\ldots$, $h(w_n)$, where $w_1,\ldots,w_n$ 
are some factors of $w[x,y]$ as defined in Lemma \ref{lemma:J-factorization}.
Of course, this cannot be done by comparing all the memorable values of 
$h(w_1),\ldots,h(w_n)$, as there are unboundedly many of them and, furthermore, 
it is impossible to do so using rigidly guarded tests. Below, we show how to 
perform a sort of an approximation of this computation. The rough idea is to 
apply suitable renamings to the factors, so as to decrease the number of 
distinct memorable values. In doing so, however, one has to take into account 
the fact that some memorable values are shared between adjacent factors, and 
hence the renamings must be propagated to them.

\begin{definition}\label{def:local-consistency}
We say that two sequences $s_1,\ldots,s_n$ and 
$t_1,\ldots,t_n$ of elements of $\cM$ are \emph{locally consistent} if 
\begin{itemize}
  \item all elements $s_1,\ldots,s_n$, $t_1,\ldots,t_n$, and the two products 
        $s_1\cdot\ldots\cdot s_n$ and $t_1\cdot\ldots\cdot t_n$ belong to
        the same $=_{\orbit\cJ}$-class $O$,
  \item for every $1\le i\le n-1$, there is a renaming $\pi_i$ such that 
        $\pi_i(s_i) = t_i$ and $\pi_i(s_{i+1}) = t_{i+1}$,
  \item $s_1=t_1$ and $s_n=t_n$. 
\end{itemize}
Similarly, the sequences are \emph{almost locally consistent} if the first two 
conditions hold and, instead of $s_1=t_1$ and $s_n=t_n$, one has $\pi(s_1)=t_1$ 
and $\pi(s_n)=t_n$ for some renaming $\pi$.
\end{definition}

The following lemma shows why we are interested in locally consistent sequences.

\begin{lemma}\label{lemma:local-consistency}
Let $s_1,\ldots,s_n$ and $t_1,\ldots,t_n$ be two sequences of elements of $\cM$.
\begin{enumerate}
  \item If $s_1, \ldots, s_n$ and $t_1, \ldots, t_n$ are locally consistent, 
        then $s_1 \cdot\ldots\cdot s_n = t_1 \cdot\ldots\cdot t_n$.
  \item If $s_1, \ldots, s_n$ and $t_1, \ldots, t_n$ are almost locally consistent,
        then $\tau(s_1 \cdot\ldots\cdot s_n) = t_1 \cdot\ldots\cdot t_n$
        for some renaming~$\tau$.
\end{enumerate}
\end{lemma}

\begin{proof}
To prove the lemma, we need to introduce further definitions. 
As usual we identify elements of $\cM$ with terms of some suitable term-based presentation system
(the equality relation on the elements of $\cM$ is thought of as the congruence of the term-based 
presentation system). We denote the arity of a term $s=o(d_1,\ldots,d_k)$ by $\arity(s) = k$. 
Let $\bar{s}=s_1,\ldots,s_n$ be a sequence of terms.
The \emph{domain} $V_{\bar{s}}$ is the set of pairs $(i,j)$, where $i\in\{1,\ldots,n\}$ 
identifies a term $s_i$ and $j\in\{1,\ldots,\arity(s_i)\}$ identifies a placeholder of 
a data value of $s_i$. 
We equip the domain $V_{\bar{s}}$ with a relation $E_{\bar{s}}$ that is the finest 
equivalence that groups every two elements $(i,j)$ and $(i+1,j')$ of the domain 
whenever the $j$-th value of $s_i$ and the $j'$-th value of $s_{i+1}$ coincide. 
Of course, all elements of an equivalence class of $E_{\bar{s}}$ 
have the same associated data value. 
A \emph{colouring} of $(V_{\bar{s}},E_{\bar{s}})$ is a labelling of the equivalence 
classes of $E_{\bar{s}}$ by data values. The sequence $\bar{s}$ naturally induces
a colouring of $(V_{\bar{s}},E_{\bar{s}})$ that maps each equivalence class $C$ of 
$E_{\bar{s}}$ to the data value that is associated with the elements of $C$.
An element $(i,j)$ in $V_{\bar{s}}$ is a \emph{border position} if $i=1$ or $i=n$. 
Two colourings are \emph{border-equal} if they agree on all border positions. 
A \emph{border-class} is an equivalence class that contains a border position. 
We prove two claims now.
        
\begin{claim}
If two sequences $\bar{s}$ and $\bar{t}$ are locally consistent, then they define the 
same domain $V_{\bar{s}}=V_{\bar{t}}$ and the same equivalence $E_{\bar{s}}=E_{\bar{t}}$;
furthermore, the corresponding colourings are border-equal.
\end{claim}     

\begin{proof}[Proof of claim]
Let $\bar{s}=s_1,\ldots,s_n$ and $\bar{t}=t_1,\ldots,t_n$ be locally consistent sequences of 
terms. It is obvious that $\bar{s}$ and $\bar{t}$ have the same domain $V_{\bar{s}}=V_{\bar{t}}$,
and it is equally easy to see that their colourings are border-equal. 
Moreover, the fact that two positions $(i,j)$ and $(i+1,j')$ are $E_{\bar{s}}$-equivalent 
only depends on the equalities between the data values in $s_i$ and $s_{i+1}$. 
Those equalities are the same as those for $t_i$ and $t_{i+1}$, because there is 
a renaming $\pi_i$ such that $\pi_i(s_i)=t_i$ and $\pi_i(s_{i+1})=t_{i+1}$.
\end{proof}             

We give the following additional definitions. Two colourings have a \emph{small difference} 
if their domains and equivalences are the same and they disagree on the colour of {\sl at most one}
equivalence class. Two locally consistent sequences have a small difference if their colourings 
have small difference.

\begin{claim}
If two sequences $\bar{s}=s_1,\ldots,s_n$ and $\bar{t}=t_1,\ldots,t_n$ are 
locally consistent and have small difference, then they give the same product, 
i.e.~$s_1\cdot\ldots\cdot s_n = t_1\cdot\ldots\cdot t_n$.
\end{claim}

\begin{proof}[Proof of claim]
Let $\bar{s}=s_1,\ldots,s_n$ and $\bar{t}=t_1,\ldots,t_n$ be locally consistent sequences 
of terms with small difference. By the previous claim, the two sequences define the 
same domain, hereafter denoted by $D$, and the same equivalence, hereafter denoted by $E$. 
Let $C$ be the equivalence class of $E$ for which the two colourings induced by $\bar{s}$ 
and $\bar{t}$ differ. Let $i_0$ be the smallest number such that there is $(i_0,j) \in C$ 
and let $i_1$ be the largest number such that $(u_1,j') \in C$. Since $\bar{s}$ and $\bar{t}$
are locally consistent, we have that $i_0>1$ and $i_1<n$. Let $d$ and $e$ be the colours
associated with the class $C$ by $\bar{s}$ and $\bar{t}$, respectively. We define $\pi$ 
to be the renaming that swaps $d$ and $e$ and is the identity elsewhere. 
We observe that:
\begin{enumerate}
  \item $\pi \circ \pi$ is the identity,
  \item $\pi(s_i) = t_i$, for all $i\in\{i_0,\ldots,i_1\}$,
  \item $s_i = t_i$, for all $i \in \{1,\ldots,i_0-1\} \cup \{i_1+1,\ldots,n\}$,
  \item $\pi(u)=u$ where $u=t_{i_0-1} \cdot\ldots\cdot t_{i_1+1}$. 
        This holds because neither $d$ nor $e$ are memorable values of $u$.
        Indeed, we observe that the three elements $t_{i_0-1}$, $t_{i_1+1}$, 
        and $u$ are in the same $=_{\orbit\cJ}$-class and, furthermore, 
        $t_{i_0-1} \ge_\cR u$ and $t_{i_1+1} \ge_\cL u$.
        From Lemma \ref{lemma:decomposition-stairs} we know that $t_{i_0-1} =_\cR u$
        and $t_{i_1+1} =_\cL u$, and hence all $\cR$-memorable values of $u$ must 
        occur in $t_{i_0-1}$ and all $\cL$-memorable values in $u$ must occur in 
        $t_{i_1+1}$. 
        However, neither $d$ nor $e$ is $\cR$-memorable in $t_{i_0-1}$, 
        and neither $d$ nor $e$ is $\cL$-memorable in $t_{i_1+1}$. 
        By Proposition \ref{prop:memorable-values} we conclude that 
        neither $d$ nor $e$ are memorable in $u$.
\end{enumerate}   
The above properties can be used to prove our second claim:
\begin{align*}
  s_1 \cdot\ldots\cdot s_n 
  &~=~ s_1 \cdot\ldots\cdot s_{i_0-2} ~\cdot~ 
       \pi \big( \pi(s_{i_0-1} \cdot\ldots\cdot s_{i_1+1}) \big) ~\cdot~ 
       s_{i_1+2} \cdot\ldots\cdot s_n 
       \tag{by (1)} \\[1ex]
  &~=~ s_1 \cdot\ldots\cdot s_{i_0-2} ~\cdot~ 
       \pi \big( t_{i_0-1} \cdot\ldots\cdot t_{i_1+1} \big) ~\cdot~ 
       s_{i_1+2} \cdot\ldots\cdot s_n 
       \tag{by (2)} \\[1ex]
  &~=~ t_1 \cdot\ldots\cdot t_{i_0-2} ~\cdot~ 
       \pi \big( t_{i_0-1} \cdot\ldots\cdot t_{i_1+1} \big) ~\cdot~ 
       t_{i_1+2} \cdot\ldots\cdot t_n 
       \tag{by (3)} \\[1ex]
  &~=~ t_1 \cdot\ldots\cdot t_{i_0-2} ~\cdot~ 
       t_{i_0-1} \cdot\ldots\cdot t_{i_1+1} ~\cdot~ 
       t_{i_1+2} \cdot\ldots\cdot t_n \ .
       \tag{by (4)}
\end{align*} 
\end{proof}             

We can finally turn to the main proof. Consider two locally consistent sequences 
$\bar{s}$ and $\bar{t}$. By the first claim, $\bar{s}$ and $\bar{t}$ define the 
same domain $D$ and the same equivalence $E$, and the corresponding colourings 
are border-equal. We show that one can transform $\bar{s}$ into $\bar{t}$ by 
repeatedly changing the colouring of one equivalence class (note that it is 
sufficient to describe how the colouring evolves during the transformation
steps). The second claim will then imply that the resulting sequences give 
all the same product.

Let $C_1,\ldots,C_m$ be the non-border equivalence classes of $E$ and let $f_1,\ldots,f_m$ 
be fresh data values, not appearing in $\bar{s}$ or in $\bar{t}$. Transforming $\bar{s}$ 
to $\bar{t}$ is done in two phases. One first performs the following $m$ steps: 
at step $i$, for $i=1,\ldots,m$, one recolours the $i$-th equivalence class $C_i$ by $f_i$. 
One then performs other $m$ steps during which one recolours the $i$-th equivalence class
$C_i$, for $i=1,\ldots,m$, by its colour in $\bar{t}$. 

This concludes the proof of the first part of the lemma.
The second part of the lemma that deals with two almost 
locally consistent sequences $s_1,\ldots,s_n$ and $t_1,\ldots,t_N$ 
follows easily from the existence of a renaming $\tau$ such that 
$s_1,\ldots,s_n$ is locally consistent with $\tau(t_1),\ldots,\tau(t_n)$.
\end{proof}

Lemma \ref{lemma:local-consistency} shows that the product of a sequence of elements,
such as the one induced by an $O$-factorization, does not change if one replaces the 
sequence with another locally consistent one. Our last preparatory lemma shows that 
there always exists a locally consistent sequence over a set of data values of bounded
cardinality. This enables the possibility of guessing a locally consistent sequence 
in MSO, that is, by means of a tuple of monadic predicates.

\begin{lemma}\label{lemma:small-data-locally-consistent}
For every sequence~$\bar{s}$ of elements of $\cM$, there exists a sequence $\bar{t}$ 
that is locally consistent with~$\bar{s}$ and uses at most $4\dlen{\cM}$ distinct data values.
\end{lemma}

\begin{proof}
We can reuse the objects in the proof of Lemma~\ref{lemma:local-consistency}, 
namely, by identifying elements of $\cM$ with terms, we construct from the 
sequence $\bar{s}=s_1,\ldots,s_n$ a domain $V$, an equivalence $E$ and a
corresponding colouring of the equivalence classes of $E$. Moreover, without 
loss of generality, we can assume that:
\begin{itemize}
  \item the set $D$ of all data values is the set of positive integers,
  \item the data values of $s_1$ and~$s_n$ belong to the set $\{1,\ldots,2\dlen\cM\}$ 
        (call this range of values the set of \emph{border colours}),
  \item all other data values, which do not appear in~$s_1$ or~$s_n$,
        are strictly above~$4\dlen\cM$.
\end{itemize}
By using an induction on $i=2,\ldots,n-1$, we transform the colouring induced by 
$\bar{s}$ into a new colouring that associates with all positions $(i',j)$ in the
domain, where $i'\le i$, a data value smaller than or equal to~$4\dlen{\cM}$. 
Such a transformation is performed at each step $i=2,\ldots,n-1$ as follows. 
If some colours greater than~$4\dlen\cM$ are associated with the classes of 
the positions $(i+1,1)$, $\ldots$, $(i+1,\arity(s_{i+1}))$, then we recolour 
these classes with data values ranging over $\{2\dlen{\cM}+1,\ldots,4\dlen{\cM}\}$ 
-- intuitively, this amounts at renaming the data values of $s_{i+1}$. Note that 
we can perform such a recolouring because there are at most $\dlen{\cM}$ data values 
involved in the term~$s_i$ and at most $\dlen{\cM}$ data values involved in the term
$s_{i+1}$, and hence there are always at least $\dlen{\cM}$ fresh data values from 
$\{2\dlen{\cM}+1,\ldots,4\dlen{\cM}\}$ that we can choose.

At the end of the transformation, the resulting colouring uses only data values from 
$\{1,\ldots,4\dlen{\cM}\}$. This means that the sequence of terms $\bar{t}=t_1,\ldots,t_n$ 
that corresponds to this colouring fulfils the claim of the lemma.
\end{proof}

\medskip
We have now all the ingredients to prove the inductive step for claim \refcond{claim2} 
of Lemma \ref{lemma:inductive-statement}. 

\begin{lemma}[Induction step for \refcond{claim2}]\label{lemma:cond2-mso}
For all rigid guards $\alpha(x,y)$ such that~$w\models\alpha(x,y)$ implies 
$h(w[x,y]) \ge_{\orbit\cJ} O$, one can compute the types under~$\alpha$.
\end{lemma}

\begin{proof}
The construction starts as in the proof of Lemma~\ref{lemma:cond2-fo}, but things change at 
point 4., since the construction there exploits aperiodicity. We continue the construction, 
without the assumption of aperiodicity, as follows.
\begin{enumerate}
  \setcounter{enumi}{3}
  \item Using the formula $\varphi^\fact_O(x,y,X',Y')$ from Lemma \ref{lemma:J-factorization}, 
        one guesses a pair $(X',Y')$ of monadic predicates that represent a
        rigidly-traversable $O$-factorization $w_1,\ldots,w_n$ of $w[x,y]$. 
  \item Let $G'$ be the family of formulas with respect to which the defined $O$-factorization 
        is rigidly traversable. For every factor $w_i$, one guesses a corresponding formula
        $\gamma_i$ in $G'$ that holds over the endpoints of the factor $w_i$, namely,
        such that $w \models \gamma_i(x_i,y_i)$, where $x_i$ (resp., $y_i$) is the left 
        (resp., right) endpoint of $w_i$ in $w$. This information can be encoded by
        a tuple of monadic predicates and easily verified to be correct.
  \item One then guesses some terms $t_1,\ldots,t_n$ over a finite set~$C$ of data values
        of cardinality $4\dlen{\cM}$. There are finitely many such terms, so this can be done using 
        a tuple of monadic predicates. Moreover, this is done in such a way that the information 
        concerning the term $t_i$ is located on the factor~$w_i$, say on the first letter.
  \item One also guesses some terms $u_1,\ldots,u_n$ over the same set $C$ of data 
        values as above, where each term $u_i$ is meant to be the product of the first $i$ 
        terms $t_1,\ldots,t_i$. One can represent the terms $u_1,\ldots,u_n$ with new 
        monadic predicates and then check that the guess is correct by verifying that 
        $u_1=t_1$ and $u_{i+1}=u_i\cdot t_i$ for all $1\le i\le n-1$.
  \item One checks that the sequence of terms $t_1,\ldots,t_n$ is almost locally consistent 
        with the sequence of elements that is induced by the $O$-factorization, that is,
        the sequence $h(w_1),\ldots,h(w_n)$. Below we explain how to do so.

        Consider any index~$1\le i\le n-1$ (this amounts at using a universal 
        first-order quantification). Recall from the previous constructions 
        that $\gamma_i$ and $\gamma_{i+1}$ are rigid formulas in $G'$ that 
        hold between the endpoints of $w_i$ and $w_{i+1}$, respectively.
        Also recall, from the statement of Lemma \ref{lemma:J-factorization}, 
        that one can compute the types under the guards $\gamma_i$ and $\gamma_{i+1}$.
        In particular, there are finite families of inward-rigid formulas 
        $F=(\varphi_o)_{o\in O}$ and $F'=(\varphi'_o)_{o\in O}$ that compute 
        the types under the guards $\gamma_i$ and $\gamma_{i+1}$.
        By the sub-definability Lemma \ref{lemma:subdefinability}, 
        we can rewrite each inward-rigid formula $\varphi_o(x',\bar{z},y')$ 
        (resp., $\varphi'_o(x',\bar{z},y')$) as a finite disjunction of rigid 
        formulas $\varphi_{o,1}(x',\bar{z},y') ~\vel~ \ldots ~\vel~ \varphi_{o,\ell}(x',\bar{z},y')$
        (resp., $\varphi'_{o,1}(x',\bar{z},y') ~\vel~ \ldots ~\vel~ \varphi_{o,m}(x',\bar{z},y')$).
                
        We first verify, using disjunctions, that the orbits of $h(w_i)$ and $t_i$ coincide, 
        namely, that some formula $\varphi_{o,j}(x_i,\bar{z},y_i)$ holds over the endpoints 
        $x_i,y_i$ of $w_i$ and over some positions $\bar{z}$, where $o$ is the orbit of the 
        term $t_i$. We do the same for $h(w_{i+1})$ and $t_{i+1}$, namely, we verify that 
        some formula $\varphi'_{o',j'}(x_{i+1},\bar{z}',y_{i+1})$ holds over the endpoints 
        $x_{i+1},y_{i+1}$ of $w_{i+1}$ and over some positions $\bar{z}'$, where $o'$ is the 
        orbit of the term $t_{i+1}$.
        
        Next, we verify that the equalities between the memorable values of $h(w_i)$ 
        and $h(w_{i+1})$ are the same as the equalities between the data values of 
        the terms $t_i$ and $t_{i+1}$. For this, we suppose that $t_i=o(d_1,\ldots,d_k)$ 
        and $t_{i+1}=o'(d'_1,\ldots,d'_{k'})$ and we check that for all pairs of data 
        values $d_h,d'_{h'}$:
        \begin{enumerate}
          \item if~$d_h=d'_{h'}$, 
                then $(\beta\cdot\beta')(z_h,z'_{h'}) \,\et\, z_h \sim z'_{h'}$ holds, where 
                $\beta(z_h,y') \:\eqdef\: \varphi_{o,j}(\star,\bar{\star},z_h,\bar{\star},y')$,
                $\beta'(x',z'_{h'}) \:\eqdef\: 
                                    \varphi'_{o',j'}(x',\bar{\star},z'_{h'},\bar{\star},\star)$,
                and $\beta\cdot\beta'$ is defined as in Lemma \ref{lemma:products}
                (note that $\beta\cdot\beta'$ is rigid since 
                 $\varphi_{o,j}$ and $\varphi'_{o',j'}$ are rigid);
          \item if~$d_h\neq d'_{h'}$, 
                then $(\beta\cdot\beta')(z_h,z'_{h'}) \,\et\, z_h \nsim z'_{h'}$ holds, where 
                $\beta$, $\beta'$, and $\beta\cdot\beta'$ are defined as in (a).
        \end{enumerate}

        Finally, it remains to check that the equalities between the memorable values of 
        $h(w_1)$ and $h(w_{i+1})$ are the same as the equalities between the data values 
        of $t_1$ and $t_n$. For this, we assume that $t_1=o(d_1,\ldots,d_k)$ and 
        $t_n=o(d'_1,\ldots,d'_{k'})$ and we check that for all pairs of data values 
        $d_h,d'_{h'}$:
        \begin{enumerate}
          \setcounter{enumii}{2}
          \item if~$d_h=d'_{h'}$, 
                then $\beta(z_h,z'_{h'}) \,\et\, z_h\sim z'_{h'}$ holds,
                where $\beta(z_h,z'_{h'})$ is defined as
                      $\exists x',y' ~ \varphi_{o,j}(x',\bar\star,z_h,\bar\star,\star) \,\et\,
                       \alpha(x',y') \,\et\,
                       \varphi_{o',j'}(\star,\bar\star,z'_{h'},\bar\star,y')$ 
                -- note that $\beta$ is rigid since $\varphi_{o,j}$, $\varphi'_{o',j'}$,
                and $\alpha$ are all rigid;
          \item if~$d_h\neq d'_{h'}$, 
                then $\beta(z_h,z'_{h'}) \,\et\, z_h\nsim z'_{h'}$ holds,
                where $\beta(z_h,z'_{h'})$ is defined as in (c).
        \end{enumerate}
  \item At this point, we know that $h(w[x,y])=h(w_1)\cdot\ldots\cdot h(w_n)$ 
        is in the same orbit as $u_n=t_1\cdot\ldots\cdot t_n$.
        What remains to be done is to determine some positions in $w[x,y]$ 
        that carry the memorable values of $h(w[x,y])$.

        For this, we use once more Lemma \ref{lemma:decomposition-stairs} and 
        Proposition~\ref{prop:memorable-values}.
        Since the elements $h(w_1)$ and $h(w[x,y])$ belong to the same 
        $=_{\orbit\cJ}$-class $O$ and since $h(w_1) \,\geq_\cR\, h(w[x,y])$,
        Lemma \ref{lemma:decomposition-stairs} implies $h(w_1) \,=_\cR\, h(w[x,y])$.
        In a similar way, we derive $h(w_n) \,=_\cL\, h(w[x,y])$.
        Now, Proposition \ref{prop:memorable-values} implies that every memorable 
        value of~$h(w[x,y])$ is also memorable in~$h(w_1)$ or in~$h(w_n)$. 
        Thus, a formula can easily locate in a rigid way some positions with the memorable 
        values of~$h(u_1)$ and~$h(u_n)$. 

        Towards a conclusion, we recall that the two sequences $h(w_1),\ldots,h(w_n)$
        and $t_1,\ldots,t_n$ are almost locally consistent. This means that there 
        is a correspondence between the memorable values of $h(w[x,y])$ (resp., 
        $h(w_1)$, $h(w_n)$) and the data values of $u_n$ (resp., $t_1$, $t_n$). 
        In particular, for every memorable value $e$ of $h(w[x,y])$, one can identify 
        a corresponding data value $d$ that occurs in $u_n$, and hence in $t_1$ or $t_n$. 
        From this, one determines in a rigid way a corresponding position in $w_1$ or 
        $w_n$ that carries the memorable value $e$.\qedhere
\end{enumerate}
\end{proof}

\section{Logics for finite memory automata}\label{sec:logics-for-automata}

In this section, we consider again a variant of \MSO with guards on data tests.
More precisely, we consider the logic \emph{semi-rigidly guarded \MSO}, as it 
was introduced in Section \ref{sec:logics}, but now interpreted it over the 
class of data words\footnote{Results similar to those presented in this section
can be obtained for semi-rigidly guarded \MSO interpreted over ranked data trees. 
However, we prefer to avoid dealing with the technicalities concerning trees and 
we present instead the results for data words as a proof of concept.}. 
The goal is to relate definability in semi-rigidly guarded \MSO to recognizability 
of data languages by means of finite memory automata 
\cite{finite_memory_automata,nondeterministic_reassignment}. 

Recall that, in semi-rigidly guarded \MSO, every data test $y\sim z$ 
is conjoint with a formula $\alpha(x,y) ~\et~ \beta(x,z)$, where $\alpha(x,y)$ 
(resp., $\beta(x,z)$) is a formula of semi-ridigly guarded \MSO that determines
at most one position $y$ (resp., $z$) from each position $x$.
Below, we recall the definition of finite memory automaton, precisely the one 
from \cite{automata_with_group_actions} which is closer in spirit to orbit-finite 
data monoids. For this we reuse some of the notions introduced in Section 
\ref{sec:background}, in particular, that of $\cG_D$-set and that of equivariant subset.
All elements $s$ in a $\cG_D$-set are tacitly assumed to have finite memory,
denoted for short by $\mem(s)$. For simplicity, we also avoid denoting explicitly 
the group actions underlying $\cG_D$-sets.

\begin{definition}\label{def:finite-memory-automaton}
A \emph{finite memory automaton} (\emph{FMA}) is a tuple $\cA=(\Sigma,Q,I,F,T)$, 
where:
\begin{itemize}
  \item $\Sigma$ is an orbit-finite $\cG_D$-set that represents the input alphabet 
        (e.g., $D\times A$); 
  \item $Q$ is an orbit-finite $\cG_D$-set that represents the configuration space;
  \item $I$ and $F$ are equivariant subsets of $Q$ (i.e., unions of orbits of $Q$) 
        that contain the initial and final configurations, respectively;
  \item $T$ is an equivariant subset of $Q\times\Sigma\times Q$ that describes 
        the possible transitions.
\end{itemize}
The notion of \emph{successful run} of $\cA$ and that of \emph{recognized language} 
are the usual ones for automata running on finite words (they are just generalized 
to possibly infinite alphabets and configuration spaces).
An FMA $\cA=(\Sigma,Q,I,F,T)$ is said to be
\begin{enumerate}
  \item \emph{unambiguous}, 
        if it admits at most one successful run on each input data word;
  \item \emph{deterministic}, 
        if $I$ is a singleton and $(p,c,q),(p,c,q') \in T$ implies $q=q'$;
  \item \emph{co-deterministic}, 
        if $F$ is a singleton and $(p,c,q),(p',c,q) \in T$ implies $p=p'$.
\end{enumerate}
\end{definition}

FMA can be easily shown to be closed under union and intersection (the classical 
constructions for finite state automata can be used). In addition, deterministic
FMA and co-deterministic FMA are closed under complementation. However, FMA 
cannot be determinized and are not closed under complementation. We also know
that FMA (even deterministic or co-deterministic ones) are strictly more expressive 
than orbit-finite data monoids. For example, a separating language is
$$
  L_\curvearrowright
  ~~\eqdef~ 
  \{d_1\ldots d_n ~:~ \exists 1<i\le n ~ d_1=d_i\}
$$
which is recognized by a deterministic FMA, but is not recognizable by an orbit-finite 
data monoid. Moreover, unlike classical finite automata, deterministic and 
co-deterministic FMA are incomparable, and both are strictly less expressive 
than unambiguous FMA, as easily witnessed by the language $L_\curvearrowright$ 
and its reverse.

With the lack of classical tools like the subset construction, the problem of finding 
logical characterizations of classes of data languages recognized by FMA becomes challenging. 
In the following, we give partial results towards this goal.

\medskip
The first result shows that FMA, even unambiguous ones, are at least as expressive
as semi-rigidly guarded \MSO interpreted on data words.

\begin{theorem}\label{thm:semi-rigid-inside-fma}
Every language of data words defined by a semi-rigidly guarded \MSO sentence 
is effectively recognized by an unambiguous FMA.
\end{theorem}

In proving the above result, it is quite natural to try an approach similar to 
that of Theorem \ref{thm:formula-to-monoid}. Unfortunately, it is difficult to adapt 
the approach to the present setting, because here there are semi-rigid guards that, 
from a certain position $x$, determine at most one position $y$ either to the left or to 
the right of $x$, and hence one cannot use deterministic or co-deterministic automata.
Nonetheless, we can exploit (unambiguous) non-determinism of FMA to annotate
the input data words with additional symbols and data values that ease the
translation. 

We begin by giving a simple lemma that shows that, as for classical languages, 
it is possible to reason about recognizability by unambiguous FMA modulo annotations 
computed by length-preserving functions on data words.
We fix the following notation. Let $\Sigma$ and $\Delta$ be two orbit-finite alphabets
(e.g., $\Sigma = A\times D$ and $\Delta = B\times D$). Given two data words $u\in\Sigma^*$ 
and $v\in\Delta^*$ of the same length, we denote by $u\otimes v$ the \emph{convolution} of 
$u$ and $v$, defined by $(u\otimes v)(i)=\big(u(i),v(i)\big)$ for all $1\le i\le|u|=|v|$.

\begin{lemma}\label{lemma:annotations}
Let $f:\Sigma^*\then\Delta^*$ be a function that maps any data word $u$ to a data word $f(u)$ 
of the same length, let $L\subseteq\Sigma^*$ be a data language, and let $L_f$ be the language
of data words of the form $u\otimes v$, where $u\in L$ and $v=f(u)$.
If $L_f$ is recognized by an unambiguous FMA, the so is $L$.
\end{lemma}

\begin{proof}
Let $\cA_f=(\Sigma\times\Delta,Q,I,F,T)$ be an unambiguous FMA that recognizes $L_f$.
Define $\cA=(\Sigma,Q,I,F,T')$, where $T'$ contains all and only the transitions of 
the form $(p,a,q)$ such that $\big(p,(a,b),q\big)\in T$ for some $b\in\Delta$.
Clearly $\cA$ recognizes $L$. Moreover, $\cA$ is unambiguous: if $u\in L$, 
then there is a unique $v\in\Delta^*$ such that $u\otimes v\in L_f$, and
there is a unique successful run of $\cA$ on $u\otimes v$.
\end{proof}

Below, we fix a sentence $\psi$ of semi-rigidly guarded \MSO and we prove the data 
language defined by $\psi$ is recognized by an unambiguous FMA. 
To ease the translation, we will use Lemma \ref{lemma:annotations} to freely annotate
data words with the outputs computed functional finite memory transducers. 
Formally, we define a {\em functional finite memory transducer} ({\em FMT} for short) 
just as an unambiguous FMA over the alphabet $\Sigma\times\Delta$ that never accepts 
two data words of the form $u\otimes v$ and $u\otimes v'$, with $v\neq v'$.

The first step consists of rewriting $\psi$ into a sentence of classical MSO 
(hence containing no data tests), but interpreted over data words with appropriate 
annotations. This will reduce the problem to constructing some FMTs computing 
the appropriate annotations. We use a technique similar to the proof of 
Theorem \ref{thm:decidability}, namely, we substitute in $\psi$, starting from 
the innermost subformulas, every occurrence of a semi-rigidly guarded data test 
$\varphi(x,y,z) ~=~ \alpha(x,y) ~\et~ \beta(x,z) ~\et~ y\sim z$, 
with the formula 
$\varphi^*(x,y,z) ~=~ \alpha^*(x,y) ~\et~ \beta^*(x,z) ~\et~ x\in c_{\alpha,\beta}^\sim$,
where $c_{\alpha,\beta}^\sim$ is a fresh unary predicate associated with that occurrence.
Let $\psi^*$ be the sentence of classical MSO that is obtained from $\psi$ by applying 
the above transformation and let $C$ be the alphabet with the new predicates
$c^\sim_{\alpha,\beta}$. We have that for every data word $u\in(A\times D)^*$
$$
  u \sat \psi
  \qquad\text{iff}\qquad
  u\otimes v \sat \psi^*
$$
where $v$ is the unique annotation for $u$ that satisfies
$$
  \theta_{\alpha,\beta}
  ~\eqdef~
  \forall x ~ 
  \big(\, x\in c_{\alpha,\beta}^\sim ~\leftrightarrow~ 
                  \exists y,z ~ 
                  \alpha^*(x,y) \:\et\: \beta^*(x,z) \:\et\: y\sim z \,\big)
$$
for each $c_{\alpha,\beta}^\sim\in C$.
The MSO sentence $\psi^*$ can be translated to a deterministic finite state automaton 
$\cA^*$, and the latter can be intersected with suitable FMTs recognizing the languages 
defined by $\theta_{\alpha,\beta}$. Thus, we have reduced the problem of constructing 
an unambiguous FMA for $\psi$ to the problem of constructing FMTs for the sentences $\theta_{\alpha,\beta}$. 

Now, we focus our attention on the sentences $\theta_{\alpha,\beta}$.
First of all, note that a sentence $\theta_{\alpha,\beta}$ contains no nested data tests;
in particular, semi-rigid guards are classical MSO formulas.
In order to further ease the translation, we apply a second normalization step. 
This time the goal is to rewrite $\theta_{\alpha,\beta}$ so as to avoid the presence 
of semi-rigid guards with {\em crossing patterns}, namely, guards that on the same 
annotated word $w=u\otimes v$, determine distinct positions $y$ and $y'$ from $x$ and $x'$, 
respectively, such that 
$$
  \big[\min(x,y),\max(x,y)\big] ~\cap~ \big[\min(x',y'),\max(x',y')\big] ~\neq~ \emptyset.
$$
We formalize below the key argument that enables such a transformation:

\begin{lemma}\label{lemma:cross-free}
Every semi-rigid MSO formula $\varphi(x,y)$ which determines $y$ from $x$ can be rewritten 
as a finite disjunction of semi-rigid MSO formulas $\varphi_1(x,y),\ldots,\varphi_n(x,y)$ 
that are {\em cross-free}, namely, such that for all words $w$, all positions $x,x',y,y'$ 
in it, and all indices $i=1,\ldots,n$:
$$
  \begin{cases}
    ~w \sat \varphi_i(x,y) \\
    ~w \sat \varphi_i(x',y') \\
    ~\big[\min(x,y),\max(x,y)\big] \cap \big[\min(x',y'),\max(x',y')\big] \neq \emptyset
  \end{cases}
  \quad\text{implies}\qquad
  y = y'.
$$
\end{lemma}

\begin{proof}
Without loss of generality, suppose that $\varphi(x,y)$ always entails $x\le y$
(if this is not the case, rewrite $\varphi$ as a disjunction of two formulas
that entail respectively $x\le y$, $y\le x$, and then prove the lemma 
separately for each disjunct).
The lemma follows essentially from compositionality of MSO \cite{composition_method_shelah}.
Indeed, we can reuse the same argument as in the beginning of the proof of the
sub-definability Lemma \ref{lemma:subdefinability} (cf. Claim \ref{claim:composition})
to show that the formula $\varphi(x,y)$ is equivalent to a finite disjunction of 
the form
$$
  \bigvee_{i=1\ldots n}
  \exists z ~~ \alpha_i(x,z) ~\et~ \beta_i(z,y) ~\et~ x\le z\le y.
$$
for some MSO formulas $\alpha_i(x,z)$, $\beta_i(z,y)$.

Thus, the desired formulas $\varphi_i(x,y)$ are nothing but the disjuncts above, 
namely, the semi-rigid formulas 
$\exists z \: \alpha_i(x,z) \,\et\, \beta_i(z,y) \,\et\, x\le z\le y$,
for $i=1\ldots n$.
Indeed, suppose that there are positions $x,x',y,y'$ in $w$ such that
$w \sat \varphi_i(x,y)$, $w \sat \varphi_i(x',y')$, and 
$\big[\min(x,y),\max(x,y)\big] \cap \big[\min(x',y'),\max(x',y')\big] \neq \emptyset$.
The latter condition implies the existence of a position $z$ satisfying
$x\le z\le y$, $x'\le z\le y'$, $\alpha_i(x,z) \,\et\, \beta_i(z,y)$, and
$\alpha_i(x',z) \,\et\, \beta_i(z,y')$. It follows that $z$ satisfies 
also $\alpha_i(x,z) \,\et\, \beta_i(z,y')$, and hence 
$w \sat \varphi_i(x,y')$. Finally, since $\varphi_i$ is semi-rigid, we
deduce that $y=y'$.
\end{proof}

We can apply Lemma \ref{lemma:cross-free} to all semi-rigid guards
$\alpha^*(x,y)$ and $\beta^*(x,z)$ in a sentence $\theta_{\alpha,\beta}$ 
and then commute disjunctions and existential quantifications. In this way 
we obtain a sentence equivalent to $\theta_{\alpha,\beta}$:
$$
  \theta'_{\alpha,\beta}
  ~\eqdef~
  \forall x ~ 
  \big(\, x\in c_{\alpha,\beta}^\sim ~\leftrightarrow~ 
                  \bigvee\nolimits_{\begin{smallmatrix} i=1\ldots n \\ 
                                                        j=1\ldots m \end{smallmatrix}}
                  \exists y,z ~~ 
                  \alpha^*_i(x,y) \:\et\: \beta^*_j(x,z) \:\et\: y\sim z \,\big)
$$
where $\alpha^*_i(x,y)$ and $\beta^*_j(x,z)$ are semi-rigid cross-free formulas. 
The following definition and lemma give the additional ingredients to complete 
the translation of $\theta'_{\alpha,\beta}$ into an equivalent FMT.

\begin{definition}\label{def:trace}
Let $\varphi(x,y)$ be a semi-rigid cross-free formula and let $w$ be a data word. 
The {\em trace of $\varphi$ on $w$} is the annotated data word $w\otimes t_\varphi$, 
where $t$ is the word over the orbit-finite alphabet $\{\bot\}\uplus D$ 
such that, for all $1\le z\le |w|=|t|$,
$$
  t_\varphi(z) ~\eqdef~ \begin{cases}
                          d & \text{if there are $x,y$ such that $z\in[\min(x,y),\max(x,y)\big]$,} \\
                            & \text{\phantom{if} $w\sat\varphi(x,y)$, and $\mem(w(y))=\{d\}$} \\[1ex]
                          \bot & \text{otherwise}.
                        \end{cases}
$$
(note that $d$ above is well defined because $\varphi$ is semi-rigid and cross-free).
\end{definition}

\begin{lemma}\label{lemma:trace}
For every semi-rigid cross-free formula $\varphi(x,y)$, one can construct 
an FMT $\cT$ that defines the traces of $\varphi$ on the input data words.
\end{lemma}

\begin{proof}
For the sake of brevity, let us call {\em memorable position of $z$ in $w$}
the unique position $y$ (if there is any) for which there exists $x$ such that 
$z\in\big[\min(x,y),\max(x,y)\big]$ and $w\sat\varphi(x,y)$.
The memorable position of $z$ is clearly MSO-definable from $z$, and so are 
the following properties parametrized by $z$: ``$z$ has some memorable position'',
``$z$ is the memorable position of itself'', and
``$z$ and $z+1$ have the same memorable position''.
As usual, we can assume without loss of generality that the latter properties 
are explicitly encoded on the positions $z$ of an input data word $w$ by means 
of additional bits of information (we recall that unambiguity is preserved when 
we project out these bits).

It easy to construct a deterministic FMA that parses a word 
$w\otimes t$ and performs the following actions on the basis 
of the bits of information at position $z$:
\begin{itemize} 
  \item it checks that $t(z)\neq\bot$ iff $z$ has some memorable position,
  \item if $z$ is the memorable position of itself, then it also verifies that
        the data value in $t(z)$ is equal to the data value in $w(z)$,
  \item finally, if $z$ has the same memorable position as $z+1$, then it stores
        the data value of $t(z)$ and on the next position verifies
        that the value is the same as the one in $t(z+1)$.
\end{itemize}
We omit the formal specification of the deterministic FMA and we only remark 
that it accepts precisely the words of the form $w\otimes t_\varphi$, when the 
bits of information on $z$ are accessible. Moreover, the latter bits can be 
produced by a functional transducer, so using Lemma \ref{lemma:annotations} 
we can easily obtain an unambiguous FMA recognizing the desired language.
\end{proof}

We can now conclude the proof of Theorem \ref{thm:semi-rigid-inside-fma}.
We recall that the sentences $\theta'_{\alpha,\beta}$ are of the form 
$$
  \forall x ~ 
  \big(\, x\in c_{\alpha,\beta}^\sim ~\leftrightarrow~ 
                  \bigvee\nolimits_{\begin{smallmatrix} i=1\ldots n \\ 
                                                        j=1\ldots m \end{smallmatrix}}
                  \exists y,z ~~ 
                  \alpha^*_i(x,y) \:\et\: \beta^*_j(x,z) \:\et\: y\sim z \,\big)
$$
Assuming that the above sentences are evaluated on data words annotated
with the traces $t_{\alpha^*}$ and $t_{\beta^*}$, we can rewrite them as
$$
  \forall x ~ 
  \big(\, x\in c_{\alpha,\beta}^\sim ~\leftrightarrow~ 
                  \bigvee\nolimits_{\begin{smallmatrix} i=1\ldots n \\ 
                                                        j=1\ldots m \end{smallmatrix}}
          t_{\alpha^*}(x) = t_{\beta^*}(x) \neq \bot \,\big)
$$
which can be easily verified by a deterministic (hence unambiguous) FMA. Putting everything 
together, we exploit closure of unambiguous FMA under unions, intersections, and projections
of FMT-definable annotations to obtain an unambiguous FMA that recognizes the language 
defined by $\psi$.

\medskip
The second result shows that semi-rigidly guarded \MSO does not capture 
the entire class of data languages recognizable by unambiguous FMA. 
In fact, the separating language is even recognized by a deterministic FMA:

\begin{proposition}\label{prop:separating-example}
There is a data language recognized by a deterministic FMA that cannot be 
defined in semi-rigidly guarded \MSO.
\end{proposition}

\begin{proof}
For the sake of simplicity, we will assume that data values range over
the natural numbers.
Consider the language $L_{\curvearrowright^*}$ that contains all and 
only the data words over the alphabet $\bbN$ of the form 
\footnote{This language is a variant of an example given in 
\cite{graph_reachability_with_pebbles} to study data languages 
recognized by pebble automata.} 
\vspace{4mm}
$$
  ~\td{node1}{}\!\!\!\!d_{i_0+1}%
  ~\ldots%
  ~\td{node2}{}\!\!\!\!d_{i_1}%
  ~\td{node3}{}\!\!\!\!d_{i_1+1}%
  ~\ldots%
  ~\td{node4}{}\!\!\!\!d_{i_2}%
  ~\ldots\ldots
  ~\td{node5}{}\!\!\!\!d_{i_{\ell-1}+1}%
  ~\ldots
  ~\td{node6}{}\!\!\!\!d_{i_\ell}
  \tl{node1.90}{out=60,in=120,shorten <=1mm,shorten >=1mm}{node2.90}%
  \tl{node3.90}{out=60,in=120,shorten <=1mm,shorten >=1mm}{node4.90}%
  \tl{node5.90}{out=60,in=120,shorten <=1mm,shorten >=1mm}{node6.90}%
$$
where $\ell\in\bbN$, $0=i_0<i_1<i_2<\ldots<i_{\ell-1}<i_\ell=n$, 
$d_{i_j + 1}=d_{i_{j+1}}$ for all $0\le j<\ell$, and $d_k\neq d_{i_j}$ 
for all $i_j+1<k<i_{j+1}$
(for the sake of readability, we added arcs linking those data values
that are required to be equal).
It is easy to see that $L_{\curvearrowright^*}$ is recognized by a deterministic 
FMA: at each phase, the automaton stores the value under the current position and 
then moves to the right looking for another occurrence of the stored value; if it 
does not find such an occurrence, then it rejects, otherwise, as soon as the occurrence 
is found, the automaton moves to the next position (if there is any, otherwise it accepts) 
and starts a new phase. 

Below, we fix a generic sentence $\psi$ of semi-rigidly guarded \MSO and 
we show that it cannot define the language $L_{\curvearrowright^*}$. 
For this, we consider data words in $L_{\curvearrowright^*}$ of the form
\vspace{4mm}
$$
  w_n ~~=~~
  \td{bridge1}{1} ~u_n^{(1)}~ \td{bridge2}{1}~%
  \td{bridge3}{2} ~u_n^{(2)}~ \td{bridge4}{2}~%
  \ldots\ldots~%
  \td{bridge5}{n} ~u_n^{(n)}~ \td{bridge6}{n}%
  \tl{bridge1.90}{out=60,in=120,shorten <=1mm,shorten >=1mm}{bridge2.90}%
  \tl{bridge3.90}{out=60,in=120,shorten <=1mm,shorten >=1mm}{bridge4.90}%
  \tl{bridge5.90}{out=60,in=120,shorten <=1mm,shorten >=1mm}{bridge6.90}%
$$
where each $u_n^{(i)}$ has length exactly $n$ symbols and the only equalities 
between data values are those represented by the arcs, namely, the juxtaposition 
$u_n^{(1)}\, u_n^{(2)}\, \ldots\, u_n^{(n)}$ contains pairwise distinct 
values from the set $\bbN\setminus\{1,\ldots,n\}$. 
We aim at proving that, for $n$ sufficiently large, almost all semi-rigidly 
guarded data tests performed by $\psi$ fail, thus making it impossible to 
distinguish $w_n$ from other data words outside $L_{\curvearrowright^*}$.

Our argument will make extensive use of the encodings of data tests 
described in Theorem \ref{thm:semi-rigid-inside-fma}. We begin by introducing
the alphabet $C$ consisting of one symbol $c^\sim_{\alpha,\beta}$ for each 
data test in $\psi$ of the form $\alpha(x,y) \,\et\, \beta(x,z) \,\et\, y\sim z$.
We then transform the data word $w_n$ into a classical word $w_n^*$ over $C$ 
by labelling every position $x$ of $w_n^*$ with the set of symbols 
$c^\sim_{\alpha,\beta}$ such that 
$w_n\sat\exists y,z ~~ \alpha(x,y) \,\et\, \beta(x,z) \,\et\, y\sim z$.
Accordingly, we transform every sub-formula $\varphi$ of $\psi$ to a
classical MSO formula $\varphi^*$ that is equivalent in the following
sense:
$$
  w_n\sat\varphi
  \qquad\text{iff}\qquad
  w^*_n\sat\varphi^*
$$
(in particular, $\alpha(x,y) \,\et\, \beta(x,z) \,\et\, y\sim z$ is transformed 
into $\alpha^*(x,y) \,\et\, \beta^*(x,z) \,\et\, x\in c^\sim_{\alpha,\beta}$).

Now, we prove that every semi-rigid guard $\alpha(x,y)$ in $\psi$ defines a 
position $y$ from $x$ that is either close to $x$ or close to one of the 
endpoints of $w_n$. This is formally stated in the following claim:

\begin{claim}
For every semi-rigid guard $\alpha(x,y)$, there is a number $k_\alpha$ that 
depends only on $\alpha$ (and not on $n$) and satisfies
$$
  \forall n>1,~ 1\le x\le|w_n| \quad
  w_n\sat\alpha(x,y) 
  \quad\text{implies}\quad
  \begin{array}{ll}
    1) & y \le (n+2)\cdot k_\alpha, \quad \text{or} \\
    2) & y \ge (n+2)\cdot (n-k_\alpha), \quad \text{or} \\
    3) & |y-x| \le k_\alpha.
  \end{array}
$$
\end{claim}

\noindent
The proof of the above claim is by induction on the number of data tests
in $\alpha(x,y)$. 

For the base case, we suppose that $\alpha(x,y)$ is a semi-rigid guard without 
data tests, namely, a formula of classical MSO defining a partial function. 
The claim follows essentially from the fact that there are no distinguished 
symbols in $w_n$ that can be used to determine from $x$ a position $y$ that is 
far both from $x$ and from the endpoints of $w_n$.
Let $\cA$ be a finite state automaton recognizing 
the language over $\{0,1\}\times\{0,1\}$ defined by $\alpha(x,y)$.
Suppose that the claim above does not hold, namely, that for all $k\in\bbN$, 
there are $n\in\bbN$ and $1\le x,y\le |w_n|$ such that 
$k < y < |w_n| - k$, $|y-x| > k$, and $w_n\sat\alpha(x,y)$.
Note that the latter condition $w_n\sat\alpha(x,y)$ can be equally stated as
$\ang{w_n^-,\{x\},\{y\}}\in \sL(\cA)$, where $w_n^-$ is the word of length 
$|w_n|$ over a singleton alphabet. 
A simple pumping argument shows that there exist $n\in\bbN$ and some non-empty 
factors $u_n,v_n$ of $\ang{w_n^-,\{x\},\{y\}}$ such that (i) $u_n$ occurs strictly
between the positions $x$ and $y$, $v_n$ occurs to the right of both $x$
and $y$, and erasing or repeating $u$ and $v$ results in new words that are also
accepted by $\cA$. We distinguish two cases depending on whether $|u|=|v|$ or not. 
In the former case, by removing the factor $u$ and by repeating twice the factor $v$
we obtain a new word of the same length as $\ang{w_n^-,\{x\},\{y\}}$ that 
is also accepted by $\cA$, but that identifies a new pair of positions 
$(x,y')$, with $y'\neq y$. This contradicts the fact that $\alpha(x,y)$ functionally
defines $y$ from $x$.
In the latter case, we can reach the same contradiction by swapping the order
of the two factors $u$ and $v$.
The above arguments prove that the claim holds for the considered semi-rigid 
guard $\alpha(x,y)$ without data tests. 

As for the inductive case, suppose that $\alpha(x,y)$ uses some data tests 
based on semi-rigid guards $\beta_1,\ldots,\beta_h$. Assume that the inductive hypothesis
holds for the guards $\beta_1,\ldots,\beta_h$, let $k_{\beta_1},\ldots,k_{\beta_h}$
be the corresponding constants, and define $k_\beta$ to be the maximum of these
constants. Thanks to the inductive hypothesis, we know that for $n$ sufficiently large, 
say $n\ge 2k_\beta$, every data test performed by $\alpha(x,y)$ between pairs of positions 
determined from $x'\in \big\{ (n+2)\cdot k_\beta+1, \ldots, (n+2)\cdot(n-k_\beta) \big\}$
are bound to fail -- indeed, the only way a data test could succeed is by comparing 
positions of $w_n$ that are connected by an arc, but the arcs span large distances
by construction.
Now, recall that $w_n\sat\alpha(x,y)$ iff $w^-_n\sat\alpha^-(x,y)$, where $w^-_n$
is the word annotated with the outcomes of the data tests performed by $\alpha(x,y)$.
In particular, we have that $w^-_n$ is almost constant, that is,
for all positions $x'\in \big\{ (n+2)\cdot k_\beta+1, \ldots, (n+2)\cdot(n-k_\beta) \big\}$,
$c^\sim_{\beta_i,\beta_j} \nin w^-_n(x')$. Finally, by arguing like in the proof of the 
base case, we can conclude that there is a sufficiently large number $k_\alpha$ that 
satisfies the statement for $\alpha(x,y)$ and for all $n\in\bbN$.

We have just proved that the semi-rigid guards $\alpha(x,y)$ in $\psi$ define positions that 
are either close to $x$ or close to one of the endpoints of $w_n$. This means that,
for $n$ large enough, almost all data tests performed by $\psi$ fail. In particular,
there exist $n\in\bbN$ and two positions $y = (n+2)\cdot i + 1$ and $z = (n+2)\cdot(i+1)$
that carry the same data value $i$, but are never tested in $\psi$. This means that 
the word $w'_n$ obtained from $w_n$ by replacing the data value in $z$ with a fresh
value, also satisfies $\psi$. However, $w'_n$ is not in the language $L_{\curvearrowright^*}$,
and hence $\psi$ cannot define $L_{\curvearrowright^*}$.
\end{proof}

\section{Conclusion and future work}\label{sec:conclusion}

We have shown that the algebraic notion of orbit-finite data monoid
corresponds to a variant of monadic second-order logic which is -- 
and this is of course subjective -- natural. It is natural in the 
sense that it only relies on a single and understandable principle: 
guarding data equality tests by rigidly definable relations. 

Of course, it is not the first time the principle of guarding non-monadic predicates 
with suitable formulas is used as a mean of taming the expressiveness of a logic and 
recover decidability. What is more original and interesting in the present context is 
the equivalence with the algebraic object, which shows that this approach is in some 
sense maximal: it is not just a particular technique among others for having decidability, 
but it is sufficient for completely capturing the expressive power of the very natural 
algebraic model.

Another contribution of the present work is the development of the structural
understanding of orbit-finite data monoids. By structural understanding, we 
refer to Green's relations, which form a major tool in most involved proofs 
concerning finite monoids.
The corresponding study of Green's relations for orbit-finite data monoids
was already a major argument in the proof of \cite{data_monoids}, and it
had to be developed even further in the present work.

\medskip
We are only at the beginning of understanding the various notions of recognizability 
for data languages. However, several interesting questions were raised during our study.
Some of them concern the fine structure of the logic:

\begin{quote}
\em
The nesting level of guards seems to be a robust and relevant parameter in our logic.
Can we understand it algebraically? Can we decide it?

\medskip\noindent
Also recall that, in the classical setting of languages over finite alphabets,
there exist effective characterizations of fragments of first-order logic 
within the class of regular languages.
To what extent can we generalise these results in the presence of data over infinite alphabets?
\end{quote}

\noindent
Other questions are of purely algebraic nature:

\begin{quote}
\em
We used in our proofs the notion of term-based presentation of an orbit-finite data monoid. 
(cf. Definition \ref{def:term-rep-system}). 
Such a presentation can be in a simple form, where the congruence $\dsim$ is trivial,
namely, where any two terms $o(d_1,\ldots,d_n)$ and $o'(d'_1,\ldots,d'_n)$ represent
the same monoid element only if they are syntactically equal.
Is it the case that every orbit-finite data monoid is the quotient of some orbit-finite 
data monoid having a simple term-based presentation?
\end{quote}

\noindent
We can answer positively the above question. Indeed, by Theorem \ref{thm:monoids-to-logic}
the elements $s$ of an orbit-finite data monoid $\cM$ can be translated to rigidly guarded 
\MSO sentences $\psi_s$ defining the set of data words whose products evaluate to $s$.
We also know from Corollary \ref{cor:sentence-to-monoid} that the languages defined by 
$\psi_s$ are recognized by an orbit-finite data monoid $\cM'$. 
Moreover, a close inspection to the proof of Corollary \ref{cor:sentence-to-monoid}
reveals that the inductive construction of $\cM'$ can be performed at the level of 
term-based presentations of simple form.

\medskip
We finally considered more powerful notions of recognizability, 
such as those obtained by extending finite state automata with registers
\cite{finite_memory_automata,nondeterministic_reassignment,%
finite_memory_automata_on_trees,automata_with_group_actions}:

\begin{quote}
\em
Does there exist a larger class of data languages that is as robust as that of 
orbit-finite data monoid, and gives, in particular, closures under all Boolean 
operations and restricted forms of projections? Can we match new notions of 
recognizability with suitable logical formalisms? 
\end{quote}

\noindent
In particular, we left open the problem of finding a logic that captures precisely 
the class of data languages recognized by unambiguous FMA, which is a candidate
model for a robust class of data languages. As a matter of fact, in \cite{rigid-mso} 
we described a logic similar to semi-rigidly guarded \MSO that captures data languages 
recognized by non-deterministic FMA. 
However, that logic was not natural, in the sense that it was not closed 
under negation, and, moreover, did not ease characterizations of sub-classes 
of data languages such as those definable in first-order logic.
Regarding the latter problem, we also recall a result from \cite{automata_vs_logics} 
that shows that the problem of determining whether a language recognized by a 
non-deterministic FMA is definable in \FO is undecidable. Thus, the following 
question is also worth to be investigated:

\begin{quote}
\em
Can we characterize, among the languages recognized by unambiguous (or even deterministic)
FMA, those recognizable by orbit-finite data monoids, or those definable in \FO ?
\end{quote}

\noindent
Finally, the problem of finding a natural logic with the same expressiveness of 
unambiguous FMA is clearly related to the possibility of proving effective closure 
of unambiguous FMA under complementation. The latter problem is also open and
challenging -- we remark that a similar closure property holds trivially for 
{\em strongly unambiguous FMA} \cite{forms-of-determinism}, namely, FMA that 
admit exactly one (accepting or rejecting) run on each data word.

\paragraph{Acknowledgements} 
We would like to thank Michael Benedikt, Anca Muscholl, and Zhilin Wu 
for the many helpful remarks on the paper.

\bibliographystyle{plain}
\bibliography{main}

\end{document}